\documentclass[12pt, letterpaper]{article}

\usepackage[utf8]{inputenc}
\usepackage{amsmath, amssymb, amsthm}
\usepackage{mathtools}
\usepackage{setspace}
\usepackage[backend=biber,style=authoryear,natbib=true,sorting=nyt,
  giveninits=true,useprefix=true,maxcitenames=2,mincitenames=1,maxbibnames=99,
  uniquelist=false,uniquename=false,dashed=false,isbn=false,
  doi=false,url=false,eprint=true]{biblatex}
\usepackage{csquotes}
\DeclareNameAlias{sortname}{family-given}
\DeclareDelimFormat{nameyeardelim}{\addspace}
\renewbibmacro*{in:}{%
  \ifentrytype{article}
    {}
    {\printtext{\bibstring{in}\intitlepunct}}}
\let\cite\textcite
\usepackage{booktabs}
\usepackage{xcolor}
\definecolor{darkgreen}{rgb}{0,0.45,0}
\usepackage{graphicx}
\usepackage{placeins}
\usepackage{needspace}
\usepackage[hidelinks,plainpages=false,pdfpagelabels=true]{hyperref}
\newbibmacro*{sauss:titlelink}[1]{%
  \iffieldundef{doi}
    {\iffieldundef{url}
      {\iffieldundef{eprint}
        {#1}
        {\href{https://arxiv.org/abs/\thefield{eprint}}{#1}}}
      {\href{\thefield{url}}{#1}}}
    {\href{https://doi.org/\thefield{doi}}{#1}}}
\DeclareFieldFormat{title}{\usebibmacro{sauss:titlelink}{\mkbibemph{#1}}}
\DeclareFieldFormat[article,inbook,incollection,inproceedings,misc,patent,thesis,unpublished]
  {title}{\usebibmacro{sauss:titlelink}{\mkbibquote{#1\isdot}}}
\DeclareFieldFormat{doi}{%
  \href{https://doi.org/#1}{\nolinkurl{doi:#1}}}
\DeclareFieldFormat{eprint}{%
  \href{https://arxiv.org/abs/#1}{\nolinkurl{arXiv:#1}}}
\DeclareFieldFormat{eprint:arxiv}{%
  \href{https://arxiv.org/abs/#1}{\nolinkurl{arXiv:#1}}}
\newcounter{saussmainpages}

\usepackage[ruled, noline]{algorithm2e}

\usepackage{array}
\usepackage{siunitx}
\newcommand{\spacingset}[1]{\renewcommand{\baselinestretch}{#1}\small\normalsize}
\newcommand{\anon}{1}

\theoremstyle{plain}
\newtheorem{theorem}{Theorem}
\newtheorem{proposition}{Proposition}
\newtheorem{corollary}{Corollary}
\newtheorem{lemma}{Lemma}
\theoremstyle{definition}
\newtheorem{example}{Example}
\newtheorem{assumption}{Assumption}
\newtheorem{definition}{Definition}
\theoremstyle{remark}
\newtheorem{remark}{Remark}

\newcommand{\DX}{\Delta X}
\newcommand{\Rsim}{R}
\newcommand{\Kmax}{K_{\max}}

\newcommand{\SAUSStitle}{SAUSS: Stochastic Approximation with Unbiased Simulated Scores for
Limited Dependent Variable Models}
\newcommand{\SAUSSauthors}{
Sokbae Lee\thanks{Department of Economics, Columbia University, New York, NY 10027, USA.}
\and
Yuan Liao\thanks{Department of Economics, University of Iowa, Iowa City, IA 52242, USA.}
\and
Myung Hwan Seo\thanks{Department of Economics, The Hong Kong University of Science and Technology, Hong Kong.}
\and
Youngki Shin\thanks{Department of Economics, McMaster University, Hamilton, ON L8S 4L8, Canada.}
}
\title{\SAUSStitle\thanks{We thank Wei Biao Wu and Yichen Zhang for their helpful and encouraging comments.}}
\author{\SAUSSauthors}
\date{November 2, 2026}

\begin{document}
\begin{refsection}

\spacingset{1}
\if1\anon
\maketitle
\else
\bigskip\bigskip\bigskip
\begin{center}
{\LARGE\bfseries \SAUSStitle\par}
\end{center}
\medskip
\fi

\begin{abstract}
Simulated maximum likelihood is widely used in latent-variable models where
likelihood evaluation requires intractable integrals or probabilities. With a
fixed number of simulation draws, its estimating equation is generally biased.
We develop Stochastic Approximation with Unbiased Simulated Scores (SAUSS), combining score simulation with stochastic gradient descent. The Hajivassiliou--McFadden identity, an instance of Fisher's identity, expresses the likelihood score as the conditional expectation of the complete-data score given the observed data. Exact conditional draws yield an unbiased simulated score for
any fixed number of draws. The estimator averages stochastic gradient iterates
computed using mini-batches of fixed size regardless of sample size, avoiding
repeated full-sample likelihood evaluations. Under local regularity and
initialization conditions, we establish asymptotic normality and a functional
central limit theorem for the partial-sum process of the iterates under
multi-pass sampling from an observed sample. These results support
plug-in and random-scaling confidence intervals that remain asymptotically valid
when algorithmic uncertainty is non-negligible relative to sampling uncertainty.
We develop multinomial probit as a leading application. In simulations
and two applications, SAUSS returns estimates in every reported run, including
cases where the simulated maximum likelihood implementation returns none.
When both return estimates, SAUSS gives comparable estimates in less time.
\end{abstract}


\newpage

\onehalfspacing

\section{Introduction}\label{sec:introduction}

In statistical modeling, likelihood evaluation is often difficult because it requires integrals or probabilities that are unavailable in closed form or costly to approximate numerically. This occurs in latent-variable, missing-data, and limited dependent variable models, among others. Simulated maximum likelihood (SML) addresses this problem by replacing the likelihood function with simulation-based approximations.

A limitation of SML is simulation bias. Let \(P_i(\theta)\) denote the likelihood function for observation \(i\), and let \(\widehat P_i(\theta)\) be an unbiased simulation-based estimator, so that \(\mathbb E_u\{\widehat P_i(\theta)\}=P_i(\theta)\), where \(\mathbb E_u\) denotes expectation over the simulation draws conditional on the data. SML maximizes the logarithm of \(\widehat P_i(\theta)\). At the true parameter value \(\theta_*\), the simulated first-order condition generally is not centered at zero:
\[
    \mathbb E\!\left[
        \mathbb E_u\!\left\{
            \nabla_\theta\log \widehat P_i(\theta_*)
        \right\}
    \right]
    \neq 0.
\]
Classical SML controls this bias by increasing the number of simulation draws with the sample size, which increases computational cost.

An alternative to SML is to simulate the likelihood score directly \citep{hajivassiliou1998method}. Let \(s_i(\theta)\coloneqq\nabla_\theta\log P_i(\theta)\) denote the likelihood score for observation \(i\), and let \(\widehat s_i(\theta)\) denote a simulated score. The key objective is to construct \(\widehat s_i(\theta)\) such that
\[
    \mathbb E_u\{\widehat s_i(\theta)\}
    =
    s_i(\theta).
\]
When this condition holds, simulation error enters as mean-zero noise around the exact likelihood score rather than as bias in the first-order condition. The number of simulation draws can remain fixed: additional draws reduce simulation variance but are not required to eliminate simulation bias.

This approach is particularly useful for stochastic gradient descent (SGD), where estimation relies directly on the score function. In many latent-variable models, however, neither the likelihood function nor its score can be evaluated directly. The central computational problem is therefore to construct an unbiased simulated score without first evaluating the intractable likelihood function. In the pioneering work in this direction,  \citet{chen2024online} develop a stochastic gradient method that approximates gradients by finite differences of function values along random directions. Their approach applies when the objective function can be evaluated, and the resulting gradient estimator is biased, with the bias explicitly characterized.

In the model considered here, the likelihood function itself is  difficult to evaluate directly.  Our approach  leverages  the Hajivassiliou--McFadden identity, an instance of Fisher's identity \citep[see, e.g.,][]{louis1982finding}: let \(f_i(y^*;\theta)\) denote the density of the latent variable \(Y_i^*\) for observation \(i\). Let \(D_i\) denote the observed data. Under regularity conditions,
\[
    \nabla_\theta\log P_i(\theta)
    =
    \mathbb E_\theta
    \left[
        \nabla_\theta\log f_i(Y_i^*;\theta)
        \mid D_i
    \right].
\]
If exact draws can be generated from the conditional distribution of \(Y_i^*\) given \(D_i\), an ordinary Monte Carlo average of the complete-data scores is therefore unbiased for the exact observed-data score for any fixed number of draws.

We apply this identity in the context of SGD by developing Stochastic Approximation with Unbiased Simulated Scores (SAUSS). At iteration \(t\), SAUSS samples a mini-batch of fixed size, generates a fixed number of exact conditional draws for each observation, and updates the parameter using the resulting simulated score. The proposed estimator averages the resulting iterates. The stochastic-approximation construction does not require individual simulated scores to be smooth in the parameter. SAUSS estimates the parameters using unbiased simulated scores while avoiding repeated full-sample likelihood evaluations.

While simulation does not introduce bias, it instead contributes an additional variance component. We develop asymptotic theory that explicitly accounts for this simulation uncertainty.
With a fixed number of exact conditional draws, we establish asymptotic normality and functional limits under local regularity and initialization conditions for multi-pass sampling from an observed sample. These results support plug-in and random-scaling inference
that remains asymptotically valid when algorithmic uncertainty is
non-negligible relative to sampling uncertainty. In multi-pass sampling, the algorithmic component becomes
negligible relative to sampling uncertainty when the number of updates
grows sufficiently fast relative to the sample size.

As a leading application, we develop SAUSS for the multinomial probit (MNP) model and derive its asymptotic theory and inference. The softmax action-selection rule in reinforcement learning has the multinomial-logit probability form \citep{suttonBarto2018}. Unlike multinomial logit, which imposes independence of irrelevant alternatives \citep{Train2009}, MNP allows a general variance matrix for Gaussian latent utilities and hence more flexible substitution patterns. This flexibility comes at a computational cost: likelihood evaluation requires multivariate Gaussian probabilities over regions whose dimension increases with the number of alternatives \citep{gewekeKeaneRunkle1994alternative,gewekeKeane2001computational}. Multi-pass experiments show near-nominal coverage even when algorithmic variation dominates sampling uncertainty. In the reported comparisons, SAUSS returns estimates in every replication, including settings in which the default procedure based on the Geweke--Hajivassiliou--Keane (GHK) simulator returns none.

The closest contemporaneous work is \citet{frazierLoaizaMayaNibbering2026scalable}, who propose a gradient-ascent method with stochastically estimated gradients. Their method uses conditional score simulation in \emph{full-sample} updates, whereas SAUSS uses score simulation based on exact accept--reject conditional draws in \emph{mini-batch} updates. Their theory assumes global strong concavity of the sample log-likelihood and that the optimization error is negligible relative to \(n^{-1/2}\) sampling error. In contrast, we develop theory under local conditions without requiring global concavity, which is relevant in MNP models, and the asymptotic variance explicitly accounts for algorithmic variation, so our confidence intervals remain valid when the number of mini-batch updates is of the same order as the sample size. Nevertheless, the two approaches share the conditional-score lineage of the method of simulated scores and both use iterate averaging.

There is also a parallel Bayesian literature that avoids direct likelihood
evaluation through latent-data augmentation
\citep{albertChib1993bayesian,mccullochRossi1994exact}.
More recent work develops variational posterior approximations with
observation subsampling \citep{loaizaMayaNibbering2023fast} and conditional
variational inference for large choice sets and sample sizes
\citep{kimKangKockBansalSohn2026scalable}. These approaches share
latent-variable computations with SAUSS. Posterior simulation differs from
our frequentist likelihood-based analysis, while variational methods
approximate distributions or estimation objectives. Other approaches use
approximate composite marginal likelihood \citep{bhat2011macml},
expectation propagation to approximate conditional moments within
expectation-maximization (EM) \citep{dingImbensQuYe2024probit}, or
neural-network emulation of choice probabilities
\citep{huchKeane2026amortized}. SAUSS instead estimates the
parameters using conditionally unbiased simulated scores, with
frequentist inference that accounts for sampling and algorithmic uncertainty.

Stochastic approximation with simulated latent data has a long history,
including Markov chain Monte Carlo methods for incomplete-data estimation
\citep{guKong1998stochastic} and stochastic-approximation EM
\citep{delyonLavielleMoulines1999convergence}.
\citet{toulisAiroldi2017asymptotic} study the asymptotic and
finite-sample properties of stochastic-gradient estimators, and
\citet{chenLeeTongZhang2020statistical,zhuChenWu2023online,wei2026online} develop variance
estimation methods for averaged SGD. \citet{liLiangChenZhang2026stream}
develop online inference for optimizing steady-state costs using
parameter-dependent Markov chains. SAUSS instead uses exact simulation of
latent variables conditional on observed outcomes to construct unbiased
simulated scores. Resampling from
an observed sample is also used to combine estimation and inference
\citep{forneron2024estimation}. Related stochastic-approximation methods provide
scalable estimation and inference for semiparametric models
\citep{chenTamerYao2026fast} and overidentified moment condition models
\citep{chenLeeLiaoSeoShinSong2025sgmm}. Following
\citet{chenKimLeeSeoSong2025slim}, we distinguish sampling uncertainty from
algorithmic variation; in SAUSS, this variation arises from both subsampling
the observed sample and simulating observed-data likelihood scores. Our theory
accommodates potentially discontinuous simulated scores under local
conditions and verifies moment and empirical regularity conditions for MNP.

The rest of the paper is organized as follows.
Section~\ref{sec:sauss-framework} presents the conditional-score
construction and the SAUSS recursion, and Section~\ref{sec:mnp-sauss}
develops its MNP implementation. Section~\ref{sec:application}
illustrates its use in two empirical applications. Sections~\ref{sec:asymptotic-theory}
and~\ref{sec:inference} establish the asymptotic theory and inference procedures
for multi-pass sampling.
Section~\ref{sec:monte-carlo} provides practical guidance on hyperparameter
choices and summarizes the Monte Carlo evidence. Section~\ref{sec:conclusion} concludes.
The Supplementary Material contains all proofs, additional
implementation and numerical details, and the extension to online sampling
with fresh observations.

\section{The SAUSS Framework}\label{sec:sauss-framework}

SAUSS uses stochastic approximation driven by conditionally unbiased
simulated scores. We first construct these scores using the Hajivassiliou--McFadden
identity for limited dependent variable models, then set out the resulting
mini-batch recursion under multi-pass sampling.

\subsection{\texorpdfstring{Hajivassiliou--McFadden identity}{Hajivassiliou--McFadden identity}}\label{sec:ldv-identity}

Let \(P_i(\theta)\) denote the probability of the observed
outcome for observation \(i\). SML replaces the
exact log-likelihood \(\mathcal L_n(\theta)=\sum_{i=1}^n\log P_i(\theta)\)
with \({ \widehat{\mathcal L}}_n(\theta)=\sum_{i=1}^n\log\widehat P_i(\theta)\),
where \(\widehat P_i(\theta)\) is a simulated probability estimator.
The GHK simulator approximates
MNP choice probabilities and is widely used in SML
estimation; its smooth simulated log-likelihood facilitates numerical maximization \citep[e.g.,][]{bolduc1999practical,Train2009}.

Even when \(\widehat P_i(\theta)\) is unbiased for \(P_i(\theta)\), the nonlinear
transformation \(\log\widehat P_i(\theta)\) is generally biased.
For a fixed finite number of simulation draws,
\begin{equation}\label{eq:sml-score-bias}
    \mathbb E_u\!\left\{
        \nabla_\theta \log \widehat P_i(\theta)
    \right\}
    \neq
    \nabla_\theta \log P_i(\theta)
    \qquad
    \text{in general},
\end{equation}
where \(\mathbb E_u\) denotes expectation over the simulation
draws conditional on the observed data, at a fixed value of \(\theta\). For the conditional simulators below, we write this
expectation as \(\mathbb E_\theta\{\cdot\mid D_i\}\).
Classical SML controls this error by increasing the number
of draws per observation with the sample size
\citep[Section~4]{gewekeKeane2001computational}.

The method of simulated
scores \citep{hajivassiliou1998method} instead simulates the likelihood
score directly.
Write the observed data as \(D_i=(Y_i,X_i)\), where
\(X_i\) collects the observed covariates and \(Y_i\) is
generated by a latent vector \(Y_i^*\).
Suppose that the observed event \(Y_i=y_i\) can be represented as
\(Y_i^*\in\mathcal D_i\), where \(\mathcal D_i=\mathcal D(y_i,X_i)\)
does not depend on \(\theta\). All probabilities and densities below are
conditional on \(X_i\). The parameter enters through the latent density
\(f_i(y^*;\theta)\), and the likelihood contribution is
\begin{equation}\label{eq:ldv-prob}
    P_i(\theta)
    =
    \int_{\mathcal{D}_i}
        f_i({y^*};\theta)
    \,d{y^*}.
\end{equation}
\begin{lemma}[Hajivassiliou--McFadden identity]\label{lem:hm-identity}
Fix \(i\), condition on \(X_i\), and fix \(\theta\) in the interior of
the parameter space.
Assume that the observed outcome \(Y_i=y_i\) is equivalent to
\(Y_i^*\in\mathcal{D}_i\) for a measurable set \(\mathcal{D}_i\) that
does not depend on \(\theta\). Assume also that \(P_i(\theta)>0\) and
that \(f_i({y^*};\theta)>0\) for almost every \({y^*}\in\mathcal{D}_i\).
Finally, suppose that there is a neighborhood \(\mathcal{N}_\theta\) of
\(\theta\) such that, for almost every \({y^*}\in\mathcal{D}_i\), the map
\(\vartheta\mapsto f_i({y^*};\vartheta)\) is differentiable on
\(\mathcal{N}_\theta\), and there exists an integrable envelope \(M_i\)
such that
\[ \sup_{\vartheta\in\mathcal{N}_\theta}
    \left\|
        \nabla_\vartheta f_i({y^*};\vartheta)
    \right\|
    \leq
    M_i({y^*}),
    \qquad
    \int_{\mathcal{D}_i} M_i({y^*})\,d{y^*}<\infty .
\]
Then
\begin{equation}\label{eq:ldv-score-identity}
    g_i(\theta)
    \coloneqq
    -\nabla_\theta \log P_i(\theta)
    =
    \mathbb{E}_\theta
    \left[
        -\nabla_\theta \log f_i(Y_i^*;\theta)
        \mid
        Y_i^* \in \mathcal{D}_i
    \right].
\end{equation}
\end{lemma}

The Hajivassiliou--McFadden identity is an instance of
Fisher's identity.
The observed-data likelihood \(P_i(\theta)\) is costly to evaluate because it requires integrating the latent density \(f_i(y^*;\theta)\) over \(\mathcal D_i\), whereas the complete-data score \(\nabla_\theta\log f_i(y^*;\theta)\) is easy to evaluate at any given value of \(y^*\), even though \(Y_i^*\) is unobserved. The Hajivassiliou--McFadden identity yields a conditionally unbiased Monte Carlo estimator
of the negative observed-data score: we can simulate from the conditional  distribution of \(Y_i^*\) given
\(Y_i^* \in \mathcal{D}_i\). Let
\(Y_{ir}^*(\theta)\), \(r=1,\ldots,\Rsim\), be independent exact draws from this distribution, and define the corresponding negative complete-data scores
\(g_{ir}(\theta)\coloneqq
-\left.\nabla_\theta\log f_i({y^*};\theta)\right|_{{y^*}=Y_{ir}^*(\theta)}\).
Their average is an unbiased estimator of \(g_i(\theta)\) for every
fixed \(\Rsim\geq1\):
\begin{equation}\label{eq:one-observation-direction}
    \widehat g_{\Rsim}(\theta;D_i)
    \coloneqq
    \frac{1}{\Rsim}\sum_{r=1}^{\Rsim}g_{ir}(\theta)
\end{equation}
and ${ \mathbb E_\theta}\{\widehat g_{\Rsim}(\theta;D_i)\mid D_i\}
    =g_i(\theta)$.
Binary probit needs no simulation, but it shows the construction most clearly.

\begin{example}[Binary probit as a toy example]\label{ex:binary-probit}
For \(D_i=(Y_i,\DX_i)\), let
\(Y_i^*=\DX_i\theta+\varepsilon_i\), where
\(\varepsilon_i\sim\mathcal N(0,1)\) is independent of \(\DX_i\), and
\(Y_i=\mathbb I\{Y_i^*>0\}\). The latent region is
\(\mathcal D_i=(0,\infty)\) when \(y_i=1\) and
\(\mathcal D_i=(-\infty,0]\) when \(y_i=0\).
Let $\Phi(\cdot)$ denote the standard normal CDF. The Hajivassiliou--McFadden identity becomes
\[
    g_i(\theta)
=
\mathbb E_\theta
\left[
-\Delta X_i'(Y_i^*-\Delta X_i\theta)
\mid
Y_i^*\in\mathcal D_i
\right],
\]
where
   $ g_i(\theta)
    =-\nabla_\theta
\log \Phi\!\left\{(2y_i-1)\Delta X_i\theta\right\} $  is the  negative likelihood score of the model.
Supplementary Appendix~\ref{sec:supp-binary-probit} verifies this identity,
gives the simulation variance in closed form, and quantifies
its contribution to algorithmic variance.
We simulate draws \(Y_{ir}^*(\theta)\) from
\(\mathcal N(\DX_i\theta,1)\) truncated to \(\mathcal D_i\), which  can be achieved using a simple    accept--reject simulator:   draw a   Gaussian
proposal \(\varepsilon_i\sim\mathcal N(0,1)\) and accept  it if
\(\DX_i\theta+\varepsilon_i\in\mathcal D_i\). Then let $Y_{ir}^*(\theta) = \DX_i\theta+\varepsilon_i$ if accepted.
This gives:
\[
    \widehat g_{\Rsim}(\theta;D_i)
    =
    -\frac{1}{\Rsim}\sum_{r=1}^{\Rsim}
        \DX_i'\{Y_{ir}^*(\theta)-\DX_i\theta\}.
\]

The form of the simulated score provides  an intuitive statistical insight for this approach: note that  $\widehat g_{\Rsim}(\theta;D_i)$ takes the familiar form of the OLS gradient, with the unobserved latent outcome $Y_i^*$ replaced by its simulated counterpart $Y_{ir}^*(\theta)$.  So the simulator  simply replaces latent outcome  with a simulated latent outcome  that is compatible with the observed outcome: $Y_{ir}^*(\theta)>0$ when $Y_i=1$ and $Y_{ir}^*(\theta)\leq0$ when $Y_i=0$.

\end{example}

In \citet{hajivassiliou1998method}, the simulated score function
enters a full-sample Z-estimator and the estimator is found as the root of the equation. This root-finding problem can be difficult
when the simulated contributions are discontinuous in \(\theta\), as often
occurs with simulators based on  accept--reject constructions.

\subsection{\texorpdfstring{SAUSS as stochastic approximation}{SAUSS as stochastic approximation}}\label{sec:sgd}

SAUSS applies the same construction in general: we use the complete-data likelihood score $\nabla_\theta\log f_i(Y_i^*;\theta)$, replacing the unobserved latent variable $Y_i^*$ with simulated draws that are consistent with the observed data. Averaging the resulting complete-data scores yields an unbiased estimator of the observed-data likelihood score, whose negative is an unbiased stochastic gradient of the negative log-likelihood. Our theory allows individual simulated scores to be discontinuous in $\theta$ and imposes local conditions on their means, variances, and moments.

Let \(D_{1:n}\) be the observed sample, reused throughout multi-pass
estimation, and let \(\theta_0=\theta_{\mathrm{init}}\) be a preliminary
estimator. Let \(\mathcal F_{t-1}\) denote the information generated by the
data set, preliminary estimator, and previous batch selections and simulation draws.
For a fixed finite \(m\geq1\), conditional on \(\mathcal F_{t-1}\),
the batch observations \(D_{tb}^{\mathrm{update}}\), \(b=1,\ldots,m\), are
independent draws from the empirical distribution of \(D_{1:n}\). For each
selected observation, \(\widehat g_{\Rsim}(\theta_{t-1};D_{tb}^{\mathrm{update}})\)
is the estimator in \eqref{eq:one-observation-direction}, based on \(\Rsim\)
independent exact conditional draws at \(\theta_{t-1}\). These estimators
are independent conditional on the batch and \(\mathcal F_{t-1}\). The
mini-batch average of the simulated negative scores is
\begin{equation}\label{eq:sa-batch-direction}
    \widehat G_t(\theta)
    \coloneqq
    \frac1m\sum_{b=1}^m
        \widehat g_{\Rsim}(\theta;D_{tb}^{\mathrm{update}}).
\end{equation}
The recursion is
\begin{equation}\label{eq:sa-update}
    \theta_t
    =
    \theta_{t-1}
    -\gamma_t\widehat G_t(\theta_{t-1}),
\end{equation}
where \(\gamma_t\) is a decreasing learning rate. The proposed estimator is the Polyak--Ruppert average
\begin{equation}\label{eq:pr-average}
    \bar\theta_T
    \coloneqq
    \frac{1}{T-T_0}
    \sum_{t=T_0+1}^{T}
        \theta_t,
\end{equation}
where \(T\) is the number of updates and \(T_0\), \(0\leq T_0<T\),
specifies the start of averaging.

Define the population and empirical mean negative scores by
\(\mathcal R(\theta)=\mathbb E\{g_i(\theta)\}\) and
\(\mathcal R_n(\theta)=n^{-1}\sum_{i=1}^n g_i(\theta)
=-n^{-1}\nabla_\theta\mathcal L_n(\theta)\).
The conditional mean of the simulated batch score is
\[
    \mathbb E\{\widehat G_t(\theta_{t-1})\mid\mathcal F_{t-1}\}
    =\mathcal R_n(\theta_{t-1}).
\]
Such conditional unbiasedness of the simulated
scores holds for every fixed \(\Rsim\geq1\).
Algorithm~\ref{alg:sauss-general} states the general SAUSS recursion.

\begin{algorithm}[!htb]
\spacingset{1}
\DontPrintSemicolon
\caption{SAUSS}\label{alg:sauss-general}

Initialize at $\theta_0=\theta_{\mathrm{init}}$ and $\bar\theta_0=0$. For $t=1,\ldots,T$, given the previous iterate $\theta_{t-1}$, do:

 \textbf{Step 1:} Draw $D_{t1}^{\mathrm{update}},\ldots,D_{tm}^{\mathrm{update}}$
    independently with replacement from  $D_{1:n}$.

 \textbf{Step 2:} For each $b=1,\ldots,m$, let $Y_{tb}^*$, $\mathcal D_{tb}$, and $f_{tb}$ denote the
        latent vector, observed-outcome region, and latent density of $D_{tb}^{\mathrm{update}}$.
        Simulate $\Rsim$ independent draws
        $Y_{tbr}^*(\theta_{t-1})$, $r=1,\ldots,\Rsim$, from the conditional distribution
        \[
        Y_{tb}^* \mid
        \{Y_{tb}^*\in\mathcal D_{tb}\}
        \]
        at $\theta_{t-1}$, and compute the simulated negative score
        \[
        \widehat g_{\Rsim}(\theta_{t-1};D_{tb}^{\mathrm{update}})
        =-\frac{1}{\Rsim}\sum_{r=1}^{\Rsim}
        \left.\nabla_\theta\log f_{tb}(y^*;\theta)
        \right|_{\substack{\theta=\theta_{t-1}\\
        y^*=Y_{tbr}^*(\theta_{t-1})}}.
        \]

     \textbf{Step 3:} Compute the mini-batch average of the simulated negative scores as in \eqref{eq:sa-batch-direction}:
    \[
    \widehat G_t(\theta_{t-1})
=
    \frac{1}{m}\sum_{b=1}^m
    \widehat g_{\Rsim}(\theta_{t-1};D_{tb}^{\mathrm{update}})
    \]
   and update \[
    \theta_t
=
    \theta_{t-1}
    -\gamma_t\widehat G_t(\theta_{t-1}),
    \qquad
    \bar\theta_t
    =\bar\theta_{t-1}
    +\bigl(\theta_t\mathbb I\{t>T_0\}-\bar\theta_{t-1}\bigr)/\max\{t-T_0,1\}.
    \]
The final estimator is $\bar\theta_T$, which equals \eqref{eq:pr-average}.

\end{algorithm}

\section{SAUSS for Multinomial Probit}\label{sec:mnp-sauss}

This section applies SAUSS to multinomial probit models.
Consider \(n\) agents choosing among \(J\) alternatives. The latent
utility for agent \(i\) and alternative \(j\) is
\[
    U_{ij} = x_{ij}' \beta + \varepsilon_{ij},
    \quad j=1,\ldots,J,
\]
where \(x_{ij} \in \mathbb{R}^K\), \(\beta \in \mathbb{R}^K\), and
\(\varepsilon_i\coloneqq (\varepsilon_{i1},\ldots,\varepsilon_{iJ})'\)
satisfies \(\varepsilon_i\sim\mathcal{N}(0,\Sigma)\). Agent \(i\)
chooses \(y_i=j\) if \(U_{ij}>U_{ik}\) for all \(k \neq j\).

Utility levels are not identified, so we work in differences
relative to alternative 1. Let
\(W_i\coloneqq\Delta U_i=(U_{i2}-U_{i1},\ldots,U_{iJ}-U_{i1})'\),
and \(\nu_i=(\nu_{i2},\ldots,\nu_{iJ})'\) with
\(\nu_{ij}\coloneqq\varepsilon_{ij}-\varepsilon_{i1}\). Let \(\DX_i\) be the
\(d\times K\) matrix whose \((j-1)\)th row is
\((x_{ij}-x_{i1})'\), where \(d=J-1\).  Then
\[
    W_i = \DX_i\beta+\nu_i,
    \qquad
    \nu_i \sim \mathcal{N}(0,\Omega).
\]
The variance matrix of differenced errors satisfies
\[
    \Omega_{j-1,k-1}
    =\Sigma_{jk}+\Sigma_{11}-\Sigma_{j1}-\Sigma_{1k},
    \qquad j,k=2,\ldots,J.
\]
For \(w=(w_2,\ldots,w_J)'\in\mathbb R^d\), set
\(w_1=0\). All choice regions can then be written as
\[
    \mathcal D_j
    =\{w\in\mathbb R^d:w_j>w_k\text{ for all }k\in\{1,\ldots,J\}\setminus\{j\}\},
    \qquad j=1,\ldots,J.
\]
Choosing alternative 1 means that all non-base utility
differences are negative; choosing a non-base alternative means that its
utility difference is positive and exceeds all other non-base differences.
For observation \(i\), set
\(\mathcal D_i=\mathcal D_{y_i}\) and \(D_i=(y_i,\DX_i)\).
The observed-choice event \(W_i\in\mathcal D_i\) is the fixed-region
representation used in Lemma~\ref{lem:hm-identity}.

The model is identified only up to scale:
\((\beta,\Omega)\mapsto(c\beta,c^2\Omega)\), \(c>0\), leaves all choice
probabilities unchanged. So we use a Cholesky normalization of the variance matrix: for $\lambda_k=\log L_{kk}$,
\begin{equation}\label{eq:covariance-cholesky}
    \Omega=LL',
    \qquad
    L\ \text{lower triangular},
    \qquad
    L_{11}=1,
    \qquad
    L_{kk}=e^{\lambda_k},\quad k=2,\ldots,d.
\end{equation}
Thus \(L_{11}=1\) fixes the variance of \(\nu_{i2}\) at one and thereby
fixes the scale; the remaining diagonal elements are positive by
construction, and the strictly lower-triangular elements are unrestricted.
Let \(\theta_L\) stack the strictly lower-triangular entries
of \(L\) and \(\lambda_2,\ldots,\lambda_d\), excluding the fixed element
\(L_{11}\). The parameter of interest is \(\theta=(\beta',\theta_L')'\).
The probability of choice \(j\) is
\[
    P_{ij}(\theta)
    \coloneqq
    \int_{\mathcal D_j}
        \phi_d(w;\DX_i\beta,\Omega)
    \,dw,
\]
where \(\phi_d(\cdot;\mu,\Omega)\) denotes the \(d\)-variate normal
density.

The MNP log-likelihood is
\(\mathcal L_n(\theta)=\sum_{i=1}^n\log P_{iy_i}(\theta)\).
 For one observed choice, Lemma~\ref{lem:hm-identity} gives
\[
    g_i(\theta)
    =-\nabla_\theta\log P_{iy_i}(\theta)
    =\mathbb E_\theta
    \left[
        -\nabla_\theta\log\phi_d(W_i;\DX_i\beta,\Omega)
        \mid W_i\in\mathcal D_i
    \right].
\]

Write \(\nu\coloneqq w-\DX_i\beta\) for the residual at a complete-data value \(w\).
The negative complete-data score contributions for \(\beta\) and \(\Omega\) are, respectively,
\[
\begin{aligned}
   -\nabla_\beta\log\phi_d(w;\DX_i\beta,\Omega)&=-\DX_i'\Omega^{-1}\nu,\\
   -\nabla_\Omega\log\phi_d(w;\DX_i\beta,\Omega)&=-\frac12\Omega^{-1}(\nu\nu'-\Omega)\Omega^{-1}.
\end{aligned}
\]
The chain rule through \(\Omega=LL'\) gives the negative score matrix
\(-\Omega^{-1}(\nu\nu'-\Omega)\Omega^{-1}L\) for \(L\).
For any \(d\times d\) matrix \(A\), write
\(\mathcal{C}_{\theta_L}(A;L)\) for the vector that keeps the free
lower-triangular entries of \(A\), replaces each   diagonal entry
\(A_{kk}\), \(k=2,\ldots,d\), by \(A_{kk}L_{kk}\), and drops the fixed coordinate
\((1,1)\).
The resulting negative complete-data score for a residual \(\nu\) is
\begin{equation}\label{eq:mnp-one-draw-gradient}
    \Gamma_i(\theta,\nu)
    \coloneqq
    -
    \begin{pmatrix}
        \DX_i'\Omega^{-1}\nu\\[3pt]
        \mathcal C_{\theta_L}
        \{\Omega^{-1}(\nu\nu'-\Omega)\Omega^{-1}L;L\}
    \end{pmatrix}.
\end{equation}

We use the   accept--reject simulator, which draws a jointly Gaussian
proposal \(\nu_i\sim\mathcal N(0,\Omega)\) and accepts it if
\(\DX_i\beta+\nu_i\in\mathcal D_i\). Thus an accepted residual
\(\nu_i^\star(\theta)\) has the exact conditional law
\(\nu_i\mid\{\DX_i\beta+\nu_i\in\mathcal D_i\}\).
For independent accepted draws \(\nu_{ir}^\star(\theta)\), set
\(g_{ir}(\theta)=\Gamma_i\{\theta,\nu_{ir}^\star(\theta)\}\).
Lemma~\ref{lem:hm-identity} gives the simulated negative score
\begin{equation}\label{eq:mnp-one-observation-gradient}
    \widehat g_{\Rsim}(\theta;D_i)
    \coloneqq
    \frac{1}{\Rsim}
    \sum_{r=1}^{\Rsim}
        g_{ir}(\theta),
    \qquad
    \mathbb E_\theta
    \{\widehat g_{\Rsim}(\theta;D_i)\mid D_i\}
    =g_i(\theta).
\end{equation}
When \(J=2\), identify the binary outcome in
Example~\ref{ex:binary-probit} as \(Y_i=\mathbb I\{y_i=2\}\);
the normalization gives \(\Omega=1\) and leaves no free error variance parameter
\(\theta_L\). Thus \(\Gamma_i(\beta,\nu)=-\DX_i'\nu\), and
\(\widehat g_{\Rsim}(\beta;D_i)\) reduces to the simulated score in that example.

\begin{algorithm}[!htb]
\spacingset{1} % Keep the 12-point pseudocode readable within one float page.
\DontPrintSemicolon
\caption{SAUSS for Multinomial Probit}\label{alg:sauss-mnp}
\KwIn{Observed sample \(D_{1:n}\); pilot \(\theta_{\mathrm{init}}\); \(T\geq1\), \(T_0\in\{0,\ldots,T-1\}\), fixed integers \(m,\Rsim\geq1\), \(\gamma_0>0\), \(a\in(1/2,1)\)}
\KwOut{Averaged parameter \(\bar\theta_T\), reported as \({ \widehat\beta}\) and \({ \widehat\Omega}={ \widehat L}{ \widehat L}'\)}
Set \(\theta_0\leftarrow\theta_{\mathrm{init}}\), \(N_{\mathrm{avg}}\leftarrow 0\), and \(\bar\theta\leftarrow 0\)\;
\For{\(t=1,\ldots,T\)}{
    Draw \(D_{t1}^{\mathrm{update}},\ldots,D_{tm}^{\mathrm{update}}\)
    independently with replacement from
    \(D_{1:n}\)\;
    Construct \(\beta,L\) from \(\theta_{t-1}\), with \(L_{11}=1\),
    and set \(\Omega\leftarrow LL'\), \(\gamma_t\leftarrow\gamma_0t^{-a}\)\;
    \For{\(b=1,\ldots,m\)}{
        With \((y_{tb},\DX_{tb})=D_{tb}^{\mathrm{update}}\), draw \(\Rsim\) independent residuals
        \(\nu_{tbr}^\star\) from \(\mathcal N(0,\Omega)\) truncated to
        \(\{\nu:\DX_{tb}\beta+\nu\in\mathcal D_{y_{tb}}\}\), using the hybrid sampler of Remark~\ref{rem:hybrid}
        (accept--reject, completed by tilting if needed)\;
        Compute \(\widehat g_{\Rsim}(\theta_{t-1};D_{tb}^{\mathrm{update}})\)
        from \eqref{eq:mnp-one-observation-gradient}\;
    }
    \(\widehat G_t\leftarrow
    m^{-1}\sum_{b=1}^m
    \widehat g_{\Rsim}(\theta_{t-1};D_{tb}^{\mathrm{update}})\)\;
    \(\theta_t\leftarrow \theta_{t-1}-\gamma_t\widehat G_t\)\;
    \If{\(t>T_0\)}{
        \(N_{\mathrm{avg}}\leftarrow N_{\mathrm{avg}}+1\)\;
        \(\bar\theta\leftarrow\bar\theta+(\theta_t-\bar\theta)/N_{\mathrm{avg}}\)\;
    }
}
Set \(\bar\theta_T\leftarrow\bar\theta\)\;
Construct \({ \widehat\beta}\) and \({ \widehat L}\) from \(\bar\theta_T\), and set \({ \widehat\Omega}\leftarrow{ \widehat L}{ \widehat L}'\)\;
\end{algorithm}

Algorithm~\ref{alg:sauss-mnp} implements the recursion in
Section~\ref{sec:sgd} for any fixed finite mini-batch size \(m\geq1\).
Each evaluation uses new draws. Under
multi-pass sampling, each sample observation is selected \(mT/n\) times
on average. The reported average uses
\(\theta_{T_0+1},\ldots,\theta_T\), and the asymptotic results require
\(T_0/T\to0\). The extension to fixed block-specific preconditioning is
stated after Theorem~\ref{thm:multipass}.

\begin{remark}[Capped computation]\label{rem:capped}
In practice, we impose a large cap \(\Kmax\) on the number of proposal
draws. If the cap is reached after at least one acceptance but before
\(\Rsim\) acceptances, averaging the available accepted draws remains
centered at the exact negative score, but uses a random, smaller
simulation size and has a different conditional variance. If no draw
is accepted, returning a zero contribution may bias the capped simulated score. We therefore monitor the number of proposals per accepted draw and the
frequencies of partial-cap and no-acceptance events.
\end{remark}

\begin{remark}[Exact hybrid conditional sampling]\label{rem:hybrid}
The capped simulator of Remark~\ref{rem:capped} can be replaced by an exact
\emph{hybrid} sampler in the conditional-draw step of
Algorithm~\ref{alg:sauss-mnp}. For
each observation, ordinary accept--reject runs until \(\Rsim\) draws are
accepted or \(K_0\) proposals have been used; any missing draws are then
generated, independently, by the minimax exponential tilting sampler of
\citet{botev2017normal}, which draws exactly from the same truncated Gaussian
law. Because switching depends only on proposal and acceptance counts, the
completed draws have the same joint law as \(\Rsim\) draws from accept--reject
without a cap, so the results of Sections~\ref{sec:asymptotic-theory}
and~\ref{sec:inference} apply without change. The tilting stage stops with
probability one, but its run time is random. Supplementary
Appendix~\ref{sec:supp-hybrid} gives the construction, the proof, and
practical safeguards. In our simulations, the hybrid sampler is the fastest and never fails
(Section~\ref{sec:monte-carlo-results}).
All numerical results in this paper use the hybrid sampler with \(K_0=8\Rsim J\), except the
comparison of samplers in Supplementary Section~\ref{sec:mc-simulator}.
\end{remark}

\section{\texorpdfstring{Empirical Applications}{Empirical Applications}}

\label{sec:application}

We apply SAUSS to two multinomial probit models: labor-supply choices among \(J=3\) alternatives and brand
choices among \(J=5\) alternatives in scanner-panel data. In both, SAUSS produces reliable estimates and standard
errors. Its estimates are similar to those of GHK simulated maximum likelihood (GHK SML), at a fraction of the
computation time.

The implementation follows Section~\ref{sec:numerical-settings}. Mini-batches contain \(m=10\) observations sampled
with replacement, and each observation receives \(\Rsim=5\) exact conditional draws from the hybrid sampler. Before
estimation, regressors are divided by their standard deviations (in the brand-choice application, only the price),
and estimates are reported on the original scale. A short pilot run from \(\beta=0\) and \(\Omega=I\) with
\((\gamma_\beta,\gamma_L)=(0.1,0.2)\) gives the starting value, the learning rate by the curvature rule with \(c=0.5\)
and \(\gamma_L=2\gamma_\beta\), and the run length \(T\) by the variance-based rule. The main run restarts at \(t=1\)
with fresh draws and averages its \(T\) iterates: \(T=187{,}374\) updates (about \(652\) passes) for labor supply and
\(T=4{,}299{,}792\) (about \(3{,}137\) passes) for brand choice. Standard errors use the plug-in variance
\(\widehat C_{n,T}\) of Section~\ref{sec:mnp-covariance}, which accounts for sampling and algorithmic variation.
Reported times cover the whole procedure, from the pilot to the standard errors. Supplementary
Appendix~\ref{app:supp-applications} gives further details.

\subsection{\texorpdfstring{Maternal labor-supply intentions}{Maternal labor-supply intentions}}\label{sec:app-bgr}

We revisit the maternal labor-supply model of \citet{boneva2026}, which asks whether respondents' beliefs about
the consequences of maternal employment predict the labor-supply choice that they or their partner would make.
Column~2 of its Table~4 fits a multinomial probit to the choices of \(n=2{,}873\) respondents among not working,
part-time, and full-time work, in the scenario where full-time daycare is available. Utility is
linear in five alternative-specific belief and norm measures and six case-specific controls. The original estimates
use GHK SML with \(2{,}000\)
draws. We compute GHK SML with our \texttt{Rcpp} port of the
\href{https://CRAN.R-project.org/package=mlogit}{\texttt{mlogit}} estimator \citep{croissant2020mlogit}. The SAUSS
pilot makes \(50\) passes through the sample.

\begin{table}[!htb]
\centering
\caption{Logit and multinomial probit estimates, maternal labor-supply intentions}
\label{tab:bgr}
\begingroup
\spacingset{1}
\begin{tabular*}{\textwidth}{@{\extracolsep{\fill}}lrrrrr@{}}
\toprule
\multicolumn{6}{@{}l}{\emph{Panel A: Regression coefficients}} \\
Method & \shortstack{Child\\skills} & \shortstack{Family\\outcomes}
 & \shortstack{Maternal earnings\\(1,000 EUR)}
 & \shortstack{Family's\\opinion} & \shortstack{Friends'\\opinion} \\
\midrule
Logit & $0.6061\phantom{)}$ & $1.2351\phantom{)}$ & $0.0041\phantom{)}$ & $0.2326\phantom{)}$ & $0.2122\phantom{)}$ \\
 & $(0.1617)$ & $(0.1463)$ & $(0.0017)$ & $(0.0255)$ & $(0.0254)$ \\
GHK SML & $0.4553\phantom{)}$ & $1.1614\phantom{)}$ & $0.0033\phantom{)}$ & $0.2689\phantom{)}$ & $0.2616\phantom{)}$ \\
 & $(0.1684)$ & $(0.1876)$ & $(0.0014)$ & $(0.0365)$ & $(0.0388)$ \\
SAUSS & $0.4673\phantom{)}$ & $1.2376\phantom{)}$ & $0.0038\phantom{)}$ & $0.2814\phantom{)}$ & $0.2817\phantom{)}$ \\
 & $(0.1795)$ & $(0.1992)$ & $(0.0016)$ & $(0.0384)$ & $(0.0408)$ \\
\bottomrule
\end{tabular*}
\begin{tabular*}{\textwidth}{@{\extracolsep{\fill}}lrrrc@{}}
\multicolumn{5}{@{}l}{\emph{Panel B: Error variance and fit}} \\
Method & $\Omega_{21}$ & $\Omega_{22}$ & Log-likelihood
 & Time (s) \\
\midrule
Logit & -- & -- & $-2{,}214.5695$ & \phantom{0}0.1 \\
GHK SML & $-0.0022\phantom{)}$ & $0.4494\phantom{)}$ & $-2{,}189.9321$ & 34.4 \\
 & $(0.0946)$ & $(0.1381)$ &  &  \\
SAUSS & $-0.0423\phantom{)}$ & $0.5434\phantom{)}$ & $-2{,}190.9725$ & \phantom{0}5.9 \\
 & $(0.1003)$ & $(0.1732)$ &  &  \\
\bottomrule
\end{tabular*}
\par
\endgroup
\begin{minipage}{\textwidth}
\vspace{1ex}
\spacingset{1}
\footnotesize
\emph{Notes:} Daycare scenario of \citet{boneva2026}, Table~4, column~2; \(n=2{,}873\), \(J=3\). Standard
errors in parentheses; for GHK SML, the outer-product-of-scores estimator, as in \texttt{mlogit}.
For comparison, the reported log-likelihoods are evaluated by
bivariate-normal quadrature at the SAUSS and GHK SML parameter estimates.
Logit coefficients and standard errors are divided by \(\pi/\sqrt3\).
Computation times refer to a single estimation run on one core of an AMD EPYC 9655 processor.
\end{minipage}
\end{table}

Table~\ref{tab:bgr} reports the five coefficients emphasized in the original paper and the error variances. The
SAUSS and GHK SML estimates are similar, and their log-likelihoods differ by
\(1.04\). SAUSS takes \(5.9\) seconds and GHK SML \(34.4\) seconds on the same core. The multinomial logit gives different
estimates even after rescaling to the probit normalization: its coefficient on child skills is \(30\%\) to \(33\%\)
larger than the probit estimates, and its coefficients on the family's and friends' opinions are \(13\%\) to
\(25\%\) smaller.

\begin{figure}[!htb]
\centering
\caption{Estimation paths, maternal labor-supply intentions}
\label{fig:paths-bgr}
\includegraphics[width=\textwidth]{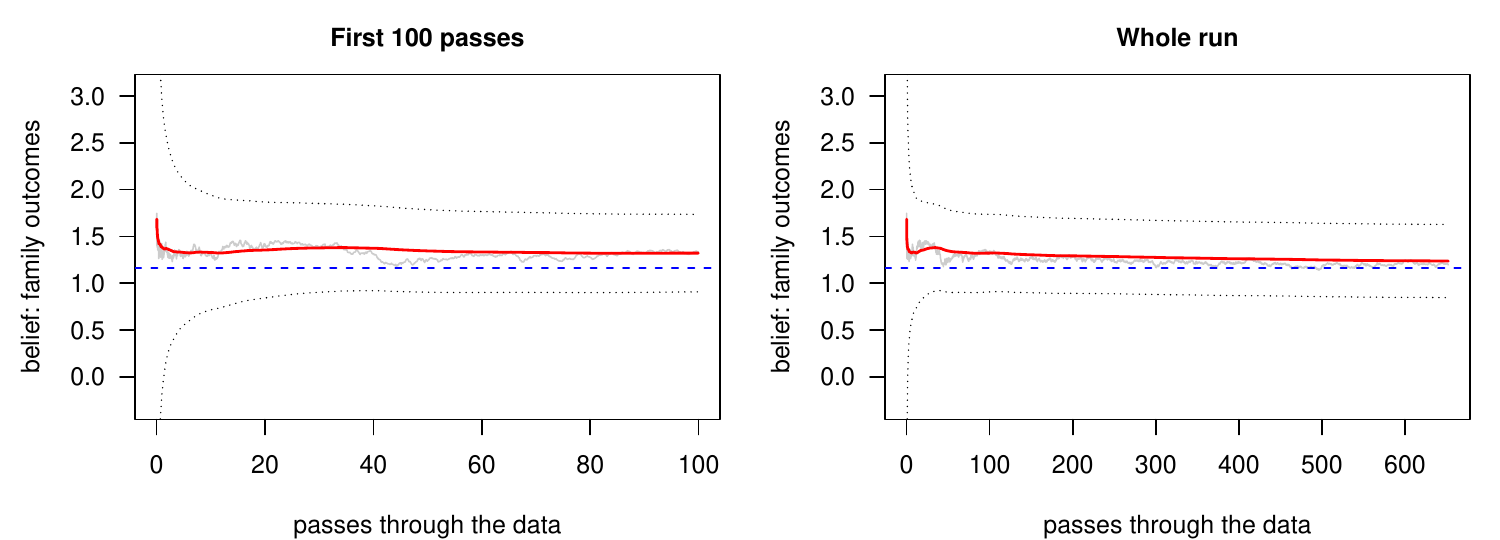}
\par
\begin{minipage}{\textwidth}
\vspace{1ex}
\spacingset{1}
\footnotesize
\emph{Notes:} Coefficient on the belief about family outcomes from one SAUSS run. Left: the first
\(100\) passes through the data; right:
the full run of approximately \(652\) passes.
Grey:
the iterate; red: the SAUSS estimate after \(t\) updates; dotted: pointwise \(95\%\) plug-in confidence
bands around the SAUSS estimate, accounting for sampling and algorithmic variation; dashed: the
GHK SML estimate.
\end{minipage}
\end{figure}

Figure~\ref{fig:paths-bgr} shows the estimation path of the family-outcomes coefficient. The iterates fluctuate
widely at first, the SAUSS estimate settles near the GHK SML estimate, and the confidence band narrows as the
algorithmic variation declines. At the end of the run, the algorithmic components account for
\(0.78\%\) of the summed coefficient variances, so sampling uncertainty dominates.

\subsection{\texorpdfstring{Canned tuna brand choice}{Canned tuna brand choice}}\label{sec:app-tuna}

We next estimate a brand-choice model on the canned-tuna scanner-panel data of
\citet{kimBlattbergRossi1995}, which record \(13{,}705\) purchase occasions by \(3{,}093\) households among \(J=5\)
brands: Star-Kist and Chicken of the Sea, each packed in water and in oil, and a private label packed in water. The
shelf prices of all five brands are observed at every occasion. We fit a multinomial probit in utility differences
relative to Star-Kist in water, with four brand intercepts, a common price coefficient, and an unrestricted
\(4\times4\) variance matrix \(\Omega=LL'\) of the differenced
errors, with \(L_{11}=1\) as in Section~\ref{sec:mnp-sauss}, so the model has \(14\) parameters.
For simplicity, we treat purchase occasions as independent and identically distributed,
abstracting from within-household dependence.

The SAUSS pilot makes \(40\) passes through the sample. For comparison, we also fit a
multinomial logit with the same four intercepts and price. \citet{kimBlattbergRossi1995} use this logit with
household-specific coefficients; we keep the coefficients common to all households, as in the probit.

\begin{table}[!htb]
\centering

\caption{Logit and multinomial probit estimates, canned-tuna brand choice}
\label{tab:tuna}
\begingroup
\spacingset{1}
\begin{tabular*}{\textwidth}{@{\extracolsep{\fill}}lrrrrr@{}}
\toprule
\multicolumn{6}{@{}l}{\emph{Panel A: Brand intercepts and price coefficient}} \\
Method & \shortstack{Chicken of\\the Sea, water} & \shortstack{Star-Kist,\\oil}
 & \shortstack{Chicken of\\the Sea, oil} & \shortstack{Private label,\\water}
 & \shortstack{Price\\(USD)} \\
\midrule
Logit & $-0.489\phantom{)}$ & $-0.498\phantom{)}$ & $-0.836\phantom{)}$ & $-1.280\phantom{)}$ & $-3.505\phantom{)}$ \\
 & $(0.015)$ & $(0.014)$ & $(0.017)$ & $(0.020)$ & $(0.049)$ \\
GHK SML & $-0.471\phantom{)}$ & $-0.693\phantom{)}$ & $-0.544\phantom{)}$ & $-1.373\phantom{)}$ & $-2.855\phantom{)}$ \\
 & $(0.030)$ & $(0.059)$ & $(0.032)$ & $(0.091)$ & $(0.111)$ \\
SAUSS & $-0.463\phantom{)}$ & $-0.669\phantom{)}$ & $-0.528\phantom{)}$ & $-1.307\phantom{)}$ & $-2.784\phantom{)}$ \\
 & $(0.029)$ & $(0.058)$ & $(0.032)$ & $(0.088)$ & $(0.115)$ \\
\bottomrule
\end{tabular*}
\par
\begin{tabular*}{\textwidth}{@{\extracolsep{\fill}}lrrrrc@{}}
\multicolumn{6}{@{}l}{\emph{Panel B: Error variance and fit}} \\
Method & $\Omega_{22}$ & $\Omega_{33}$ & $\Omega_{44}$ & Log-likelihood
 & Time (s) \\
\midrule
Logit & -- & -- & -- & $-16{,}278.219$ & $\phantom{00}0.7$ \\
GHK SML & $1.448\phantom{)}$ & $0.617\phantom{)}$ & $1.306\phantom{)}$ & $-16{,}273.213$ & $424.3$ \\
 & $(0.143)$ & $(0.080)$ & $(0.180)$ &  &  \\
SAUSS & $1.378\phantom{)}$ & $0.577\phantom{)}$ & $1.201\phantom{)}$ & $-16{,}273.729$ & $209.0$ \\
 & $(0.136)$ & $(0.078)$ & $(0.168)$ &  &  \\
\bottomrule
\end{tabular*}
\par
\endgroup
\begin{minipage}{\textwidth}
\vspace{1ex}
\spacingset{1}
\footnotesize
\emph{Notes:} \(n=13{,}705\) purchase occasions, \(J=5\); base alternative Star-Kist in water.
Other entries as in Table~\ref{tab:bgr}, except that probit log-likelihoods are evaluated at the GHK SML and SAUSS
estimates with one common set of \(1{,}000\) GHK draws per observation.
\end{minipage}
\end{table}

Table~\ref{tab:tuna} reports the estimates. SAUSS and GHK SML again give similar estimates, and the
log-likelihoods differ by \(0.52\). SAUSS takes \(209.0\) seconds, about
half the \(424.3\) seconds of GHK SML. The logit gives a different price response. On the probit scale its price
coefficient is \(-3.51\), \(23\%\) to \(26\%\) larger in magnitude than the probit estimates. The usual logit
specification assumes i.i.d.\ extreme-value utility errors and implies independence of irrelevant alternatives. The
probit variance estimates in Table~\ref{tab:tuna} depart from the equal variances of utility differences implied by
this error specification.

\section{Asymptotic Theory}\label{sec:asymptotic-theory}

We develop asymptotic theory for SAUSS based on \(n\)
independent and identically distributed observations.
Throughout, the mini-batch size \(m\), the
number \(\Rsim\) of conditional draws, the learning-rate parameters
\(\gamma_0\) and \(a\), and the parameter dimension \(p\) are fixed,
while the sample size \(n\), the number of updates \(T\), and the pilot
sample size \(n_0\) tend to infinity with no additional relative-rate restriction.
For the infinite-horizon exit events below, the recursion is continued beyond the reporting horizon \(T\); this does not alter the reported estimator.
Let \(\Theta\subseteq\mathbb R^p\) denote the
parameter space.

Let \(\widehat\theta_n\) denote a local solution of the full-sample likelihood
score equation \(\mathcal R_n(\widehat\theta_n)=0\).  The total error of \(\bar\theta_T\) combines sampling variation in
\(\widehat\theta_n\) and algorithmic variation of \(\bar\theta_T\) around
\(\widehat\theta_n\).
Recall that \(\theta_*\) solves \(\mathcal R(\theta_*)=0\), and define
$H(\theta):=\nabla_\theta\mathcal R(\theta)$ and $H_*:=H(\theta_*)$. Let \(\Sigma_g:=\operatorname{Var}\{g_i(\theta_*)\}\). When verifying the high-level conditions in MNP,   $\mathcal R_n(\theta)$ is the sample average negative likelihood score and the information identity gives \(\Sigma_g=H_*\). Our general theory allows them to differ.

\subsection{Setup and regularity conditions}\label{sec:theory-setup}

We first state the initialization and learning-rate conditions.
Throughout, \(C<\infty\) and \(c>0\) denote generic constants.

\begin{assumption}[Initialization and learning rate]\label{ass:algorithmic-setup}
The recursion starts from a preliminary estimator
\(\theta_0=\theta_{\mathrm{init}}\), based on a pilot of size \(n_0\),
such that \(\theta_0\xrightarrow{p}\theta_*\) as \(n_0\to\infty\).
The learning rate has the form
$\gamma_t=\gamma_0t^{-a}$ with $a\in(1/2,1)$ and $t\geq1$.
\end{assumption}

Pilot consistency reflects the local nature of the analysis and imposes no
convergence rate. Under multi-pass sampling, the pilot may use the full
data set or a subset: consistency of both \(\theta_0\) and the full-sample
root \(\widehat\theta_n\) gives \(\|\theta_0-\widehat\theta_n\|={ o_{\mathbb P}}(1)\)
without an additional relative-rate condition.

For any prescribed \(\epsilon\in(0,1)\), pilot consistency and a
sufficiently small fixed learning-rate constant bound the probability of
exiting the local neighborhood by \(\epsilon+o(1)\). Our results condition on
the recursion remaining in this neighborhood, so we use the following notion
of local weak convergence.

\begin{definition}[Local \(\epsilon\)-weak convergence]\label{def:local-epsilon}
For random elements \(Z_k\) and \(Z\) taking values in a
metric space \(\mathbb S\), where \(k\) indexes the relevant asymptotic
sequence, write
\[
    Z_k\overset{\epsilon}{\Rightarrow}Z
\]
if, for every bounded continuous \(f:\mathbb S\to\mathbb R\),
\[
    \limsup_{k\to\infty}
    \big|\mathbb Ef(Z_k)-\mathbb Ef(Z)\big|
    \leq2\epsilon\|f\|_\infty.
\]
\end{definition}

This is not ordinary weak convergence, because the bound depends on
\(\epsilon\). For each prescribed \(\epsilon\), the learning-rate
constant \(\gamma_0\) is chosen sufficiently small and then held fixed along
the asymptotic sequence. Our local analysis accommodates objectives that
need not be globally convex, as in MNP, but allows a small, nonvanishing
probability of leaving the neighborhood. The bound therefore does not by
itself establish ordinary weak convergence for a fixed \(\gamma_0\).
Establishing ordinary weak convergence here would require stronger control of paths outside the local neighborhood.

\subsection{Multi-pass SAUSS}\label{sec:multipass-sauss}

\begin{assumption}[Sampling with replacement]\label{ass:empirical-sampling}
The data set \(D_{1:n}\) consists of \(n\) i.i.d.\ observations. Conditional on the
data set and initialization, each update selects \(m\) indices independently
and uniformly from \(\{1,\ldots,n\}\), independently of previous index
selections and simulation draws.
\end{assumption}

The recursion targets the locally defined full-sample root \(\widehat\theta_n\).
The full-sample estimator \(\widehat\theta_n\) is used to separate
sampling and algorithmic errors. SAUSS does not require computing the full-sample root \(\widehat\theta_n\).

Set
\[
    H_n=\nabla_\theta\mathcal R_n(\widehat\theta_n).
\]
For a trial parameter \(\theta\), the conditional update variance is
\begin{equation}\label{eq:empirical-vmr}
    V_{n,m,\Rsim}(\theta)
    \coloneqq \frac1{mn}\sum_{i=1}^n
       \{g_i(\theta)-\mathcal R_n(\theta)\}
       \{g_i(\theta)-\mathcal R_n(\theta)\}'
    +\frac{1}{m\Rsim n}\sum_{i=1}^n
       \operatorname{Var}\{g_{ir}(\theta)\mid D_i\}.
\end{equation}
The first term is variability from resampling observations
from the empirical distribution; the second is variability from conditional
simulation. Also define the limiting one-update variance
\begin{equation}\label{eq:vmr}
    V_{m,\Rsim}
    \coloneqq
    \frac1m\operatorname{Var}\{g_i(\theta_*)\}
    +\frac1{m\Rsim}
    \mathbb E\!\left[
        \operatorname{Var}\{g_{ir}(\theta_*)\mid D_i\}
    \right].
\end{equation}

Write \(\mathbb E_n^*\) for expectation over one independently
generated update conditional on the data set, at a fixed trial parameter.

\begin{assumption}[Empirical local regularity]\label{ass:empirical-regularity}
The following conditions hold.
\begin{enumerate}
\item[(i)] With probability tending to one, the equation
\(\mathcal R_n(\theta)=0\) has a locally unique, data-measurable solution
\(\widehat\theta_n\) satisfying
\begin{equation}\label{eq:sample-root-expansion}
    \sqrt n(\widehat\theta_n-\theta_*)
    =-H_*^{-1}\frac1{\sqrt n}\sum_{i=1}^n g_i(\theta_*)+{ o_{\mathbb P}}(1).
\end{equation}
Here \(\mathbb E g_i(\theta_*)=0\), \(\Sigma_g\) is finite,
and \(H_n\xrightarrow{p}H_*=H_*'\succ0\).
\item[(ii)] There are fixed \(\rho,c,C>0\) and data events \(E_n\)
with \(\mathbb P(E_n)\to1\) such that, on \(E_n\), the ball
\(B_{n,\rho}\coloneqq\{\theta:\|\theta-\widehat\theta_n\|\leq\rho\}\)
lies in \(\Theta\), \(cI\preceq H_n\preceq CI\), and
\[
    \|\mathcal R_n(\theta)-H_n(\theta-\widehat\theta_n)\|
    \leq C\|\theta-\widehat\theta_n\|^2
    \quad \text{for $\theta\in B_{n,\rho}$}.
\]
\item[(iii)] For some \(\delta>0\), on \(E_n\),
\[
\begin{split}
    \sup_{\theta\in B_{n,\rho}}
    \mathbb E_n^*\|\widehat G_t(\theta)-\mathcal R_n(\theta)\|^{2+\delta}
    &\leq C,\\
    \|V_{n,m,\Rsim}(\theta)-V_{n,m,\Rsim}(\widehat\theta_n)\|
    &\leq C\|\theta-\widehat\theta_n\|.
\end{split}
\]
Moreover,
\(V_{n,m,\Rsim}(\widehat\theta_n)\xrightarrow{p}V_{m,\Rsim}\succ0\),
where \(V_{m,\Rsim}\) is given by \eqref{eq:vmr}.
\item[(iv)] For some \(p_0>(1-a)^{-1}\), on \(E_n\),
\begin{equation}\label{eq:empirical-path-moment}
    \sup_{\theta\in B_{n,\rho}}
    \mathbb E_n^*\|\widehat G_t(\theta)-\mathcal R_n(\theta)\|^{2p_0}
    \leq C.
\end{equation}
\end{enumerate}
\end{assumption}

These conditions impose uniform local bounds on
the empirical mean negative score and update noise, with constants
independent of \(n\) and \(T\). They do not require
continuity of individual simulated scores.
Proposition~\ref{prop:mnp-empirical} verifies the
conditions for MNP.

Write
\[
    \Sigma_{\mathrm{stat}}\coloneqq H_*^{-1}\Sigma_gH_*^{-1},
    \qquad
    \Sigma_{\mathrm{alg}}\coloneqq H_*^{-1}V_{m,\Rsim}H_*^{-1}.
\]

\begin{theorem}[Joint limit for multi-pass SAUSS]\label{thm:multipass}
Suppose Assumptions~\ref{ass:algorithmic-setup} and
\ref{ass:empirical-sampling} and conditions (i)--(iii) of
Assumption~\ref{ass:empirical-regularity} hold, with fixed \(m\) and \(\Rsim\).
Let \(n,T,n_0\to\infty\), with no additional relative-rate restriction.
Choose the radius \(\rho\) in
Assumption~\ref{ass:empirical-regularity} small enough that
\(C\rho\leq c/2\).\par
Using this radius and the data events \(E_n\), define
\begin{equation}\label{eq:multipass-local-event}
    \tau_n\coloneqq\inf\{t\geq0:\theta_t\notin B_{n,\rho}\},
    \qquad
    \mathcal E_n\coloneqq E_n\cap\{\tau_n=\infty\}.
\end{equation}
For every \(\epsilon\in(0,1)\), there is \(\bar\gamma(\epsilon)>0\) such that
every fixed \(0<\gamma_0\leq\bar\gamma(\epsilon)\) gives
\(\mathbb P(\mathcal E_n^c)\leq\epsilon+o(1)\) and:
\begin{enumerate}
\item[(i)]
    $\|\theta_T-\widehat\theta_n\|\mathbb I\{\mathcal E_n\}
    \xrightarrow{p}0$
    and
    $\|\bar\theta_T-\widehat\theta_n\|\mathbb I\{\mathcal E_n\}
    \xrightarrow{p}0$.
\item[(ii)] Let
\(Z_{\mathrm{stat}}\sim\mathcal N(0,\Sigma_{\mathrm{stat}})\) and
\(Z_{\mathrm{alg}}\sim\mathcal N(0,\Sigma_{\mathrm{alg}})\) be independent.
\[
    \left(
      \sqrt n(\widehat\theta_n-\theta_*),
      \sqrt T(\bar\theta_T-\widehat\theta_n)
    \right)
    \overset{\epsilon}{\Rightarrow}
    (Z_{\mathrm{stat}},Z_{\mathrm{alg}})
\]
in \(\mathbb R^p\times\mathbb R^p\).
\item[(iii)] Suppose additionally that condition (iv) of
Assumption~\ref{ass:empirical-regularity} holds.
Define
\(\mathbb X_{n,T}(r)\coloneqq
T^{-1/2}\sum_{t=1}^{\lfloor Tr\rfloor}(\theta_t-\widehat\theta_n)\).
Then
\begin{equation}\label{eq:empirical-joint-fclt}
    \left(\sqrt n(\widehat\theta_n-\theta_*),
          \mathbb X_{n,T}(\cdot)\right)
    \overset{\epsilon}{\Rightarrow}
    \left(Z_{\mathrm{stat}},\Sigma_{\mathrm{alg}}^{1/2}\mathbb W(\cdot)\right)
\end{equation}
in \(\mathbb R^p\times D[0,1]^p\), equipped with the product of the
Euclidean and usual Skorokhod topologies.
Here \(\mathbb W\) is a standard \(p\)-dimensional Wiener process independent
of \(Z_{\mathrm{stat}}\).
\end{enumerate}
\end{theorem}

Condition (iv) of Assumption~\ref{ass:empirical-regularity}
controls the averaging remainder uniformly over the path and is used for
the FCLT and random-scaling inference; it is not required for the endpoint
joint CLT or plug-in inference. The conclusions also hold under a fixed
symmetric positive-definite preconditioner, including fixed block-specific
scaling, and deterministic burn-in satisfying \(T_0/T\to0\), with an
adjusted small-step threshold for preconditioning.
The separate learning-rate constants for the regression and
variance parameters in Section~\ref{sec:application} are an instance
of this fixed block scaling.

Our asymptotic theory separates sampling uncertainty from
algorithmic variation due to mini-batch selection and conditional simulation. The joint limit in Theorem~\ref{thm:multipass} adapts the
statistical and algorithmic error decomposition of
\citet{chenKimLeeSeoSong2025slim} to unbiased simulated scores
and local exit control. The estimation error has the exact decomposition
\[
\begin{split}
    \sqrt T(\bar\theta_T-\theta_*)
    &=\sqrt{T/n}\,\sqrt n(\widehat\theta_n-\theta_*)
      +\sqrt T(\bar\theta_T-\widehat\theta_n),\\
    \sqrt n(\bar\theta_T-\theta_*)
    &=\sqrt n(\widehat\theta_n-\theta_*)
      +\sqrt{n/T}\,\sqrt T(\bar\theta_T-\widehat\theta_n).
\end{split}
\]
If \(T/n\to0\), the statistical error is negligible under
\(\sqrt T\) normalization, with limiting variance \(\Sigma_{\mathrm{alg}}\).
If \(T/n\) has a positive finite limit, both errors remain first order.
When \(T/n\to\infty\), SAUSS has the same first-order
distribution as the full-sample likelihood estimator under \(\sqrt n\)
normalization, in the same local \(\epsilon\)-sense.
Equivalently, for fixed \(m\), these regimes
are indexed by the expected-pass ratio \(mT/n\).
Section~\ref{sec:inference-regimes} develops inference that accounts jointly
for sampling and algorithmic uncertainty in the multi-pass setting.

\subsection{Verification for multinomial probit}\label{sec:mnp-verification}

We verify the model-specific regularity conditions for the exact MNP
simulator of Section~\ref{sec:mnp-sauss} under multi-pass sampling.
Let \(q=\max\{2+\delta,2p_0\}\), where \(\delta>0\) and
\(p_0>(1-a)^{-1}\), as in conditions (iii)--(iv) of
Assumption~\ref{ass:empirical-regularity}.
Since \(a\in(1/2,1)\), we have \(q>4\).

\Needspace{4\baselineskip}
\begin{proposition}[MNP regularity conditions for multi-pass SAUSS]\label{prop:mnp-empirical}
Consider the exact MNP recursion in the normalized coordinates of
\eqref{eq:covariance-cholesky}, under sampling as in
Assumption~\ref{ass:empirical-sampling}. Suppose that:
\begin{enumerate}
\item[(i)] The MNP model is correctly specified at $\theta_*$.

\item[(ii)] The point $\theta_*$ is interior to the parameter space, and
there exists a convex neighborhood $\mathcal N_*$ of $\theta_*$ whose closure
is a compact subset of the parameter space.

\item[(iii)] The covariates satisfy
\begin{equation}\label{eq:mnp-verification-X-moment}
    \mathbb E\|\DX_i\|^{2q}<\infty.
\end{equation}

\item[(iv)] For every $v\in\mathbb R^p$, the choice-probability gradients satisfy
\begin{equation}\label{eq:mnp-verification-rank}
\left[
v'\nabla_\theta P_{ij}(\theta_*)=0
\ \text{almost surely for all }j
\right]\ \Longrightarrow\ v=0.
\end{equation}
\end{enumerate}
Then the exact MNP simulator satisfies
Assumption~\ref{ass:empirical-regularity}. The limiting variance is
\begin{equation}\label{eq:mnp-verification-variance}
    V_{m,\Rsim}
    =\frac1m H_*+\frac1{m\Rsim}
      \mathbb E\!\left[\operatorname{Var}\{g_{ir}(\theta_*)\mid D_i\}\right],
    \qquad
    \lambda_{\min}(V_{m,\Rsim})\geq\frac{\lambda_{\min}(H_*)}m>0.
\end{equation}
Moreover, \(\Sigma_g=H_*\).
\end{proposition}

For MNP, the Hessian matrix takes the form (see \eqref{eq:mnp-verification-observed-minus-missing} in the Supplementary Material for the formal definitions)
\begin{equation}
  H_*
  =I_{\mathrm c}(\theta_*)
   -\mathbb E\!\left[\operatorname{Var}\{g_{ir}(\theta_*)\mid D_i\}\right],
\end{equation}
where $I_{\mathrm c}(\theta_*)$ is the Fisher information that would be available if the latent utility differences $W_i$ were observed and used directly in the Gaussian likelihood.
This shows that not observing the latent utilities loses Fisher information, so positive
complete-data information alone does not ensure \(H_*\succ0\). The latter
follows from the probability-gradient rank condition
\eqref{eq:mnp-verification-rank}, imposed in normalized coordinates:
fixing \(L_{11}=1\) removes scale but does not itself ensure local
identification.

\section{Inference}\label{sec:inference}

We develop plug-in Wald and delta-method inference for
multi-pass SAUSS. The estimated asymptotic variance accounts for sampling uncertainty
and the additional variation from mini-batch resampling and score simulation.
Throughout this section, \(T_0=0\).

\subsection{Plug-in inference}\label{sec:inference-definitions}\label{sec:inference-regimes}

Wald inference based on Theorem~\ref{thm:multipass} uses
symmetric estimators \({ \widehat H}\),
\({ \widehat\Sigma}_g\), and \({ \widehat V}_{m,\Rsim}\) of \(H_*\), \(\Sigma_g\),
and \(V_{m,\Rsim}\), respectively. The estimated algorithmic variance is
\begin{equation}\label{eq:algorithmic-covariance-estimate}
    \widehat\Sigma_{\mathrm{alg}}
    \coloneqq{ \widehat H}^{-1}{ \widehat V}_{m,\Rsim}{ \widehat H}^{-1}.
\end{equation}
These components give the first-order variance approximation and its estimator
\begin{equation}\label{eq:multipass-covariance}
    C_{n,T}\coloneqq
    \frac1n\Sigma_{\mathrm{stat}}+\frac1T\Sigma_{\mathrm{alg}},
    \qquad
    \widehat C_{n,T}\coloneqq
    \frac1n{ \widehat H}^{-1}{ \widehat\Sigma}_g{ \widehat H}^{-1}
    +\frac1T\widehat\Sigma_{\mathrm{alg}}.
\end{equation}
To distinguish the two sources of algorithmic variation, define
\[
    \Sigma_{\mathrm{sim}}\coloneqq H_*^{-1}
       \mathbb E\!\left[\operatorname{Var}\{g_{ir}(\theta_*)\mid D_i\}\right]
       H_*^{-1},
\]
where \(g_{ir}\) is the negative complete-data score from one exact
conditional draw, with conditional mean \(g_i\).
Equation~\eqref{eq:vmr} gives
\(\Sigma_{\mathrm{alg}}
=\Sigma_{\mathrm{stat}}/m+\Sigma_{\mathrm{sim}}/(m\Rsim)\).
Thus the first-order variance approximation has three components:
\begin{equation}\label{eq:multipass-three-components}
    C_{n,T}
    =\underbrace{\frac{\Sigma_{\mathrm{stat}}}{n}}_{\text{sampling}}
     +\underbrace{\frac{\Sigma_{\mathrm{stat}}}{mT}}_{\text{mini-batch resampling}}
     +\underbrace{\frac{\Sigma_{\mathrm{sim}}}{m\Rsim T}}_{\text{score simulation}}.
\end{equation}
The sampling component represents uncertainty in the full-sample estimator
\(\widehat\theta_n\). The resampling component is additional variation from
approximating the full-sample score with randomly selected observations;
its variance matrix is proportional to that of the sampling component
because the empirical score variance
converges to \(\Sigma_g\). The simulation component arises from averaging
\(\Rsim\) complete-data scores to estimate each observation's score.
It reflects uncertainty about the unobserved latent variables, as in the
missing-information calculations in Section~\ref{sec:mnp-verification}
and Supplementary Appendix~\ref{sec:supp-binary-probit}.
The simulation error has conditional mean zero, so its variance adds
to the resampling variance without a cross term.

\begin{remark}[Variance decomposition]\label{rem:var-decom}
Increasing \(T\) reduces both algorithmic components but does not eliminate
sampling uncertainty. Increasing \(\Rsim\) reduces only simulation variance.
Increasing \(n\) with \(T\) fixed reduces only the sampling component.
Increasing \(m\) with \(T\) fixed reduces both algorithmic components while
using more observation-level evaluations; at fixed \(mT\) and \(\Rsim\),
both algorithmic components are unchanged.
Resampling variance is \(n/(mT)\) times sampling variance, where this factor
is the reciprocal of the expected number of passes. The applications in
Section~\ref{sec:application} use approximately \(652\) and \(3{,}137\) passes, so this factor is
\(0.15\%\) and \(0.03\%\).
These comparisons concern the first-order variance approximation.
\end{remark}

Let \(A\) be a fixed \(\ell\times p\) matrix with full row rank, let
\(a_0\in\mathbb R^\ell\) be fixed and known, and consider
\(H_0:A\theta_*=a_0\). Define the Wald statistic
\[
    W_{n,T}\coloneqq
    (A\bar\theta_T-a_0)'(A\widehat C_{n,T}A')^{-1}
    (A\bar\theta_T-a_0).
\]
Let \(\chi^2_{\ell,1-\alpha}\) denote the \(1-\alpha\)
quantile of the chi-squared distribution with \(\ell\) degrees of freedom.

\begin{corollary}[Wald inference for linear restrictions]\label{cor:multipass-inference}
Under the conditions for the joint CLT in
Theorem~\ref{thm:multipass}(ii), with \(\mathcal E_n\) as in
\eqref{eq:multipass-local-event}, suppose that
\[
    \big\{\|{ \widehat H}-H_*\|+\|{ \widehat\Sigma}_g-\Sigma_g\|
       +\|{ \widehat V}_{m,\Rsim}-V_{m,\Rsim}\|\big\}
       \mathbb I\{\mathcal E_n\}\xrightarrow{p}0.
\]
Suppose either that \(n/T\to\kappa\in[0,\infty)\), with
\(\Sigma_\kappa \coloneqq
\Sigma_{\mathrm{stat}}+\kappa\Sigma_{\mathrm{alg}} \succ0\), or that \(n/T\to\infty\).
Under \(H_0\) and either rate condition, for each \(\alpha\in(0,1)\),
\[
    \limsup_{n,T\to\infty}
    \left|\mathbb P\{W_{n,T}>\chi^2_{\ell,1-\alpha}\}-\alpha\right|
    \leq
    \limsup_{n,T\to\infty}\mathbb P(\mathcal E_n^c).
\]
\end{corollary}

For a scalar linear contrast \(A\theta_*\) and a smooth scalar
parameter \(h(\theta_*)\), the Wald and delta-method intervals,
respectively, are
\[
\begin{aligned}
    &A\bar\theta_T\ \pm\
    z_{1-\alpha/2}\{A\widehat C_{n,T}A'\}^{1/2},\\
    &I_{n,T}^h\coloneqq\ h(\bar\theta_T)\ \pm\ z_{1-\alpha/2}
    \left\{\nabla_\theta h(\bar\theta_T)'
       \widehat C_{n,T}\,\nabla_\theta h(\bar\theta_T)\right\}^{1/2}.
\end{aligned}
\]
Here \(z_{1-\alpha/2}\) denotes the standard-normal quantile.
Corollary~\ref{cor:multipass-inference} bounds the absolute
limiting coverage error of the linear interval by
\(\limsup_{n,T\to\infty}\mathbb P(\mathcal E_n^c)\), which pilot
consistency and a sufficiently small learning-rate constant can make
arbitrarily small. Supplementary Appendix~\ref{sec:supp-delta-method}
states the smoothness conditions and establishes the same coverage-error
bound for \(I_{n,T}^h\) in Corollary~\ref{cor:multipass-delta}.
Section~\ref{sec:mnp-covariance} gives the MNP variance estimators. Supplementary
Appendix~\ref{sec:supp-multipass-rs} develops random-scaling inference,
a conservative alternative that avoids separate variance estimation.

\subsection{MNP variance estimation}\label{sec:mnp-covariance}

For each observation, we obtain a fixed number \(S\geq2\) of
independent replicates \(\widehat g_{\Rsim}^{(s)}(\bar\theta_T;D_i)\),
\(s=1,\ldots,S\), of the \(\Rsim\)-draw score estimator in
\eqref{eq:mnp-one-observation-gradient}, evaluated at \(\bar\theta_T\).
These replicates use fresh exact conditional draws, independent across
observations and replicates conditional on \(D_{1:n}\) and \(\bar\theta_T\).
The replicate means are
\[
    \bar g_i\coloneqq\frac1S\sum_{s=1}^S
        \widehat g_{\Rsim}^{(s)}(\bar\theta_T;D_i),
    \qquad
    \bar g\coloneqq\frac1n\sum_{i=1}^n\bar g_i.
\]
Variation in \(\bar g_i\) across observations includes simulation
noise. Removing its contribution gives the estimator of the observed-data
score variance
\begin{equation}\label{eq:mnp-sampling-plugin}
\begin{split}
    \widehat\Sigma_g
    &\coloneqq\frac1n\sum_{i=1}^n\bar g_i\bar g_i'-\bar g\bar g'\\
    &\quad-\frac1{nS(S-1)}\sum_{i=1}^n\sum_{s=1}^S
        \{\widehat g_{\Rsim}^{(s)}(\bar\theta_T;D_i)-\bar g_i\}
        \{\widehat g_{\Rsim}^{(s)}(\bar\theta_T;D_i)-\bar g_i\}'.
\end{split}
\end{equation}
The subtraction estimates the simulation variance of \(\bar g_i\):
the sample variance across the \(S\) independent score replicates is
divided by \(S\). The correction is justified by
\[
\begin{aligned}
    &\mathbb E\!\left[
       \bar g_i\bar g_i'-\frac1{S(S-1)}\sum_{s=1}^S
       \{\widehat g_{\Rsim}^{(s)}(\bar\theta_T;D_i)-\bar g_i\}
       \{\widehat g_{\Rsim}^{(s)}(\bar\theta_T;D_i)-\bar g_i\}'
       \,\middle|\,D_i,\bar\theta_T\right]\\
    &\qquad=g_i(\bar\theta_T)g_i(\bar\theta_T)'.
\end{aligned}
\]
Under correct MNP specification, the information identity gives
\(H_*=\Sigma_g\). Thus \({ \widehat H}={ \widehat\Sigma}_g\) estimates curvature without
evaluating the observed-data likelihood Hessian.
To estimate \(\Sigma_{\mathrm{sim}}\) in
Section~\ref{sec:inference-regimes}, define
\begin{equation}\label{eq:mnp-simulation-plugin}
    \widehat\Sigma_{\mathrm{sim}}\coloneqq
    \frac{\Rsim}{n(S-1)}\widehat\Sigma_g^{-1}
    \left[
       \sum_{i=1}^n\sum_{s=1}^S
       \{\widehat g_{\Rsim}^{(s)}(\bar\theta_T;D_i)-\bar g_i\}
       \{\widehat g_{\Rsim}^{(s)}(\bar\theta_T;D_i)-\bar g_i\}'
    \right]\widehat\Sigma_g^{-1}.
\end{equation}
The factor \(\Rsim\) converts the variance of an \(\Rsim\)-draw score
estimate to the corresponding single-draw variance.
Combining score variation across observations with conditional
simulation variation gives the variance estimator for an
\(m\)-observation update:
\begin{equation}\label{eq:mnp-vmr-plugin}
\begin{split}
    \widehat V_{m,\Rsim}^{\mathrm{MNP}}
    &\coloneqq
    \frac1m\left(
       \widehat\Sigma_g+\frac1{\Rsim}\widehat\Sigma_g
       \widehat\Sigma_{\mathrm{sim}}\widehat\Sigma_g
    \right)\\
    &=\frac{1}{mnS}
    \sum_{i=1}^{n}\sum_{s=1}^{S}
    \{\widehat g_{\Rsim}^{(s)}(\bar\theta_T;D_i)-\bar g\}
    \{\widehat g_{\Rsim}^{(s)}(\bar\theta_T;D_i)-\bar g\}'.
\end{split}
\end{equation}
With
\({ \widehat V}_{m,\Rsim}=\widehat V_{m,\Rsim}^{\mathrm{MNP}}\),
the three estimated variance components are
\[
\begin{aligned}
    \widehat C_{n,T}
    &=\frac1n\widehat\Sigma_g^{-1}
      +\frac1{mT}\widehat\Sigma_g^{-1}
      +\frac1{m\Rsim T}\widehat\Sigma_{\mathrm{sim}}.
\end{aligned}
\]
The terms estimate sampling, mini-batch resampling, and simulation variance,
respectively, as in \eqref{eq:multipass-three-components}.
Under the conditions of Proposition~\ref{prop:mnp-empirical},
averaging across an increasing number of observations makes
\(\widehat\Sigma_g\), \(\widehat\Sigma_{\mathrm{sim}}\), and
\(\widehat V_{m,\Rsim}^{\mathrm{MNP}}\) locally
consistent even with a fixed \(S\geq2\). Increasing \(S\) reduces
Monte Carlo noise in the variance estimates;
it does not change the first-order asymptotic variance of \(\bar\theta_T\).

\begin{remark}[Variance decomposition in the applications]\label{rem:var-decom-app}
Both applications evaluate \(\widehat C_{n,T}\) at \(\bar\theta_T\) with \(\Rsim=5\) and \(S=100\), and the runs make
approximately \(652\) and \(3{,}137\) passes. For the family-outcomes coefficient in the
labor-supply application, sampling, mini-batch resampling, and score simulation account for
\(97.72\%\), \(0.15\%\), and \(2.13\%\) of the total estimated variance,
respectively. The corresponding shares for the price coefficient in the
brand-choice application are \(96.37\%\), \(0.03\%\), and \(3.60\%\).
In both applications, sampling variation dominates, the mini-batch share is
close to the inverse of the number of passes, and the simulation component is the larger algorithmic term.
\end{remark}

\section{Monte Carlo Experiments}\label{sec:monte-carlo}

\subsection{\texorpdfstring{Choice of hyperparameters}{Choice of hyperparameters}}\label{sec:numerical-settings}

For multi-pass MNP, we fix \(a=0.501\), \(m=10\), \(\Rsim=5\), \(T_0=0\), and \(K_0=8\Rsim J\), and we use the same
\(\Rsim=5\) in the pilot, the main run, and the standard-error calculations. We use \(S=100\) independent score
replicates, which controls the numerical precision of the variance estimates. These are practical defaults;
Supplementary Section~\ref{sec:supp-numerical-settings} gives fuller guidance.

Our theory requires a consistent preliminary estimate and sufficiently small fixed learning-rate scales. The pilot
may use the full sample or a representative subset whose size grows with \(n\). After the pilot, we hold the scales
fixed, restart the learning-rate sequence at \(t=1\), and average all subsequent iterates. We set the overall scale
by the curvature rule of Supplementary Section~\ref{sec:supp-numerical-settings} with the multiplier \(c=0.5\), in
the applications and in the multi-pass inference experiments.

The variance decomposition in \eqref{eq:multipass-three-components} also gives a way to plan the number of updates.
For a prespecified contrast \(d'\theta_*\) with \(d'\Sigma_{\mathrm{stat}}d>0\), let
\(q_d=d'\Sigma_{\mathrm{sim}}d/(d'\Sigma_{\mathrm{stat}}d)\), and let \(P=10T/n\) be the expected number of passes
through the sample. To first order,
\begin{equation}\label{eq:numerical-settings-ratio}
    \frac{\text{algorithmic variance of }d'\bar\theta_T}
         {\text{sampling variance of }d'\bar\theta_T}
    =\frac{1+q_d/5}{P}.
\end{equation}
Using pilot estimates of \(q_d\), we choose the run length so that the estimated first-order ratio of
algorithmic to sampling variance is at most \(5\%\). At that bound the first-order standard error exceeds the sampling-only standard error by about
\(2.5\%\). We call the smallest \(T\) at which this pilot-based ratio is at most \(5\%\) in every direction the planned run length
\(T_{\mathrm{plan}}\).

\subsection{\texorpdfstring{Monte Carlo results}{Monte Carlo results}}\label{sec:monte-carlo-results}

Plug-in intervals keep coverage near the nominal \(95\%\) level even when the run makes only \(6.4\) passes, and with the run length chosen from a pilot, SAUSS matches the accuracy of GHK SML in about a tenth of the time.
Each SAUSS design uses \(1{,}000\) replications and the hybrid sampler (except in the comparison of conditional samplers), and times are medians on one core of an AMD EPYC
9655 processor. The baseline design D1 has equicorrelated utility differences (Supplementary
Section~\ref{sec:supp-mc-design}). The multi-pass inference experiments follow Section~\ref{sec:numerical-settings}.
The scaling and sampler experiments instead fix \(\gamma_\beta=\gamma_L=0.2\) and \(T=80{,}000\), or \(800{,}000\)
where stated. Except in the sampler experiment, they start from \(\beta=0\) and \(\Omega=I\). Supplementary
Appendix~\ref{app:supp-mc} gives the complete designs and results, including two robustness designs, online SAUSS
with random-scaling intervals, and a comparison of conditional samplers.

\begin{table}[!htb]
\centering
\caption{Multi-pass SAUSS inference and comparison with GHK SML, design D1, \(J=4\)}
\label{tab:mc-main-inference}
\begingroup
\spacingset{1}
\small
\setlength{\tabcolsep}{4pt}
\begin{tabular*}{\textwidth}{@{\extracolsep{\fill}}lllcccccc@{}}
\toprule
\multicolumn{9}{@{}l}{\emph{Panel A: Inference at a fixed computational budget}} \\
 & & & & \multicolumn{2}{c}{Coverage} & \multicolumn{3}{c}{Share of variance} \\
\cmidrule(lr){5-6}\cmidrule(lr){7-9}
Target & \(n\) & \(T\) & SE/SD & SAUSS & Naive & Sampling & Minib. & Simul. \\
\midrule
\(\beta_2/\beta_3\) & 0.2M & 1.28M & 1.05 & 0.959 & 0.957 & 0.973 & 0.015 & 0.012 \\
                    & 2M   & 1.28M & 1.01 & 0.958 & 0.922 & 0.781 & 0.122 & 0.097 \\
                    & 2M   & 5.12M & 1.00 & 0.960 & 0.949 & 0.935 & 0.037 & 0.029 \\
\(\beta_2\)         & 0.2M & 1.28M & 0.98 & 0.942 & 0.895 & 0.705 & 0.011 & 0.284 \\
                    & 2M   & 1.28M & 0.97 & 0.949 & 0.595 & 0.193 & 0.030 & 0.777 \\
                    & 2M   & 5.12M & 1.03 & 0.959 & 0.837 & 0.489 & 0.019 & 0.492 \\
\(\Omega_{22}\)     & 0.2M & 1.28M & 1.00 & 0.958 & 0.924 & 0.768 & 0.012 & 0.220 \\
                    & 2M   & 1.28M & 0.96 & 0.945 & 0.645 & 0.249 & 0.039 & 0.712 \\
                    & 2M   & 5.12M & 1.01 & 0.955 & 0.868 & 0.570 & 0.022 & 0.408 \\
\bottomrule
\end{tabular*}
\par\vspace{1ex}
\footnotesize
\setlength{\tabcolsep}{2pt}
\begin{tabular*}{\textwidth}{@{\extracolsep{\fill}}lrrccccc@{}}
\multicolumn{8}{@{}l}{\emph{Panel B: SAUSS with \(T_{\mathrm{plan}}\) and GHK SML on the same data sets, \(n=2{,}000\)}} \\
 & & & \multicolumn{5}{c}{RMSE \(\times10^3\) (coverage)} \\
\cmidrule(lr){4-8}
Estimator & Passes & Time (s) & \(\beta_2/\beta_3\) & \(P_2(x_0)\) & \(\beta_2\) & \(C_{23}\) &
\(\Omega_{22}\) \\
\midrule
SAUSS, \(T_{\mathrm{plan}}/16\) & 51 & 0.8 & 31.6 (0.947) & 19.4 (0.951) & 74.1 (0.961)
 & 104.1 (0.957) & 189.4 (0.936) \\
SAUSS, \(T_{\mathrm{plan}}/4\) & 203 & 1.3 & 31.2 (0.944) & 18.4 (0.951) & 65.2 (0.961)
 & 96.1 (0.960) & 177.1 (0.947) \\
SAUSS, \(T_{\mathrm{plan}}\) & 813 & 3.3 & 31.0 (0.946) & 18.4 (0.947) & 64.4 (0.970) & 95.2 (0.952)
 & 177.5 (0.944) \\
GHK SML                & --  & 36.8 & 31.0 (0.945) & 18.5 (0.942) & 65.0 (0.961) & 95.3 (0.946)
 & 179.9 (0.939) \\
\bottomrule
\end{tabular*}
\par
\endgroup
\begin{minipage}{\textwidth}
\vspace{1ex}
\spacingset{1}
\footnotesize
\emph{Notes:} SE/SD: mean standard error over the Monte Carlo standard deviation. Coverage of
nominal \(95\%\) intervals: SAUSS uses the full plug-in variance \(\widehat C_{n,T}\), which adds the mini-batch and
simulation components to the sampling component; Naive uses the sampling component \(\widehat\Sigma_{\mathrm{stat}}/n\)
alone. \(P_2(x_0)\): choice
probability of alternative 2 at a fixed covariate value; \(C_{23}\): correlation of utility differences 2 and 3;
M: million. Shares of variance: sampling,
mini-batch, and simulation components of the full plug-in variance. Panel B:
\(T_{\mathrm{plan}}\) of Section~\ref{sec:numerical-settings} (median \(162{,}505\) updates); Time
(s): median computation time, including the SAUSS pilot and calibration, the logit start for GHK SML,
and variance estimation for both methods.
\end{minipage}
\end{table}

We first ask whether the plug-in intervals stay valid when the run makes few passes over the
data. With \(J=4\), the multi-pass experiment increases \(n\) from \(2{,}000\) to \(2\) million at \(T=1.28\) million
updates, so the run makes from \(6{,}400\) down to \(6.4\) passes. Panel A of Table~\ref{tab:mc-main-inference}
reports the two largest sample sizes for the coefficient ratio \(\beta_2/\beta_3\), \(\beta_2\), and the variance
\(\Omega_{22}\); Supplementary Table~\ref{tab:mc-inf-multipass} has all five targets and sample sizes. With the full
plug-in variance \(\widehat C_{n,T}\) of Section~\ref{sec:mnp-covariance}, the \(95\%\) confidence intervals cover
between \(0.938\) and \(0.962\) for every target and sample size. The naive interval, based on the sampling component
alone, has nearly the same coverage when \(n\leq20{,}000\), where the run makes hundreds of passes. At \(n=2\) million
and \(6.4\) passes, it covers only \(0.595\) for \(\beta_2\) and \(0.645\) for \(\Omega_{22}\). The simulation
component then accounts for \(0.78\) and \(0.71\) of their total variance, as \eqref{eq:multipass-three-components}
predicts when \(n\) grows at fixed \(T\). Lengthening the run to \(T=5.12\) million raises the naive coverage to
\(0.84\) to \(0.95\), while the full interval stays at \(0.948\) to \(0.960\).

We next let the pilot also choose the run length. Panel B of Table~\ref{tab:mc-main-inference}
reports SAUSS run with the full procedure, with \(J=4\) and
\(n=2{,}000\): the pilot, the curvature rule, and \(T_{\mathrm{plan}}\) (median \(162{,}505\) updates and \(813\)
passes). GHK SML is fitted to the same data sets. At \(T_{\mathrm{plan}}\) the root mean squared errors of the two
estimators are within \(1.5\%\) of each other on all five targets, and coverage is within about two percentage
points of the nominal level. SAUSS takes \(3.3\) seconds, including standard errors, against \(36.8\) seconds for GHK SML.
A quarter of \(T_{\mathrm{plan}}\) gives nearly the same accuracy for the five reported quantities in \(1.3\) seconds, and even a
sixteenth keeps coverage near the nominal level, because the plug-in variance accounts for the larger algorithmic
variation of a shorter run.

With the default implementation and initialization, SAUSS completes runs for which GHK SML returns no estimate. For \(n=2{,}000\) and
\(20{,}000\) (\(T=80{,}000\) and \(800{,}000\), respectively), SAUSS completes every replication
at every \(J\). GHK SML from the default start returns no estimate in the attempted runs at \(J=16\)
at either sample size; the \(n=20{,}000\) batch was stopped early. It returns estimates in every data set at \(J=4\), but no estimate in \(7.7\%\) of the data sets at \(J=8\) and \(n=2{,}000\).
On common data sets both estimators recover \(\beta\) equally well. At \(J=8\) and \(n=2{,}000\), GHK SML is less
accurate for the correlations and choice probabilities. At \(n=20{,}000\), SAUSS with \(T=800{,}000\) matches GHK SML
and is \(7\) and \(14\) times faster at \(J=4\) and \(8\). SAUSS recovers the coefficients and the choice
probabilities at every \(J\). The correlations are harder: at \(n=20{,}000\) their root mean squared error is
\(0.12\) at \(J=16\), against \(0.03\) at \(J=4\). These feasibility comparisons concern the implementations and initialization used here. Supplementary
Tables~\ref{tab:mc-ghk} and \ref{tab:mc-ghk-n20k} report the details.

The hybrid sampler matters most when the iterates start far from the truth. In a \(J=4\)
experiment started at \(\beta=8\beta_0\), the hybrid sampler takes \(1.9\) seconds,
against \(37.9\) seconds for capped accept--reject (Remark~\ref{rem:capped}), whose omitted contributions raise the
root mean squared error of \(\widehat\beta\) \(2.5\)-fold (Supplementary Section~\ref{sec:mc-simulator}).

\section{Conclusion}\label{sec:conclusion}

SAUSS combines the method of simulated scores with mini-batch stochastic approximation, with inference that accounts for both sampling and algorithmic variation. Three extensions would be useful: faster exact conditional samplers for larger choice sets,
theory that lets \(J\) grow with \(n\), and structured models for the variance matrix of the latent utilities to
estimate correlations more precisely.

\printbibliography[title={References}]
\end{refsection}
\clearpage
\setcounter{saussmainpages}{\value{page}}
\addtocounter{saussmainpages}{-1}
\typeout{SAUSS MAIN PAGES: \arabic{saussmainpages}}

% Supplementary material: a separate citation scope and numbering.
\appendix
\renewcommand{\thepage}{S-\arabic{page}}
\setcounter{page}{1}
\counterwithin{equation}{section}
\setcounter{equation}{0}
\counterwithin{figure}{section}
\setcounter{figure}{0}
\counterwithin{table}{section}
\setcounter{table}{0}
\counterwithin{theorem}{section}
\setcounter{theorem}{0}
\counterwithin{proposition}{section}
\setcounter{proposition}{0}
\counterwithin{corollary}{section}
\setcounter{corollary}{0}
\counterwithin{lemma}{section}
\setcounter{lemma}{0}
\counterwithin{assumption}{section}
\setcounter{assumption}{0}
\counterwithin{example}{section}
\setcounter{example}{0}
\counterwithin{definition}{section}
\setcounter{definition}{0}
\counterwithin{remark}{section}
\setcounter{remark}{0}

\begin{refsection}

\spacingset{1}
\begin{center}
{\large\bfseries Supplementary Material for\\
SAUSS: Stochastic Approximation with Unbiased Simulated Scores for\\
Limited Dependent Variable Models\par}
\end{center}
\onehalfspacing

This supplement provides proofs, implementation details,
and additional simulation and application material supporting the main text.
Appendices~\ref{app:supp-score} and~\ref{app:supp-mnp}
develop the conditional-score and exact hybrid-sampling constructions.
Appendices~\ref{app:supp-theory}--\ref{app:mnp-verification} provide
additional theory, proofs, inference results, and model-specific
verification; Appendices~\ref{app:supp-mc} and~\ref{app:supp-applications}
present the full simulation results and additional application material,
including practical guidance on the hyperparameter choices in
Section~\ref{sec:numerical-settings} of the main text.

\section{Conditional-score construction}\label{app:supp-score}

This appendix proves the conditional-score identity of
Lemma~\ref{lem:hm-identity} of the main text and derives
closed-form score and variance expressions for the binary-probit model
in Example~\ref{ex:binary-probit} of the main text.
Fisher's identity expresses the observed-data likelihood score
as the conditional expectation of the complete-data score given the observed
data \citep[see, e.g.,][]{louis1982finding}. For limited dependent variable
models, \citet{hajivassiliou1998method} develop simulated-score estimators
based on a conditional-expectation representation of the likelihood score.
All probabilities and densities below are
conditional on \(X_i\).

\subsection{Proof of the conditional-score identity}

\begin{proof}[Proof of Lemma~\ref{lem:hm-identity} of the main text]
The integrable-envelope condition permits differentiation
under the integral:
\begin{align}
    -\nabla_\theta P_i(\theta)
    &=
    \int_{\mathcal{D}_i}
        -\nabla_\theta f_i({y^*};\theta)
    \,d{y^*} \notag \\
    &=
    \int_{\mathcal{D}_i}
        \left[
            -\nabla_\theta \log f_i({y^*};\theta)
        \right]
        f_i({y^*};\theta)
    \,d{y^*}.
    \label{eq:supp-ldv-score-derivation}
\end{align}
Dividing \eqref{eq:supp-ldv-score-derivation} by \(P_i(\theta)\) gives
\[
    g_i(\theta)
    =
    -\nabla_\theta \log P_i(\theta)
    =
    \int_{\mathcal{D}_i}
        -\nabla_\theta \log f_i({y^*};\theta)
        \frac{f_i({y^*};\theta)}{P_i(\theta)}
    \,d{y^*},
\]
which is the identity in Lemma~\ref{lem:hm-identity} of the main text because \(f_i({y^*};\theta)/P_i(\theta)\) is the conditional density of \(Y_i^*\) given \(Y_i^*\in\mathcal{D}_i\), for \(y^*\in\mathcal D_i\).
\end{proof}

\subsection{Binary probit in closed form}\label{sec:supp-binary-probit}

For the binary probit of Example~\ref{ex:binary-probit} of the main text,
the negative likelihood score, its simulation variance, and the one-update variance
\(V_{m,\Rsim}\) all have closed forms. Write \(\mu_i\coloneqq\DX_i\theta\),
\(s_i\coloneqq2y_i-1\), and
\(\varepsilon_i\coloneqq Y_i^*-\mu_i\sim\mathcal N(0,1)\), so that the
observed outcome has probability \(\Phi(s_i\mu_i)\). Let \(\phi\) denote the
standard normal density, and define
\[
    \lambda_i\coloneqq\frac{\phi(\mu_i)}{\Phi(s_i\mu_i)},
    \qquad
    \kappa(\mu)\coloneqq\frac{\phi(\mu)^2}{\Phi(\mu)\Phi(-\mu)}.
\]

\textit{Negative likelihood score.} Conditional on the observed outcome, \(s_i\varepsilon_i\) is a
standard normal variable truncated below at \(-s_i\mu_i\). Its mean is the
inverse Mills ratio \(\lambda_i\), so
\(\mathbb E_\theta(\varepsilon_i\mid Y_i^*\in\mathcal D_i)=s_i\lambda_i\) and
\[
    \mathbb E_\theta\{-\DX_i'\varepsilon_i\mid Y_i^*\in\mathcal D_i\}
    =-s_i\lambda_i\DX_i'
    =-\nabla_\theta\log\Phi(s_i\mu_i)
    =g_i(\theta).
\]
This verifies the identity in Example~\ref{ex:binary-probit} of the main text.

\textit{Simulation variance and missing information.} The variance of the
truncated variable is \(1-\lambda_i(\lambda_i+s_i\mu_i)\), and
\(-\partial^2\log\Phi(s_i\mu)/\partial\mu^2=\lambda_i(\lambda_i+s_i\mu_i)\) at
\(\mu=\mu_i\). Hence
\[
    \operatorname{Var}\{g_{ir}(\theta)\mid D_i\}
    =\{1-\lambda_i(\lambda_i+s_i\mu_i)\}\DX_i'\DX_i,
    \qquad
    \nabla_\theta g_i(\theta)
    =\lambda_i(\lambda_i+s_i\mu_i)\DX_i'\DX_i.
\]
Thus
\(\nabla_\theta g_i(\theta)=\DX_i'\DX_i-\operatorname{Var}\{g_{ir}(\theta)\mid D_i\}\):
for each observation, observed information equals the complete-data
information \(\DX_i'\DX_i\) minus the missing information, as in
\eqref{eq:mnp-verification-louis}.
For \(\Rsim\) independent conditional draws, the conditional
variance of \(\widehat g_{\Rsim}(\theta;D_i)\) equals
\(\operatorname{Var}\{g_{ir}(\theta)\mid D_i\}/\Rsim\).

\textit{One-update variance.} Assume the binary-probit model
is correctly specified, \(\mathbb E\|\DX_i\|^2<\infty\), and \(H_*\succ0\).
At \(\theta_*\), with \(\mu_i=\DX_i\theta_*\)
and \(\mathbb P(y_i=1\mid\DX_i)=\Phi(\mu_i)\), averaging over \(y_i\) gives
\(\mathbb E(\lambda_i^2\mid\DX_i)
=\mathbb E\{\lambda_i(\lambda_i+s_i\mu_i)\mid\DX_i\}=\kappa(\mu_i)\). Hence
\[
    \Sigma_g=H_*=\mathbb E\{\kappa(\mu_i)\DX_i'\DX_i\},
    \qquad
    \mathbb E[\operatorname{Var}\{g_{ir}(\theta_*)\mid D_i\}]
    =\mathbb E[\{1-\kappa(\mu_i)\}\DX_i'\DX_i],
\]
and the general variance formula \eqref{eq:vmr}
of the main text becomes
\[
    V_{m,\Rsim}
    =\frac1m\mathbb E\!\left[\left\{\kappa(\mu_i)
      +\frac{1-\kappa(\mu_i)}{\Rsim}\right\}\DX_i'\DX_i\right].
\]
The function \(\kappa\) attains its maximum \(2/\pi\) at \(\mu=0\). Indeed, with
\(h(x)\coloneqq\phi(x)/\Phi(-x)\) the standard normal hazard,
\((\log\kappa)'(\mu)=h(\mu)-h(-\mu)-2\mu=\int_{-\mu}^{\mu}\{h'(x)-1\}\,dx\),
where \(1-h'(x)=\operatorname{Var}(Z\mid Z>x)>0\) for
\(Z\sim\mathcal N(0,1)\). The derivative is therefore negative for
\(\mu>0\) and positive for \(\mu<0\). On average over
the outcome, at least the fraction \(1-2/\pi\approx0.36\) of the complete-data
information is therefore missing. Since \(1-\kappa\geq(\pi/2-1)\kappa\),
\[
    V_{m,\Rsim}\succeq\left(1+\frac{\pi/2-1}{\Rsim}\right)\frac{H_*}{m},
    \qquad
    \Sigma_{\mathrm{alg}}\succeq\left(1+\frac{\pi/2-1}{\Rsim}\right)
      \frac{\Sigma_{\mathrm{stat}}}{m},
\]
with equality if and only if \(\mu_i=0\) almost surely. Exact scores give
\(\Sigma_{\mathrm{alg}}=\Sigma_{\mathrm{stat}}/m\); simulating them with
\(\Rsim\) draws therefore inflates the algorithmic variance
under these conditions by at least the factor
\(1+(\pi/2-1)/\Rsim\approx1+0.571/\Rsim\), which is \(1.114\) at \(\Rsim=5\).
Thus \(m\Sigma_{\mathrm{alg}}\) is at least
\(1+(\pi/2-1)/\Rsim\) times \(\Sigma_{\mathrm{stat}}\) in the
positive-semidefinite ordering, strengthening the bound in
\eqref{eq:rs-covariance-bound} in Appendix~\ref{sec:supp-multipass-rs} used for adjusted
random-scaling inference.

\section{Exact hybrid conditional sampling}\label{app:supp-mnp}\label{sec:supp-hybrid}

This appendix gives the construction and exactness proof for the hybrid
sampler of Remark~\ref{rem:hybrid} of the main text. Fix an observation
\(D_i\) and the current parameter \(\theta\), and let \(\mathcal Q_{i,\theta}\)
denote the law of \(\nu\sim\mathcal N(0,\Omega)\) conditional on
\(\DX_i\beta+\nu\in\mathcal D_i\), the exact accepted conditional law of
Section~\ref{sec:mnp-sauss} of the main text. Take a fixed finite
integer \(K_0\geq1\). The sampler has two steps.
\begin{enumerate}
\item[(1)] Draw independent proposals
\(\nu^{(1)},\nu^{(2)},\ldots\sim\mathcal N(0,\Omega)\), set
\(A_l\coloneqq\mathbb I\{\DX_i\beta+\nu^{(l)}\in\mathcal D_i\}\), and stop after
\(N\coloneqq\min\{\ell\geq1:\sum_{l=1}^{\ell}A_l=\Rsim\}\wedge K_0\) proposals.
Retain the \(k\coloneqq\sum_{l=1}^{N}A_l\) accepted proposals.
\item[(2)] If \(k<\Rsim\), generate \(\Rsim-k\) further independent draws from
\(\mathcal Q_{i,\theta}\) with the minimax exponential tilting sampler of
\citet{botev2017normal}, called the Botev sampler below, using randomness
independent of step~(1). The Botev sampler uses an exponentially tilted Gaussian proposal
with an accept--reject correction, so its draws are exact.
\end{enumerate}

\begin{lemma}[Exactness of the hybrid sampler]\label{lem:hybrid-exact}
The \(\Rsim\) completed draws are independent with law \(\mathcal Q_{i,\theta}\).
Consequently, \(\widehat g_{\Rsim}(\theta;D_i)\) in
\eqref{eq:mnp-one-observation-gradient} of the main text has the same law as
under accept--reject without a cap.
\end{lemma}

\begin{proof}
The pairs \((\nu^{(l)},A_l)\), \(1\leq l\leq K_0\), are independent and identically
distributed. Conditional on \((A_1,\ldots,A_{K_0})\), the
proposals are therefore independent, and each proposal with \(A_l=1\) has law
\(\mathcal Q_{i,\theta}\). The stopping index \(N\) and the count \(k\) are
functions of the indicators alone. Hence, conditional on the indicators, the
\(k\) retained draws are independent with law \(\mathcal Q_{i,\theta}\), and the
\(\Rsim-k\) draws from step~(2) are independent of them and have the same law.
The completed sample therefore has conditional law
\(\mathcal Q_{i,\theta}^{\otimes\Rsim}\), independently of the indicator
configuration.
\end{proof}

Step~(1) stops after at most \(K_0\) proposals.
For fixed \((D_i,\theta)\), the Botev sampler in the uncapped
step~(2) has positive acceptance probability. Obtaining the missing draws
therefore requires a sum of \(\Rsim-k\) independent geometric waiting times,
so step~(2) terminates almost surely. By
Lemma~\ref{lem:hybrid-exact}, the hybrid sampler is then an exact MNP
simulator, so Proposition~\ref{prop:mnp-empirical} of the main text covers it.

Our implementation limits step~(2) to \(2\times10^6\) proposals.
No zero-acceptance or partial-draw events occurred with the hybrid sampler
in the Monte Carlo experiments or empirical applications.

In our numerical comparison, \(K_0=8\Rsim J\).
Under correct specification, at \(\theta_*\), accept--reject
without a cap needs \(\Rsim/P_{iy_i}(\theta_*)\) proposals on average for
observation \(i\). Because every choice probability is positive
(Lemma~\ref{lem:mnp-accepted-moments}),
\[
    \mathbb E\{1/P_{iy_i}(\theta_*)\mid\DX_i\}
    =\sum_{j=1}^J P_{ij}(\theta_*)\,\frac{1}{P_{ij}(\theta_*)}
    =J,
\]
Thus, at the true parameter, accept--reject sampling requires
\(\Rsim J\) proposals per observation on average. The threshold
therefore allows eight times this expected cost before switching. The Botev
sampler solves a small optimization problem for each observation at which it
is invoked, so the hybrid sampler incurs that cost only where accept--reject
is slow. Section~\ref{sec:mc-simulator} compares the three samplers; the
empirical applications and the other Monte Carlo experiments use the hybrid sampler.

\section{Additional theory statements and proofs}\label{app:supp-theory}\label{app:proof-roadmap}

This appendix collects the additional theory statements and
proofs supporting Sections~\ref{sec:asymptotic-theory} and \ref{sec:inference} of the main text. It includes the
online sampling scheme, regularity conditions and limit theorem, the
localization calculation, the full smooth-transformation theorem, details on
burn-in and one-pass sampling, and the proofs of both limit theorems and
inference corollaries. The proof setup and roadmap are in
Section~\ref{sec:supp-proof-setup}; MNP verification is collected in
Appendix~\ref{app:mnp-verification}. Notation follows the main text.

\subsection{Online SAUSS and local exit control}\label{sec:supp-localization}\label{sec:online-sauss}

The online extension uses the recursion \eqref{eq:sa-update} and
Polyak--Ruppert average \eqref{eq:pr-average} of the main text, with each
update using \(m\) newly arriving population observations. In
Algorithm~\ref{alg:sauss-mnp}, replace the empirical sampling step by
independent draws \(D_{t1}^{\mathrm{update}},\ldots,D_{tm}^{\mathrm{update}}\)
from the population, independently of the pilot and previous history. The
conditional simulator, parameter update, and averaging steps are unchanged.
Throughout, \(m,\Rsim,\gamma_0,a\), and \(p\) are fixed, while
\(T,n_0\to\infty\) with no relative-rate restriction. Initialization and
learning rates satisfy Assumption~\ref{ass:algorithmic-setup} of the main text.

\begin{assumption}[Online sampling]\label{ass:online-sampling}
An initial block of \(n_0\) i.i.d.\ observations supplies the pilot. At each
update \(t\geq1\), \(m\) mutually independent new observations from the same
population arrive independently of the initial information and the previous
algorithmic history.
\end{assumption}

The pilot block and update observations may be consecutive segments of a
single i.i.d.\ stream; any pilot randomization is included in the initial
information, and the pilot estimator itself need not be computed online.

For online sampling, set
\(n\coloneqq mT\), the number of observations used after the pilot; the total
number observed is \(n_0+n\). The theorem below is stated using \(\sqrt T\);
the equivalent result using \(\sqrt n\) is given afterward. The learning-rate
clock starts after initialization.

Independent online sampling gives
\(\mathbb E\{\widehat G_t(\theta)\mid\mathcal F_{t-1}\}
=\mathcal R(\theta)\).
The population drift \(\mathcal R\) is defined in
Section~\ref{sec:sgd} of the main text; the curvature \(H_*\) and
one-update variance \(V_{m,\Rsim}\) are defined in
Section~\ref{sec:asymptotic-theory} of the main text.

\begin{assumption}[Online local regularity]\label{ass:online-regularity}
The following conditions hold on a neighborhood \(\mathcal N_*\) of
\(\theta_*\).
\begin{enumerate}
\item[(i)] The point \(\theta_*\) is the unique zero of \(\mathcal R\) in
\(\mathcal N_*\). The Jacobian \(H(\theta)\) is locally Lipschitz and
bounded, and \(H_*=H_*'\succ0\).
\item[(ii)] The exact observation-level negative score is continuously
differentiable and
\[
    \sup_{\theta\in\mathcal N_*}
    \mathbb E\|\nabla_\theta g_i(\theta)\|^2\leq C.
\]
\item[(iii)] For some \(\delta>0\) and
\(p_0>(1-a)^{-1}\), let
\(q\coloneqq\max\{2+\delta,2p_0\}\). Then
\[
    \sup_{\theta\in\mathcal N_*}
    \mathbb E\|g_{ir}(\theta)\|^q\leq C,
\]
where the expectation averages over the population observation and its
accepted conditional draw.
\item[(iv)] For all \(\vartheta_1,\vartheta_2\in\mathcal N_*\),
\[
\left\|\mathbb E[\operatorname{Var}\{g_{ir}(\vartheta_1)\mid D_i\}]
       -\mathbb E[\operatorname{Var}\{g_{ir}(\vartheta_2)\mid D_i\}]\right\| \leq
       C\|\vartheta_1-\vartheta_2\|.
\]
\item[(v)] The matrix \(V_{m,\Rsim}\) in \eqref{eq:vmr} is positive
definite.
\end{enumerate}
\end{assumption}

Conditional unbiasedness is the maintained simulator setup from
Sections~\ref{sec:sgd}--\ref{sec:mnp-sauss}. The moment condition for one accepted
draw also supplies the corresponding bounds for \(g_i(\theta)\) and, because
\(m\) is fixed, for the mini-batch direction. Individual simulated
scores may remain discontinuous in \(\theta\).

The pilot-consistency condition reflects the local nature of the analysis: it
ensures that the initialized iterate lies in any fixed neighborhood of
\(\theta_*\) with probability approaching one. No convergence rate for the
pilot is imposed. Under multi-pass sampling, the pilot may use the full
data set or a subset. Let \(\widehat\theta_n\) denote the full-sample estimator
satisfying \(\mathcal R_n(\widehat\theta_n)=0\), as introduced in
Section~\ref{sec:sgd} of the main text. If both \(\theta_0\) and \(\widehat\theta_n\) are
consistent for \(\theta_*\), then
\(\|\theta_0-\widehat\theta_n\|={ o_{\mathbb P}}(1)\) without an additional relative-rate
condition. Constructing the pilot remains model dependent.

The learning rates decrease while satisfying
\(\sum_{t\geq1}\gamma_t=\infty\) and
\(\sum_{t\geq1}\gamma_t^2<\infty\). Their cumulative magnitude remains
unbounded, while square summability controls the accumulated noise under
the local moment bounds below. The schedule also gives
\(\sum_{t\geq1}\gamma_t/\sqrt{t}<\infty\).

Assumption~\ref{ass:online-regularity}(i) implies that, for a sufficiently
small ball
\(B_\rho\coloneqq\{\theta:\|\theta-\theta_*\|\leq\rho\}
\subset\mathcal N_*\),
\[
    (\theta-\theta_*)'\mathcal R(\theta)
    \geq c_0\|\theta-\theta_*\|^2,
    \qquad
    \|\mathcal R(\theta)-H_*(\theta-\theta_*)\|
    \leq C\|\theta-\theta_*\|^2
\]
for some \(c_0>0\).

The general theory does not invoke the likelihood information identity. Under
correct likelihood specification,
\(\operatorname{Var}\{g_i(\theta_*)\}=H_*\), so
\[
    V_{m,\Rsim}
    =\frac1mH_*
     +\frac1{m\Rsim}
      \mathbb E[\operatorname{Var}\{g_{ir}(\theta_*)\mid D_i\}]
    \succeq\frac1mH_*\succ0.
\]
Proposition~\ref{prop:mnp-empirical} of the main text verifies this specialization for MNP.

We now make precise the local \(\epsilon\)-formulation used below. Define
\begin{equation}\label{eq:online-local-event}
    \tau\coloneqq\inf\{t\geq0:\theta_t\notin B_\rho\},
    \qquad
    \mathcal E\coloneqq\{\tau=\infty\}.
\end{equation}
The parameter \(\epsilon\) introduced below controls the probability that the
recursion exits \(B_\rho\); it arises from the stopped-process argument. For any
\(b\in(0,\rho)\), the supermartingale maximal inequality used in the proof
gives
\[
    \mathbb P(\mathcal E^c)
    \leq
    \mathbb P(\|\theta_0-\theta_*\|>b)
    +\frac{b^2+C\gamma_0^2\sum_{s\geq1}s^{-2a}}{\rho^2}.
\]
For a fixed \(\epsilon\in(0,1)\), set
\(b_\epsilon\coloneqq\rho\sqrt{\epsilon/2}\) and choose \(\gamma_0\)
sufficiently small that
\[
    C\gamma_0^2\sum_{s\geq1}s^{-2a}
    \leq\frac{\rho^2\epsilon}{2}.
\]
Pilot consistency then gives
\(\mathbb P(\mathcal E^c)\leq\epsilon+o(1)\), and consistency statements
below hold on \(\mathcal E\).

\begin{definition}[Local \(\epsilon\)-weak convergence]\label{def:supp-local-epsilon}
For random elements \(Z_k\) and \(Z\) taking values in a
metric space \(\mathbb S\), where \(k\) indexes the relevant asymptotic
sequence, write
\[
    Z_k\overset{\epsilon}{\Rightarrow}Z
\]
if, for every bounded continuous \(f:\mathbb S\to\mathbb R\),
\[
    \limsup_{k\to\infty}
    \big|\mathbb Ef(Z_k)-\mathbb Ef(Z)\big|
    \leq2\epsilon\|f\|_\infty.
\]
\end{definition}

The admissible \(\gamma_0\) depends on \(\epsilon\); ordinary weak
convergence for one fixed learning-rate constant would require additional
control outside \(B_\rho\). For the infinite-horizon exit event, the
recursion is continued beyond the reporting horizon \(T\).

We state the main result for \(T_0=0\), so that
\(\bar\theta_T=T^{-1}\sum_{t=1}^T\theta_t\).
Remark~\ref{rem:burnin} extends the result to deterministic burn-in satisfying
\(T_0/T\to0\).

\begin{theorem}[Consistency, asymptotic normality, and FCLT for online SAUSS]\label{thm:sa-linear}\label{thcflc}
Fix finite \(m\geq1\) and \(\Rsim\geq1\), and suppose
Assumptions~\ref{ass:algorithmic-setup}, \ref{ass:online-sampling}, and
\ref{ass:online-regularity} hold. For \(r\in[0,1]\), define
\begin{align*}
    \mathbb X_T(r)
    &\coloneqq
    \frac1{\sqrt T}\sum_{t=1}^{\lfloor Tr\rfloor}
       (\theta_t-\theta_*),\\
    \mathbb X(r)
    &\coloneqq
    \bigl(H_*^{-1}V_{m,\Rsim}H_*^{-1}\bigr)^{1/2}\mathbb W(r),
\end{align*}
where \(\mathbb W\) is a \(p\)-dimensional standard Wiener process. With
\(\mathcal E\) as defined in \eqref{eq:online-local-event}, for every
\(\epsilon\in(0,1)\), there is \(\bar\gamma(\epsilon)>0\) such that every
fixed \(0<\gamma_0\leq\bar\gamma(\epsilon)\) gives
\(\mathbb P(\mathcal E^c)\leq\epsilon+o(1)\) and the following conclusions:
\begin{enumerate}
\item[(i)] \(\|\theta_T-\theta_*\|\mathbb I\{\mathcal E\}
\xrightarrow{p}0\) as \(T\to\infty\).
\item[(ii)]
\[
    \sqrt T(\bar\theta_T-\theta_*)
    \overset{\epsilon}{\Rightarrow}
    \mathcal N\!\left(0,H_*^{-1}V_{m,\Rsim}H_*^{-1}\right).
\]
\item[(iii)] \(\mathbb X_T\overset{\epsilon}{\Rightarrow}\mathbb X\) in
\(D[0,1]^p\), equipped with the usual Skorokhod topology.
\end{enumerate}
\end{theorem}

When \(T_0=0\), the endpoint identity
\(\sqrt T(\bar\theta_T-\theta_*)=\mathbb X_T(1)\) is exact. Since
\(n=mT\) and \(V_{m,\Rsim}=V_{1,\Rsim}/m\),
\[
    \sqrt n(\bar\theta_T-\theta_*)
    \overset{\epsilon}{\Rightarrow}
    \mathcal N\!\left(
       0,H_*^{-1}V_{1,\Rsim}H_*^{-1}
    \right).
\]
Thus, for any fixed \(m\), \(\sqrt n(\bar\theta_T-\theta_*)\) has the same
limiting distribution as in the case \(m=1\), while the algorithm uses
\(T=n/m\) updates. The variance matrix in this limit combines variability
across population observations and variability from conditional simulation;
the pilot observations are not included in \(n\).

\subsection{Smooth transformations}\label{sec:supp-transformation}

The assumptions and the local convergence notion below are those of
Theorem~\ref{thm:sa-linear} and Definition~\ref{def:supp-local-epsilon}.
\begin{theorem}[Smooth transformations]\label{thm:supp-transformation}
Under the conditions of Theorem~\ref{thm:sa-linear}, let
\(h:\Theta\to\mathbb R\) be twice continuously differentiable on
\(\mathcal N_*\) and satisfy
\[
    \sup_{\theta\in\mathcal N_*}
    \|\nabla_\theta^2h(\theta)\|\leq C,
    \qquad
    \nabla_\theta h(\theta_*)\ne0.
\]
Then
\[
    \frac1{\sqrt T}\sum_{t=1}^{\lfloor Tr\rfloor}
       \{h(\theta_t)-h(\theta_*)\}
    \overset{\epsilon}{\Rightarrow}
    \nabla_\theta h(\theta_*)'
    \bigl(H_*^{-1}V_{m,\Rsim}H_*^{-1}\bigr)^{1/2}\mathbb W(r),
    \qquad r\in[0,1],
\]
in \(D[0,1]\), equipped with the usual Skorokhod topology.
\end{theorem}

For the proof, write \(\Delta_t=\theta_t-\theta_*\). The continued
martingale process \(\widetilde M_T\) and the stopped-rate and averaging
approximations are defined and established in
Sections~\ref{app:main-proofs} and \ref{app:supporting-lemmas}.
\begin{proof}[Proof of Theorem~\ref{thm:supp-transformation}]
Let $\tau$, $u_t$, and $\mathcal E=\{\tau=\infty\}$ be as in Lemma~\ref{lem:stopped-bound}. On $\mathcal E$, $\theta_t\in B_\rho\subset\mathcal N_*$ for every $t$. Because $B_\rho$ is convex, the line segment joining $\theta_t$ and $\theta_*$ also lies in $B_\rho$ on this event. Hence, for some $\tilde{\theta}_t$ on that line segment,
$$
\begin{aligned}
&\left|\frac{1}{\sqrt{T}}\sum_{t=1}^{\lfloor Tr\rfloor}\{h(\theta_t)-h(\theta_*)\}- \frac{1}{\sqrt{T}}\sum_{t=1}^{\lfloor Tr\rfloor}\nabla_\theta h(\theta_*)'\Delta_t\right|\mathbb I\{\mathcal E\} \\
&\qquad\leq \frac{1}{2\sqrt{T}}\sum_{t=1}^{\lfloor Tr\rfloor}\|\Delta_t\|^2\|\nabla^2 h(\widetilde\theta_t)\|\mathbb I\{\mathcal E\}
\leq \frac{C}{\sqrt{T}}\sum_{t=1}^{\lfloor Tr\rfloor}\|\Delta_t\|^2\mathbb I\{\tau>t\}.
\end{aligned}
$$
by the smoothness condition in Theorem~\ref{thm:supp-transformation}. By Lemma~\ref{lem:rateconverg},
\begin{align*}
&\mathbb E\left( \sup_{r\in[0,1]}\frac{C}{\sqrt{T}}\sum_{t=1}^{\lfloor Tr\rfloor}\|\Delta_t\|^2\mathbb I\{\tau>t\}\right)\\
&\qquad\leq \frac{C}{\sqrt{T}}\sum_{t=1}^{\infty}\exp(-C_6t^{1-a})
+\frac{C}{\sqrt{T}}\sum_{t=1}^{T}t^{-a}
= O(T^{-1/2}+T^{1/2-a}).
\end{align*}
Also, by the PR linearization and averaging approximations in
Section~\ref{app:main-proofs},
\[
\sup_{0\leq r\leq1}
\left|
\nabla_\theta h(\theta_*)'
\left\{
  \frac1{\sqrt T}\sum_{t=1}^{\lfloor Tr\rfloor}\Delta_t
  +H_*^{-1}\widetilde M_T(r)
\right\}
\right|\mathbb I\{\mathcal E\}=o_{\mathbb P}(1).
\]
\par
From part (ii) of Lemma~\ref{lem:stopped-bound}, $\mathbb P(\mathcal E^c)\leq\epsilon+o(1)$. This gives the claimed weak convergence in the same $\epsilon$-sense.
\end{proof}

\subsection{Burn-in and one-pass sampling}\label{sec:supp-burnin-onepass}
\begin{remark}[Burn-in]\label{rem:burnin}
For the averaged estimator, let \(T_0=T_0(T)\) be deterministic
and satisfy \(T_0/T\to0\). The functional limit theorem then gives
\(\{\sum_{t=1}^{T_0}(\theta_t-\theta_*)\}\mathbb I\{\mathcal E\}
=o_{\mathbb P}(\sqrt T)\), and \(T/(T-T_0)\to1\), so the endpoint and
functional limits remain unchanged in the same \(\epsilon\)-sense as
Theorem~\ref{thm:sa-linear}.

\end{remark}

\begin{remark}[One pass without replacement]\label{rem:one-pass-bridge}
The online result also describes a finite i.i.d.\ analysis sample
used exactly once, provided the pilot is independent of that analysis sample.
Assume for simplicity that \(m\) divides \(n\). If the observations are
randomly permuted, divided into disjoint batches of size \(m\), and revealed
one batch per update, then \(T=n/m\). Under an independent
random permutation, the sequence of disjoint mini-batches has the same
unconditional joint distribution as \(T=n/m\) independent online
mini-batches. Hence Theorem~\ref{thm:sa-linear} applies to this one-pass
recursion with \(T=n/m\).
This is different from Assumption~\ref{ass:empirical-sampling} of the main text, which samples
with replacement and can omit or repeat observations even when its expected
pass count is one.

\end{remark}

\subsection{Proof setup and notation}\label{sec:supp-proof-setup}

For the inference statements, quantities involving \({ \widehat H}^{-1}\)
are defined when \({ \widehat H}\) is nonsingular, and Wald or random-scaling
statistics are defined when the indicated restriction variance or
self-normalizer is positive definite. They may be assigned arbitrary values
otherwise; under the stated conditions, these exceptional cases are
asymptotically negligible on the corresponding no-exit event.

For the proofs of Theorem~\ref{thm:sa-linear} and
Lemmas~\ref{lem1}--\ref{lem:rateconverg}, fix the common finite mini-batch
size \(m\). Writing \(D_{tb}\) for the fresh observation in
batch position \(b\) at update \(t\) and \(g_{tb,r}(\theta)\) for the
simulated score based on its \(r\)-th accepted conditional draw, the fixed-\(m\) version of the
direction used in \eqref{eq:sa-update} of the main text is
\begin{equation}\label{equpdate}
\widehat G_t(\theta)=\frac{1}{m\Rsim}
\sum_{b=1}^{m}\sum_{r=1}^{\Rsim}
g_{tb,r}(\theta),
\qquad
\theta_t = \theta_{t-1} -\gamma_t \widehat G_t(\theta_{t-1}),
\qquad t\geq1,
\end{equation}
Define locally
\(G_t(\theta)\coloneqq\mathbb E_t^*\widehat G_t(\theta)\), the expectation
over simulation draws given \(\mathcal F_{t-1}\) and the selected mini-batch.
Then \(\mathcal R(\theta)=\mathbb E\{G_t(\theta)\mid\mathcal F_{t-1}\}\),
\(H(\theta)=\nabla_\theta\mathcal R(\theta)\),
\(H_*=H(\theta_*)\), and \(\mathcal R(\theta_*)=0\). The filtration
\(\mathcal F_0\) contains the pilot block and its randomization.
Assumption~\ref{ass:online-sampling} makes the conditional drift independent
of \(t\).

For these proofs, Assumption~\ref{ass:algorithmic-setup} of the main text and
Assumptions~\ref{ass:online-sampling}--\ref{ass:online-regularity} hold for fixed
\(m\). In particular, \(q=\max\{2+\delta,2p_0\}\), with
\(p_0>(1-a)^{-1}\), is the moment order in
Assumption~\ref{ass:online-regularity}(iii).

As in the main text, all statements indexed by the run
length are taken along a sequence in which $T\to\infty$ and
$n_0=n_0(T)\to\infty$, with no relative-rate restriction; \(p\),
\(\Rsim\), \(\gamma_0\), and \(a\) are held fixed.

To display the main argument before its technical components,
Section~\ref{app:main-proofs} gives the proofs of
Theorem~\ref{thm:sa-linear} and
Corollary~\ref{cor:inference}; Section~\ref{app:supporting-lemmas}
then states and proves Lemmas~\ref{lem1}--\ref{lem:rateconverg},
which those proofs invoke.
Section~\ref{app:multipass-proof} proves the additional
multi-pass results, conditionally on the data set before taking the joint
limit.

\subsection{Proofs of the online results}\label{app:main-proofs}

\begin{proof}[Proof of Theorem~\ref{thm:sa-linear}]
For $t\geq1$, we have
\begin{equation}
    \theta_t=\theta_{t-1}-\gamma_t\widehat G_t(\theta_{t-1}),
\end{equation}

For $t\geq1$, define
\begin{equation}\label{eqzeta}
    \zeta_t\coloneqq \widehat G_t(\theta_{t-1})-G_t(\theta_{t-1}).
\end{equation}
The term $\zeta_t$ is the simulation component of the martingale-difference innovation. In addition, let
 \begin{equation}\label{eqxi}
     \xi_t \coloneqq G_t(\theta_{t-1})- \mathbb E(G_t(\theta_{t-1})|\mathcal F_{t-1}).
 \end{equation}
The term $\xi_t$ is the observation-sampling component of the martingale-difference innovation.

Recall that $H_*=\nabla_\theta\mathcal R(\theta_*)$.

Let
\begin{equation}\label{eqeta}
\eta_t\coloneqq \mathbb E(G_t(\theta_{t-1})|\mathcal F_{t-1})- H_*(\theta_{t-1}-\theta_*).
\end{equation}
The term $\eta_t$ is the nonlinear remainder from linearizing the mean drift around $\theta_*$.

From \eqref{eqzeta}--\eqref{eqeta}, we have
\begin{equation}
    \theta_{t} = \theta_{t-1} -\gamma_t H_*(\theta_{t-1}-\theta_*)-\gamma_t(\zeta_t+\xi_t+\eta_t).
\end{equation}
 Let
 $$
 \Delta_t\coloneqq \theta_t-\theta_*.
 $$
 The innovation \(\zeta_t+\xi_t\) is a martingale-difference sequence (MDS). Then
\begin{equation}\label{eq8}
    \Delta_{t} =\Delta_{t-1} -\gamma_t H_*\Delta_{t-1}-\gamma_t(\underbrace{\zeta_t+\xi_t}_{\text{{ MDS-FCLT}}}+\underbrace{\eta_t}_{\text{nonlinearity}}).
\end{equation}

Let $\tau$, $\rho$, and
$\mathcal E=\{\tau=\infty\}$ be as in Lemma~\ref{lem1}; by
Lemma~\ref{lem1}(ii),
$\liminf_T\mathbb P(\mathcal E)\geq1-\epsilon$ once
$\gamma_0\leq\bar\gamma(\epsilon)$.

Let $\widetilde\varepsilon_t$ be the continued innovation defined in Lemma~\ref{lemfclt}. Set $\widetilde\Delta_0^1\coloneqq \Delta_0$ and, for $t\geq1$, define the auxiliary linear recursion
\begin{equation}\label{eq:continued-linear-recursion}
    \widetilde\Delta_t^1
    \coloneqq
    (I-\gamma_tH_*)\widetilde\Delta_{t-1}^1
    -\gamma_t\widetilde\varepsilon_t.
\end{equation}

The claimed statements follow from the three results below. The approximation remainders are established on $\mathcal E$: they are of the form $r_T\mathbb I\{\mathcal E\}$ with $r_T\mathbb I\{\mathcal E\}\xrightarrow{p}0$.

  Lemma~\ref{lemlinear}:
  \begin{equation}\label{eq39}
{\sup_{r\in[0,1]}\left\|\frac{1}{\sqrt{T}}\sum_{t=1}^{\lfloor Tr\rfloor}(\Delta_t-\widetilde\Delta_t^1)\right\|\mathbb I\{\mathcal E\}=o_{\mathbb P}(1)}
  \end{equation}

  Lemma~\ref{lemfclt}: for $\widetilde M_T(r)$ defined in   Lemma~\ref{lemfclt},
$$
{-H_*^{-1}\widetilde M_T(r)\Rightarrow (H_*^{-1}V_{m,\Rsim}H_*^{-1})^{1/2}\mathbb W(r)}
$$

  Lemma~\ref{lemleading}:
  \begin{equation}\label{eq40}
{\sup_{r\in[0,1]}\left\|\frac{1}{\sqrt{T}}\sum_{t=1}^{\lfloor Tr\rfloor}\widetilde\Delta_t^1+H_*^{-1}\widetilde M_T(r)\right\|=o_{\mathbb P}(1).}
  \end{equation}

Let $\widetilde{\mathbb X}_T(r)\coloneqq -H_*^{-1}\widetilde M_T(r)$. By Lemma~\ref{lemfclt}, $\widetilde{\mathbb X}_T\Rightarrow \mathbb X$. On $\mathcal E$, the continued innovations used in $\widetilde M_T$ equal the actual innovations $\zeta_t+\xi_t$ at every iteration. Lemmas~\ref{lemlinear} and \ref{lemleading} therefore give
\[
    \sup_{r\in[0,1]}\|\mathbb X_T(r)-\widetilde{\mathbb X}_T(r)\|\mathbb I\{\mathcal E\}=o_{\mathbb P}(1).
\]
The weak convergence of
$\widetilde{\mathbb X}_T$ implies its tightness. Together with the preceding
coupling on $\mathcal E$, this tightness implies that, for every bounded
continuous $f$ on $D[0,1]^p$,
$\mathbb E\{|f(\mathbb X_T)-f(\widetilde{\mathbb X}_T)|
\mathbb I\{\mathcal E\}\}\to0$.
Consequently, for every bounded continuous $f$ on $D[0,1]^p$,
\[
\begin{aligned}
|\mathbb Ef(\mathbb X_T)-\mathbb Ef(\widetilde{\mathbb X}_T)|
&\leq \mathbb E\{|f(\mathbb X_T)-f(\widetilde{\mathbb X}_T)|\mathbb I\{\mathcal E\}\}
+2{ \|f\|_\infty}\mathbb P(\mathcal E^c),
\end{aligned}
\]
whose $\limsup$ is at most $2\epsilon{ \|f\|_\infty}$. This proves the stated FCLT in the $\epsilon$-sense and its endpoint normal approximation. Consistency follows from Lemmas~\ref{lem1} and \ref{lem:rateconverg}.
\end{proof}

\begin{proof}[Proof of Corollary~\ref{cor:inference}]
Let
\[
    \Sigma_{\mathrm{alg}}\coloneqq H_*^{-1}V_{m,\Rsim}H_*^{-1},
    \qquad
    \Sigma_A\coloneqq A\Sigma_{\mathrm{alg}}A'.
\]
The matrix \(\Sigma_A\) is positive definite because \(A\) has full row rank
and \(\Sigma_{\mathrm{alg}}\) is positive definite. The FCLT in
Theorem~\ref{thm:sa-linear} gives
\[
    A\mathbb X_T(\cdot)\overset{\epsilon}{\Rightarrow}
    \Sigma_A^{1/2}\mathbb W_\ell(\cdot)
\]
in the same \(\epsilon\)-sense as that theorem. For every \(t=1,\ldots,T\),
\[
    \frac{1}{\sqrt T}\sum_{s=1}^t A(\theta_s-\bar\theta_T)
    =
    A\mathbb X_T(t/T)-\frac{t}{T}A\mathbb X_T(1).
\]
Because the limiting process has continuous paths, the continuous mapping
theorem and a Riemann-sum argument yield
\[
    \widehat V_{\mathrm{RS}}(A)
    \overset{\epsilon}{\Rightarrow}
    \Sigma_A^{1/2}
    \left\{
        \int_0^1\bar{\mathbb W}_\ell(r)\bar{\mathbb W}_\ell(r)'dr
    \right\}
    \Sigma_A^{1/2}.
\]
Under \(H_0\), $a_0= A\theta_*.$
\[
    \sqrt T(A\bar\theta_T-a_0)=A\mathbb X_T(1)
    \overset{\epsilon}{\Rightarrow}
    \Sigma_A^{1/2}\mathbb W_\ell(1).
\]
Also, \(\|\widehat\Sigma_{\mathrm{alg}}-\Sigma_{\mathrm{alg}}\|
\mathbb I\{\mathcal E\}\xrightarrow{p}0\).
Apply these mappings to the auxiliary process
\(\widetilde{\mathbb X}_T\) in the proof of
Theorem~\ref{thm:sa-linear}. For its Wald statistic, use variance
\(\widehat\Sigma_{\mathrm{alg}}\) on \(\mathcal E\) when it is
positive definite, and \(\Sigma_{\mathrm{alg}}\) otherwise.
This variance converges in probability.\par The auxiliary Wald and
random-scaling statistics have the stated ordinary limits; the
Brownian-bridge matrix is positive definite almost surely.
Their differences from the actual statistics are \({ o_{\mathbb P}}(1)\) on
\(\mathcal E\). Continuity of the limiting laws at the critical
values bounds the differences in rejection probabilities by
\(\mathbb P(\mathcal E^c)+o(1)\).\par

\end{proof}

\subsection{Supporting lemmas}\label{app:supporting-lemmas}

\begin{lemma}\label{lem1}\label{lem:stopped-bound}
Suppose Assumption~\ref{ass:algorithmic-setup} of the main text and
Assumptions~\ref{ass:online-sampling}--\ref{ass:online-regularity} hold for the fixed
finite $m$. There exist $\rho>0$ and $c_0>0$ such that, with
 $$
 \tau\coloneqq \inf\{t\geq0: \|\Delta_t\|>\rho\},
 \qquad u_t\coloneqq \mathbb E\left(\|\Delta_t\|^2 \mathbb I\{\tau>t\}\right),
 $$
the following hold for all $\gamma_0\leq\bar\gamma$, where $\bar\gamma>0$ depends only on $\rho$, $c_0$, and the constants in these assumptions:
\begin{enumerate}
\item[(i)] for all $t\geq1$,
\begin{equation}\label{eq81}
u_t\leq (1-2c_0\gamma_t+C\gamma_t^2)\,u_{t-1}+C\gamma_t^2 ;
\end{equation}
\item[(ii)] for every $\epsilon\in(0,1)$, if in addition $\gamma_0\leq\bar\gamma(\epsilon)$, then
$$
{\limsup_{T\to\infty}\mathbb P(\tau<\infty)\leq\epsilon} ,
$$
that is, $\{\|\Delta_t\|\leq\rho\ \forall t\geq0\}$ holds with probability at least $1-\epsilon-o(1)$;
\item[(iii)] $u_0\leq C$, and, for all $t\geq1$,
\begin{equation}\label{eq16}
\mathbb E\left(\|\Delta_t\|^2 \mathbb I\{\tau>t\}\right)\leq C\gamma_t .
\end{equation}
\end{enumerate}
\end{lemma}

\begin{proof}
\textit{Step 1: local drift and remainder bounds.}
Since $\mathcal R(\theta_*)=0$ and $H(\cdot)$ is Lipschitz with constant $L$ on a neighborhood of $\theta_*$,
$$
\mathcal R(\theta)=\int_0^1H\big(\theta_*+s(\theta-\theta_*)\big)\,ds\,(\theta-\theta_*)
$$
for all $\theta$ in that neighborhood. Because $H_*$ is symmetric with $\lambda_{\min}(H_*)>c$, choose $\rho>0$ so small that $B_\rho=\{\theta:\|\theta-\theta_*\|\leq\rho\}$ lies in the neighborhood and $L\rho\leq c/2$. Then for all $\theta\in B_\rho$,
\begin{equation}\label{lowerbound}
(\theta-\theta_*)'\mathcal R(\theta)\geq (c-L\rho)\|\theta-\theta_*\|^2\geq c_0\|\theta-\theta_*\|^2,\qquad c_0\coloneqq c/2,
\end{equation}
and, using $\|H(\theta)\|\leq C$ on $B_\rho$,
\begin{equation}\label{lowerbound2}
\|\mathcal R(\theta)\|\leq C\|\theta-\theta_*\|,\qquad
\|\mathcal R(\theta)-H_*(\theta-\theta_*)\|\leq \tfrac{L}{2}\|\theta-\theta_*\|^2 .
\end{equation}
The second inequality in \eqref{lowerbound2} is the bound $\|\eta_t\|\leq C\|\Delta_{t-1}\|^2$ used below whenever $\theta_{t-1}\in B_\rho$.

\textit{Step 2: one-step inequality.}
Recall \eqref{eq8}: $\Delta_t=\Delta_{t-1}-\gamma_t\{\mathcal R(\theta_{t-1})+\zeta_t+\xi_t\}$, with $\mathbb E(\zeta_t+\xi_t|\mathcal F_{t-1})=0$; here $\zeta_t+\xi_t=\widehat G_t(\theta_{t-1})-\mathcal R(\theta_{t-1})$, and the conditional mean is zero because the mini-batch contains fresh population draws and each accepted-draw average is conditionally unbiased at the $\mathcal F_{t-1}$-measurable parameter $\theta_{t-1}$. Squaring,
$$
\mathbb E(\|\Delta_t\|^2|\mathcal F_{t-1})
=\|\Delta_{t-1}\|^2-2\gamma_t\Delta_{t-1}'\mathcal R(\theta_{t-1})
+\gamma_t^2\,\mathbb E\big(\|\mathcal R(\theta_{t-1})+\zeta_t+\xi_t\|^2\big|\mathcal F_{t-1}\big).
$$
By Assumption~\ref{ass:online-regularity}(iii) (with $q\geq2$) and the substitution rule for conditional expectations,
\begin{align*}
\mathbb E(\|\zeta_t+\xi_t\|^2\mid\mathcal F_{t-1})
&\leq
2\mathbb E\left(
\left\|
  \widehat G_t(\theta_{t-1})
\right\|^2
\mathrel{}\middle|\mathrel{}
\mathcal F_{t-1}
\right)
+2\|\mathcal R(\theta_{t-1})\|^2 \leq  C+C\|\Delta_{t-1}\|^2
\end{align*}
on $\{\theta_{t-1}\in B_\rho\}$, using \eqref{lowerbound2}. Hence, on the event $\{\tau>t-1\}=\{\theta_s\in B_\rho,\ s\leq t-1\}$, \eqref{lowerbound} gives
\begin{equation}\label{eqbounmom}
\mathbb E(\|\Delta_t\|^2|\mathcal F_{t-1})\leq \|\Delta_{t-1}\|^2(1-2c_0\gamma_t+C\gamma_t^2)+C\gamma_t^2 \qquad\text{on }\{\tau>t-1\}.
\end{equation}
The event $\{\tau>t-1\}$ is $\mathcal F_{t-1}$-measurable. Multiplying \eqref{eqbounmom} by $\mathbb I\{\tau>t-1\}$, using $\mathbb I\{\tau>t\}\leq\mathbb I\{\tau>t-1\}$ and taking expectations yields \eqref{eq81}, since $\mathbb E(\|\Delta_t\|^2\mathbb I\{\tau>t\})\leq \mathbb E\{\mathbb I\{\tau>t-1\}\mathbb E(\|\Delta_t\|^2|\mathcal F_{t-1})\}$ and $\mathbb P(\tau>t-1)\leq1$. This proves (i).

\textit{Step 3: the iterates stay in $B_\rho$ with high probability.}
Let $V_t\coloneqq \|\Delta_{t\wedge\tau}\|^2$ and take $\bar\gamma$ such that $C\bar\gamma\leq 2c_0$, so that $1-2c_0\gamma_t+C\gamma_t^2\leq1$ for all $t\geq1$ when $\gamma_0\leq\bar\gamma$. On $\{\tau>t-1\}$, $V_t=\|\Delta_t\|^2$ and \eqref{eqbounmom} gives $\mathbb E(V_t|\mathcal F_{t-1})\leq V_{t-1}+C\gamma_t^2$; on $\{\tau\leq t-1\}$, $V_t=V_{t-1}$.
Fix $\epsilon\in(0,1)$ and let
$B\coloneqq\{\|\Delta_0\|\leq\rho\sqrt{\epsilon/2}\}\in\mathcal F_0$.
Pilot consistency gives $\mathbb P(B)\to1$ as $n_0=n_0(T)\to\infty$.
Define
\[
N_t\coloneqq\mathbb I\{B\}
\left(V_t+C\sum_{s>t}\gamma_s^2\right).
\]
Because $B\in\mathcal F_0$, the preceding conditional inequality gives
$\mathbb E(N_t\mid\mathcal F_{t-1})\leq N_{t-1}$.
Moreover, $N_0\leq\rho^2\epsilon/2+C\gamma_0^2\sum_{s\geq1}s^{-2a}
\leq\rho^2\epsilon$ once $\gamma_0\leq\bar\gamma(\epsilon)$.
This bounded initial value ensures integrability at every time, so $N_t$
is a nonnegative supermartingale. Since
$B\cap\{\tau<\infty\}\subset\{\sup_tN_t\geq\rho^2\}$,
Ville's maximal inequality gives
\[
\mathbb P(B\cap\{\tau<\infty\})
\leq\mathbb P\{\sup_tN_t\geq\rho^2\}
\leq\frac{\mathbb E N_0}{\rho^2}\leq\epsilon,
\]
so $\mathbb P(\tau<\infty)\leq\epsilon+\mathbb P(B^c)=\epsilon+o(1)$. This proves (ii).

\textit{Step 4.} Part (iii) follows from (i) and Lemma~\ref{lem:rateconverg} below, whose proof uses only the recursion \eqref{eq81}, since $\exp(-C_6t^{1-a})\leq Ct^{-a}$.
\end{proof}

\begin{lemma}\label{lemlinear}\label{lem:pr-linearization}
Let $\mathcal E=\{\tau=\infty\}$ be as in Lemma~\ref{lem1}, and let $\widetilde\Delta_t^1$ satisfy \eqref{eq:continued-linear-recursion}. Then
\[
\sup_{r\in[0,1]}
\left\|
\frac{1}{\sqrt T}
\sum_{t=1}^{\lfloor Tr\rfloor}
(\Delta_t-\widetilde\Delta_t^1)
\right\|
\mathbb I\{\mathcal E\}
=o_{\mathbb P}(1).
\]
\end{lemma}

\begin{proof}
For integers $j\leq t$, define the deterministic transition matrices
\[
    \Phi_{j,t}\coloneqq \prod_{\ell=j}^{t}(I-\gamma_\ell H_*),
    \qquad
    \Phi_{j,t}\coloneqq I\quad\text{when }j>t.
\]
For $k\geq1$, set
\begin{equation}\label{eq:pr-finite-weights}
\begin{aligned}
B_k \coloneqq
\left(\sum_{t=1}^k\Phi_{1,t}\right)\Delta_0, \quad
K_{j,k}
&\coloneqq
\gamma_j\sum_{t=j}^k\Phi_{j+1,t},
\quad
w_j^k\coloneqq K_{j,k}-H_*^{-1},
\qquad 1\leq j\leq k,
\end{aligned}
\end{equation}
and set $B_0\coloneqq 0$. The standard product bounds for $\gamma_t=\gamma_0t^{-a}$ and positive definite $H_*$ give
\begin{equation}\label{eq:pr-weight-bounds}
    \sup_{k\geq1}\|B_k\|\leq C\|\Delta_0\|,
    \qquad
    \sup_{1\leq j\leq k}\|w_j^k\|\leq C,
    \qquad
    \sum_{j=1}^k\|w_j^k\|\leq Ck^a.
\end{equation}
These are the usual Polyak--Ruppert weight bounds; see Lemma~2 of \cite{polyak1992acceleration} and Lemma~2 of \cite{zhu2021constructing}.

On $\mathcal E$, Lemma~\ref{lemfclt} gives $\widetilde\varepsilon_j=\zeta_j+\xi_j$ for every $j$. Iterating \eqref{eq8} and \eqref{eq:continued-linear-recursion}, which have the same initial condition, therefore yields the exact finite-horizon identity
\begin{equation}\label{eq:pr-nonlinear-difference}
    \sum_{t=1}^k(\Delta_t-\widetilde\Delta_t^1)
    =
    -\sum_{j=1}^kK_{j,k}\eta_j
    =
    -\sum_{j=1}^k(H_*^{-1}+w_j^k)\eta_j
    \qquad\text{on }\mathcal E.
\end{equation}
On $\{\tau>j-1\}$, \eqref{lowerbound2} gives $\|\eta_j\|\leq C\|\Delta_{j-1}\|^2$. Hence \eqref{eq:pr-weight-bounds} and \eqref{eq:pr-nonlinear-difference} imply
\begin{equation}\label{eq:pr-nonlinear-bound}
\sup_{r\in[0,1]}
\frac{1}{\sqrt T}
\left\|
\sum_{t=1}^{\lfloor Tr\rfloor}(\Delta_t-\widetilde\Delta_t^1)
\right\|
\mathbb I\{\mathcal E\}
\leq
\frac{C}{\sqrt T}
\sum_{j=1}^{T}
\|\Delta_{j-1}\|^2\mathbb I\{\tau>j-1\}.
\end{equation}
Because the initialization may depend on $T$, we use a
finite-horizon probability bound. By Lemma~\ref{lem1}(iii),
\[
\begin{aligned}
&\mathbb E\left[
\frac{C}{\sqrt T}
\sum_{j=1}^{T}
\|\Delta_{j-1}\|^2
\mathbb I\{\tau>j-1\}
\right]\\
&\qquad\leq
\frac{C}{\sqrt T}
\left(
u_0+C\sum_{j=2}^{T}\gamma_{j-1}
\right)
=
O(T^{-1/2})+O(T^{1/2-a})
=o(1),
\end{aligned}
\]
where the last equality uses
$\sum_{j=2}^{T}\gamma_{j-1}=O(T^{1-a})$ and $a>1/2$.
Markov's inequality therefore implies that the right-hand side of
\eqref{eq:pr-nonlinear-bound} converges to zero in probability.
\end{proof}

\begin{lemma}[MDS-FCLT] \label{lemfclt}\label{lem:martingale-fclt}
For $t\geq1$, define
\[
\varepsilon_t(\theta)
\coloneqq
\frac{1}{m\Rsim}\sum_{b=1}^{m}\sum_{r=1}^{\Rsim}
g_{tb,r}(\theta)
-\mathbb E\{G_t(\theta)\mid\mathcal F_{t-1}\}.
\]
On an enlarged probability space, let $\varepsilon_t^0$ be a fresh innovation generated at $\theta_*$ that, conditional on $\mathcal F_{t-1}$, has the same distribution as $\varepsilon_t(\theta_*)$ and is independent of the actual iteration-$t$ draws. Continue to denote the enlarged filtration by $\mathcal F_t$, and define
\[
\widetilde\varepsilon_t
\coloneqq
\mathbb I\{\tau>t-1\}\varepsilon_t(\theta_{t-1})
+\mathbb I\{\tau\leq t-1\}\varepsilon_t^0,
\qquad
\widetilde M_T(r)
\coloneqq
\frac{1}{\sqrt T}\sum_{t=1}^{\lfloor Tr\rfloor}\widetilde\varepsilon_t.
\]
Then
\[
    -H_*^{-1}\widetilde M_T(r)
    \Rightarrow
    (H_*^{-1}V_{m,\Rsim}H_*^{-1})^{1/2}\mathbb W(r)
    \quad\text{in }D[0,1]^p,
\]
where $\mathbb W$ is a $p$-dimensional standard Wiener process. In particular, $\widetilde M_T(1)\xrightarrow{d}\mathcal N(0,V_{m,\Rsim})$. Moreover, on $\mathcal E=\{\tau=\infty\}$,
\[
    \widetilde\varepsilon_t
    =\varepsilon_t(\theta_{t-1})
    =\zeta_t+\xi_t
    \qquad\text{for every }t\geq1.
\]
\end{lemma}

\begin{proof}
Because $\{\tau>t-1\}\in\mathcal F_{t-1}$ and each component innovation has conditional mean zero,
\[
    \mathbb E(\widetilde\varepsilon_t\mid\mathcal F_{t-1})=0.
\]
Thus $\{\widetilde\varepsilon_t,\mathcal F_t\}$ is a martingale-difference sequence. Assumption~\ref{ass:online-regularity}(iii), Jensen's inequality, and the fact that $\theta_{t-1}\in B_\rho\subset\mathcal N_*$ on $\{\tau>t-1\}$ give the uniform conditional moment bound
\begin{equation}\label{eq:continued-moment}
    \mathbb E(\|\widetilde\varepsilon_t\|^{2+\delta}\mid\mathcal F_{t-1})\leq C.
\end{equation}

For the predictable quadratic variation, write
\[
C_t(\theta)
\coloneqq
\mathbb E\{\varepsilon_t(\theta)\varepsilon_t(\theta)'\mid\mathcal F_{t-1}\}
=
\operatorname{Var}\{G_t(\theta)\mid\mathcal F_{t-1}\}
+\frac{1}{m\Rsim}
\mathbb E\!\left[\operatorname{Var}\{g_{ir}(\theta)\mid D_i\}\right],
\]
where the expectation in the second term is over a generic population
observation. By construction,
\[
\widetilde C_t
\coloneqq
\mathbb E(\widetilde\varepsilon_t\widetilde\varepsilon_t'\mid\mathcal F_{t-1})
=
\mathbb I\{\tau>t-1\}C_t(\theta_{t-1})
+\mathbb I\{\tau\leq t-1\}C_t(\theta_*).
\]
Assumption~\ref{ass:online-regularity}(ii) and the integral mean-value formula imply
\[
\mathbb E\{\|G_t(\theta)-G_t(\theta_*)\|^2\mid\mathcal F_{t-1}\}
\leq C\|\theta-\theta_*\|^2
\qquad (\theta\in B_\rho).
\]
Together with the bounded second moments and Assumption~\ref{ass:online-regularity}(iv), this yields
\[
\|C_t(\theta)-C_t(\theta_*)\|
\leq C\{\|\theta-\theta_*\|+\|\theta-\theta_*\|^2\}
\qquad (\theta\in B_\rho).
\]
Consequently, Lemma~\ref{lem1}(iii) and Cauchy--Schwarz give, for \(t\geq2\),
\[
\begin{aligned}
\mathbb E\|\widetilde C_t-C_t(\theta_*)\|
&\leq
C\mathbb E\!\left[
\{\|\Delta_{t-1}\|+\|\Delta_{t-1}\|^2\}
\mathbb I\{\tau>t-1\}
\right]\\
&\leq C\{\gamma_{t-1}^{1/2}+\gamma_{t-1}\}
\longrightarrow0.
\end{aligned}
\]
By \eqref{eq:vmr} of the main text, \(C_t(\theta_*)=V_{m,\Rsim}\).  Also by \(|\lfloor Tr\rfloor/T-r|\leq T^{-1}\),
\[
\begin{aligned}
\sup_{0\leq r\leq1}
  \left\|\frac1T\sum_{t=1}^{\lfloor Tr\rfloor}
           \widetilde C_t-rV_{m,\Rsim}\right\|
&\leq \sup_r\frac1T\sum_{t=1}^{\lfloor Tr\rfloor}\left\|
           \widetilde C_t-V_{m,\Rsim}\right\|
     +\sup_r\left|\frac{\lfloor Tr\rfloor}{T}-r\right|\|V_{m,\Rsim}\|\\
&\leq \frac1T\sum_{t=1}^{T}\left\|
           \widetilde C_t-V_{m,\Rsim}\right\|+\frac{\|V_{m,\Rsim}\|}{T}.
\end{aligned}
\]
Thus
\begin{align}\label{eq:continued-pqv}
&\mathbb E\!\left[\sup_{0\leq r\leq1}
  \left\|\frac1T\sum_{t=1}^{\lfloor Tr\rfloor}
       \widetilde C_t-rV_{m,\Rsim}\right\|\right]
\leq \frac1T\sum_{t=1}^T
       \mathbb E\|\widetilde C_t-V_{m,\Rsim}\|+\frac{\|V_{m,\Rsim}\|}{T} \\
&\leq \frac CT\left[1+\sum_{t=2}^T
       \{(t-1)^{-a/2}+(t-1)^{-a}\}\right] =O(T^{-1}+T^{-a/2})\longrightarrow0.
\end{align}
 The second uses the preceding
variance bound and \(\gamma_t=\gamma_0t^{-a}\). Finally,
\(\sum_{s=1}^{T-1}s^{-a/2}=O(T^{1-a/2})\) and
\(\sum_{s=1}^{T-1}s^{-a}=O(T^{1-a})\).

For every $\eta>0$, \eqref{eq:continued-moment} gives the conditional Lindeberg bound
\begin{align*}
&\frac{1}{T}\sum_{t=1}^T
\mathbb E\!\left[
\|\widetilde\varepsilon_t\|^2
\mathbb I\{\|\widetilde\varepsilon_t\|>\eta\sqrt T\}
\mid\mathcal F_{t-1}
\right]\\
&\qquad\leq
\frac{1}{\eta^\delta T^{1+\delta/2}}
\sum_{t=1}^T
\mathbb E(\|\widetilde\varepsilon_t\|^{2+\delta}\mid\mathcal F_{t-1})
\leq
\frac{C}{\eta^\delta T^{\delta/2}}
\longrightarrow0.
\end{align*}
The martingale functional central limit theorem
\citep[see, e.g.,][Theorem~4.2]{hallHeyde1980martingale}, applied using
\eqref{eq:continued-pqv} and the conditional Lindeberg condition, gives
\[
    \widetilde M_T(r)\Rightarrow V_{m,\Rsim}^{1/2}\mathbb W(r).
\]
Premultiplication by $-H_*^{-1}$ gives the stated functional limit. The final equality on $\mathcal E$ follows immediately from the definition of $\widetilde\varepsilon_t$.
\end{proof}

\begin{lemma}\label{lemleading}
Let $\widetilde\Delta_t^1$ satisfy \eqref{eq:continued-linear-recursion}. Then
\[
\sup_{r\in[0,1]}
\left\|
\frac{1}{\sqrt T}
\sum_{t=1}^{\lfloor Tr\rfloor}\widetilde\Delta_t^1
+H_*^{-1}\widetilde M_T(r)
\right\|
=o_{\mathbb P}(1).
\]
\end{lemma}

\begin{proof}
Iterating \eqref{eq:continued-linear-recursion} and summing through an arbitrary integer $k\geq0$ gives the exact finite-horizon representation
\begin{equation}\label{eq:pr-leading-exact}
\sum_{t=1}^k\widetilde\Delta_t^1
=
B_k-\sum_{j=1}^kK_{j,k}\widetilde\varepsilon_j
=
-H_*^{-1}\sum_{j=1}^k\widetilde\varepsilon_j
+B_k-\sum_{j=1}^kw_j^k\widetilde\varepsilon_j,
\end{equation}
where $B_k$, $K_{j,k}$, and $w_j^k$ are defined in \eqref{eq:pr-finite-weights}; all sums are empty when their upper limit is below their lower limit. By \eqref{eq:pr-weight-bounds} and Assumption~\ref{ass:algorithmic-setup} of the main text,
\[
    \frac{1}{\sqrt T}\max_{0\leq k\leq T}\|B_k\|
    \leq
    \frac{C}{\sqrt T}\|\Delta_0\|
    =o_{\mathbb P}(1).
\]

Because $q\geq2p_0$ in Assumption~\ref{ass:online-regularity}(iii), the same argument as for \eqref{eq:continued-moment} gives
\[
    \mathbb E(\|\widetilde\varepsilon_j\|^{2p_0}\mid\mathcal F_{j-1})\leq C.
\]
For each fixed $k$, Burkholder's inequality and \eqref{eq:pr-weight-bounds} therefore imply
\begin{align*}
\mathbb E\left\|
\sum_{j=1}^kw_j^k\widetilde\varepsilon_j
\right\|^{2p_0}
&\leq
C\mathbb E\left(
\sum_{j=1}^k\|w_j^k\|^2\|\widetilde\varepsilon_j\|^2
\right)^{p_0}\leq
C\left(\sum_{j=1}^k\|w_j^k\|^2\right)^{p_0}
\leq Ck^{ap_0}.
\end{align*}
Hence, for every $\eta>0$, a union bound and Markov's inequality give
\begin{align*}
&\mathbb P\left(
\max_{1\leq k\leq T}
\frac{1}{\sqrt T}
\left\|
\sum_{j=1}^kw_j^k\widetilde\varepsilon_j
\right\|>\eta
\right)\leq
\frac{C}{\eta^{2p_0}T^{p_0}}
\sum_{k=1}^Tk^{ap_0}
=
O\!\left(T^{1-(1-a)p_0}\right)
\longrightarrow0,
\end{align*}
because $p_0>(1-a)^{-1}$. Taking $k=\lfloor Tr\rfloor$ in \eqref{eq:pr-leading-exact} proves the result uniformly over $r\in[0,1]$.
\end{proof}

\begin{samepage}
\begin{lemma}\label{lem:rateconverg} Suppose $\gamma_t= \gamma_0t^{-a}$ for $t\geq1$, with $a\in(1/2,1)$, and let $u_t=\mathbb E(\|\Delta_t\|^2\mathbb I\{\tau>t\})$ be as in Lemma~\ref{lem1}.
There are constants $C_5, C_6>0$ such that, for all $t\geq1$,
$$
u_t\leq C_5\exp(-C_6t^{1-a}) + C_5t^{-a}  .
$$
Consequently, at the terminal index \(t=T\),
$\|\Delta_T\|\mathbb I\{\tau>T\}\xrightarrow{p}0$. Under the additional small-step condition of Lemma~\ref{lem1}(ii), for every $b>0$,
\[
\limsup_{T\to\infty}\mathbb P(\|\theta_T-\theta_*\|>b)\leq\epsilon.
\]
\end{lemma}
\end{samepage}

\begin{proof}
By \eqref{eq81},
\begin{equation}\label{eq:rate-recursion}
u_t
\leq
(1-2c_0\gamma_t+C\gamma_t^2)u_{t-1}
+C\gamma_t^2,
\qquad
u_0\leq\rho^2.
\end{equation}
The above inequality is a classic recursion in the SGD literature
(e.g., \citealp{moulines2011non,chen2024online}), from which the rate of
convergence follows by a standard argument. We provide the proof for completeness.

Decreasing the upper bound on \(\gamma_0\), if necessary, ensures that
\(C\gamma_t\leq c_0\) and \(c_0\gamma_t\leq1\) for every \(t\geq1\). Hence
\[
0\leq1-c_0\gamma_t\leq1,
\qquad
1-2c_0\gamma_t+C\gamma_t^2\leq1-c_0\gamma_t,
\]
and \eqref{eq:rate-recursion} implies
\begin{equation}\label{eq:rate-positive-recursion}
    u_t\leq(1-c_0\gamma_t)u_{t-1}+C\gamma_t^2.
\end{equation}
Iterating this recursion and using $1-x\leq e^{-x}$ gives
\begin{equation}\label{eq:rate-unrolled}
u_t
\leq
u_0\exp\!\left(-c_0\sum_{i=1}^t\gamma_i\right)
+C\sum_{k=1}^t\gamma_k^2
\exp\!\left(-c_0\sum_{i=k+1}^t\gamma_i\right).
\end{equation}

Integral bounds for the learning-rate sequence give constants $c,C>0$ such that
\begin{equation}\label{eq:rate-sum-bound}
c\{t^{1-a}-s^{1-a}\}
\leq
\sum_{i=s+1}^t\gamma_i
\leq
C\{t^{1-a}-s^{1-a}\},
\qquad 0\leq s<t.
\end{equation}
The first term in \eqref{eq:rate-unrolled} is therefore bounded by $C\exp(-ct^{1-a})$.

For the second term, let
$k_\star\coloneqq \lfloor t/2\rfloor$ and first consider \(k\leq k_\star\).
Because $a>1/2$, $\sum_{k\geq1}\gamma_k^2<\infty$, and
\eqref{eq:rate-sum-bound} gives
\begin{align*}
&\sum_{k=1}^{k_\star}\gamma_k^2
\exp\!\left(-c_0\sum_{i=k+1}^t\gamma_i\right)\leq
\exp\!\left(-c_0\sum_{i=k_\star+1}^t\gamma_i\right)
\sum_{k=1}^{\infty}\gamma_k^2
\leq
C\exp(-ct^{1-a}).
\end{align*}
For $k\geq k_\star+1$, monotonicity of $\gamma_k$ gives
\(\gamma_k\leq C\gamma_t\), and hence
\begin{align*}
&\sum_{k=k_\star+1}^{t}\gamma_k^2
\exp\!\left(-c_0\sum_{i=k+1}^t\gamma_i\right)\leq
C\gamma_t
\sum_{k=k_\star+1}^{t}\gamma_k
\exp\!\left(-c_0\sum_{i=k+1}^t\gamma_i\right).
\end{align*}
Set $E_k\coloneqq \exp\{-c_0\sum_{i=k+1}^t\gamma_i\}$. Since $c_0\gamma_k\leq1$,
$E_k-E_{k-1}
=E_k(1-e^{-c_0\gamma_k})
\geq
\frac{c_0}{2}\gamma_kE_k.$

Hence
\begin{align*}
\sum_{k=k_\star+1}^{t}\gamma_k^2
\exp\!\left(-c_0\sum_{i=k+1}^t\gamma_i\right)
&\leq  C\gamma_t
\sum_{k=k_\star+1}^{t}\gamma_kE_k
\leq  C\gamma_t
\frac{2}{c_0}\sum_{k=k_\star+1}^{t}(E_k-E_{k-1})\cr
&\leq
\frac{2}{c_0} C\gamma_t\leq ct^{-a}.
\end{align*}
  Combining with \eqref{eq:rate-unrolled},
\[
    u_t\leq C_5\exp(-C_6t^{1-a})+C_5t^{-a}.
\]

For every $b>0$, taking $t=T$ and applying Markov's inequality gives
\[
\mathbb P\!\left(
\|\Delta_T\|\mathbb I\{\tau>T\}>b
\right)
\leq
\frac{u_T}{b^2}
\longrightarrow0,
\]
which proves the stopped consistency claim. Moreover,
\[
\mathbb P(\|\Delta_T\|>b)
\leq
\mathbb P(\tau\leq T)+\frac{u_T}{b^2}.
\]
Lemma~\ref{lem1}(ii) consequently gives
\[
\limsup_{T\to\infty}\mathbb P(\|\Delta_T\|>b)\leq\epsilon.
\]
This is the asserted consistency in the $\epsilon$-sense; it does not imply unconditional $o_{\mathbb P}(1)$ for a fixed learning-rate constant.
\end{proof}

\begingroup
\subsection{Multi-pass estimation and inference}\label{app:multipass-proof}

\begin{proof}[Proof of Theorem~\ref{thm:multipass} of the main text]
Let \(\mathcal G_n\) be the sigma-field generated by the data set and
initialization information, and let \(\mathcal F_{n,t}\) augment it with the
sampled indices and simulation draws through update \(t\). Let \(\widehat\theta_n\) be the locally unique root of the sample negative score,
which is used only as an auxiliary estimator for the theory, not for
practical implementation.
Assign \(\widehat\theta_n\) an arbitrary measurable value when
the local root does not exist.

\textit{Step 1: uniform local control conditional on the data.}
Set \(e_t\coloneqq\theta_t-\widehat\theta_n\) and condition on
\(\mathcal G_n\). Work on \(E_n\cap\{\|e_0\|\leq\rho\}\), whose
probability tends to one by pilot consistency and
\eqref{eq:sample-root-expansion} of the main text. All constants below are uniform on this
event.
For any \(\|e\|\leq\rho\),
\[
    e'\mathcal R_n(\widehat\theta_n+e)\geq(c/2)\|e\|^2,
    \qquad
    \|\mathcal R_n(\widehat\theta_n+e)\|\leq C\|e\|.
\]
For \(t\geq1\), define locally
\[
    \varepsilon_{n,t}\coloneqq
    \widehat G_t(\theta_{t-1})-\mathcal R_n(\theta_{t-1}).
\]
Its conditional mean is zero.
Squaring the update on \(\{\tau_n>t-1\}\), as in
Lemma~\ref{lem:stopped-bound}, therefore yields
\[
    u_{n,t}\leq(1-c\gamma_t+C\gamma_t^2)u_{n,t-1}
                  +C\gamma_t^2,
    \qquad
    u_{n,t}\coloneqq
      \mathbb E\{\|e_t\|^2\mathbb I\{\tau_n>t\}\mid\mathcal G_n\}.
\]
Since \(u_{n,0}\leq\rho^2\), the same  argument of
Lemma~\ref{lem:rateconverg} gives,
\begin{equation}\label{eq:empirical-stopped-rate}
    u_{n,t}\leq C\{\exp(-c_1t^{1-a})+t^{-a}\},\qquad t\geq1.
\end{equation}
The stopped supermartingale in Lemma~\ref{lem:stopped-bound}, now
conditional on \(\mathcal G_n\), also gives, for every \(b\in(0,\rho)\),
\[
    \mathbb P(\mathcal E_n^c)
    \leq\mathbb P(E_n^c)+\mathbb P(\|e_0\|>b)
       +\frac{b^2+C\gamma_0^2\sum_{s\geq1}s^{-2a}}{\rho^2}.
\]
Choose \(b=\rho\sqrt{\epsilon/2}\) and then \(\gamma_0\) sufficiently
small that the last fraction is at most \(\epsilon\).
The first two probabilities vanish, giving
\(\mathbb P(\mathcal E_n^c)\leq\epsilon+o(1)\).

\medskip\noindent
\textit{Step 2: Conditional Gaussian limit.}
The proof is similar to Lemmas~\ref{lemlinear}--\ref{lemleading},
but centers \(\bar\theta_T\) at \(\widehat\theta_n\).
Generate a fresh innovation \(\varepsilon_{n,t}^0\) that is independent of the actual update draws conditionally
on \(\mathcal F_{n,t-1}\), and has the same conditional distribution as
\(\widehat G_t(\widehat\theta_n)-\mathcal R_n(\widehat\theta_n)\). Define
\[
    \widetilde\varepsilon_{n,t}\coloneqq
      \mathbb I\{\tau_n>t-1\}\varepsilon_{n,t}
      +\mathbb I\{\tau_n\leq t-1\}\varepsilon_{n,t}^0,
    \qquad
    \widetilde M_{n,T}\coloneqq
      T^{-1/2}\sum_{t=1}^T\widetilde\varepsilon_{n,t}.
\]
These are martingale differences with uniformly bounded conditional
\((2+\delta)\)-moments. Write \(V_n^0\coloneqq V_{n,m,\Rsim}(\widehat\theta_n)\)
and
\(\widetilde C_{n,t}\coloneqq
\mathbb E(\widetilde\varepsilon_{n,t}\widetilde\varepsilon_{n,t}'
\mid\mathcal F_{n,t-1})\).
By construction,
\[
\widetilde C_{n,t}-V_n^0
=\mathbb I\{\tau_n>t-1\}
\{V_{n,m,\Rsim}(\theta_{t-1})-V_{n,m,\Rsim}(\widehat\theta_n)\}.
\]
\par
The preceding identity gives
\begin{align}
&\mathbb E\!\left[
 \left\|\frac1T\sum_{t=1}^T\widetilde C_{n,t}-V_n^0\right\|
 \,\middle|\,\mathcal G_n\right]
\leq\frac CT\sum_{t=1}^T
 \mathbb E\!\left[
 \|e_{t-1}\|\mathbb I\{\tau_n>t-1\}
 \,\middle|\,\mathcal G_n\right]\leq\frac CT\sum_{t=1}^T u_{n,t-1}^{1/2} \\
&\leq\frac CT\left[1+\sum_{s=1}^{T-1}
 \{e^{-c_1s^{1-a}/2}+s^{-a/2}\}\right] =O(T^{-1}+T^{-a/2})\longrightarrow0,
\end{align}
uniformly on the stated data and initialization events for fixed
\(\gamma_0\).
The first inequality uses the triangle inequality and the
variance bound in Assumption~\ref{ass:empirical-regularity}(iii) of the main text; the
second follows from conditional Cauchy--Schwarz. The third uses
\eqref{eq:empirical-stopped-rate} and \(u_{n,0}\leq\rho^2\). The final
rate follows because \(\sum_{s\geq1}e^{-c_1s^{1-a}/2}<\infty\) and
\(\sum_{s=1}^{T-1}s^{-a/2}=O(T^{1-a/2})\).
For each \(\eta>0\),
the conditional Lindeberg sum is bounded by
\(C\eta^{-\delta}T^{-\delta/2}\).
Together with \(V_n^0\xrightarrow{p}V_{m,\Rsim}\), these bounds yield
the conditional martingale triangular-array CLT
\citep{hallHeyde1980martingale}: for every \(u\in\mathbb R^p\),
\begin{equation}\label{eq:empirical-conditional-cf}
\begin{split}
\mathbb E\{\exp(iu'\widetilde B_{n,T})\mid\mathcal G_n\}
&\xrightarrow{p}\exp(-u'\Sigma_{\mathrm{alg}}u/2),\\
\widetilde B_{n,T}&\coloneqq-H_n^{-1}\widetilde M_{n,T}.
\end{split}
\end{equation}
The uniform bounds justify the conditional CLT along almost surely
convergent subsequences of the data quantities, and hence give
\eqref{eq:empirical-conditional-cf} in probability.

\medskip\noindent
\textit{Step 3: Remainders.}
Define \(B_{n,k}\), \(K_{n,j,k}\), and
\(w_{n,j}^k\coloneqq K_{n,j,k}-H_n^{-1}\) as in
\eqref{eq:pr-finite-weights}, with \(H_n\) and \(e_0\) in place of
\(H_*\) and \(\Delta_0\). Their bounds
\eqref{eq:pr-weight-bounds} also hold. Hence
\[
    \mathbb E\!\left[
      \left\|T^{-1/2}\sum_{j=1}^T
                  w_{n,j}^T\widetilde\varepsilon_{n,j}\right\|^2
      \,\middle|\,\mathcal G_n\right]
    \leq\frac CT\sum_{j=1}^T\|w_{n,j}^T\|^2
    \leq CT^{a-1}\longrightarrow0.
\]
Also \(\|B_{n,T}\|/\sqrt T\leq C\rho/\sqrt T\).
On \(\mathcal E_n\), the difference between the actual and linear
averages is bounded, as in \eqref{eq:pr-nonlinear-bound}, by
\(CT^{-1/2}\sum_{j=1}^T\|e_{j-1}\|^2\mathbb I\{\tau_n>j-1\}\).
Its conditional expectation is at most
\[
    \frac C{\sqrt T}\sum_{j=1}^T u_{n,j-1}
    =O(T^{-1/2}+T^{1/2-a})=o(1).
\]
Consequently,
\begin{equation}\label{eq:empirical-auxiliary-coupling}
    \{\sqrt T(\bar\theta_T-\widehat\theta_n)-\widetilde B_{n,T}\}
       \mathbb I\{\mathcal E_n\}\xrightarrow{p}0.
\end{equation}

\medskip\noindent
\textit{Step 4: joint independence.}
Let \(A_n\coloneqq\sqrt n(\widehat\theta_n-\theta_*)\), which is
\(\mathcal G_n\)-measurable. Equation~\eqref{eq:sample-root-expansion} of the main text
and the i.i.d.\ CLT give \(A_n\Rightarrow
\mathcal N(0,\Sigma_{\mathrm{stat}})\).
By boundedness of the conditional characteristic function and
\eqref{eq:empirical-conditional-cf}, for fixed
\(u,v\in\mathbb R^p\),
\[
\begin{split}
    \mathbb E\exp\{iv'A_n+iu'\widetilde B_{n,T}\}
    &=\mathbb E\!\left[
      e^{iv'A_n}\mathbb E\{e^{iu'\widetilde B_{n,T}}\mid\mathcal G_n\}
      \right]\\
    &\longrightarrow
      \exp\{-v'\Sigma_{\mathrm{stat}}v/2-u'\Sigma_{\mathrm{alg}}u/2\}.
\end{split}
\]
Thus \((A_n',\widetilde B_{n,T}')'\Rightarrow
(Z_{\mathrm{stat}}',Z_{\mathrm{alg}}')'\). Together with
\eqref{eq:empirical-auxiliary-coupling} and
\(\mathbb P(\mathcal E_n^c)\leq\epsilon+o(1)\), this gives the
bounded-continuous-function bound in (ii).

\medskip\noindent
\textit{Step 5: the joint FCLT.}
Suppose condition (iv) of Assumption~\ref{ass:empirical-regularity}
of the main text holds. Extend the auxiliary
endpoint in Step 2 to
\[
    \widetilde B_{n,T}(r)\coloneqq
    -H_n^{-1}T^{-1/2}
       \sum_{t=1}^{\lfloor Tr\rfloor}\widetilde\varepsilon_{n,t},
    \qquad r\in[0,1].
\]
The conditional bracket bound in that step also holds uniformly in time:
\[
    \mathbb E\!\left[
      \sup_{0\leq r\leq1}
      \left\|T^{-1}\sum_{t=1}^{\lfloor Tr\rfloor}
         \widetilde C_{n,t}-rV_n^0\right\|
      \,\middle|\,\mathcal G_n\right]
    \leq\frac CT\left(1+\sum_{t=1}^T u_{n,t-1}^{1/2}\right)=o(1).
\]
Indeed, the supremum is bounded by the sum of the variance differences
and the grid error. The same conditional Lindeberg bound therefore gives
the martingale FCLT. The subsequence argument in Step 2 yields, for each
bounded continuous \(g\) on \(D[0,1]^p\),
\[
    \mathbb E\{g(\widetilde B_{n,T})\mid\mathcal G_n\}
    \xrightarrow{p}
    \mathbb E g(\Sigma_{\mathrm{alg}}^{1/2}\mathbb W).
\]

It remains to make the averaging approximation uniform. As in Lemma~\ref{lemleading}, conditional
Burkholder and the uniform weight bounds give
\[
    \mathbb E\!\left[
      \left\|\sum_{j=1}^k w_{n,j}^k
                        \widetilde\varepsilon_{n,j}\right\|^{2p_0}
      \,\middle|\,\mathcal G_n\right]
    \leq C\left(\sum_{j=1}^k\|w_{n,j}^k\|^2\right)^{p_0}
    \leq Ck^{ap_0}.
\]
Consequently, for every \(\eta>0\),
\[
    \mathbb P\!\left(
      \max_{1\leq k\leq T}T^{-1/2}
      \left\|\sum_{j=1}^k w_{n,j}^k
                        \widetilde\varepsilon_{n,j}\right\|>\eta
      \,\middle|\,\mathcal G_n\right)
    \leq C\eta^{-2p_0}T^{1-(1-a)p_0}=o(1).
\]
The initial term is bounded uniformly by \(C\rho/\sqrt T\). The nonlinear
difference for every partial sum is bounded by the same
\(CT^{-1/2}\sum_{j=1}^T\|e_{j-1}\|^2\mathbb I\{\tau_n>j-1\}\) as in Step~3,
whose conditional first moment tends to zero.
Thus
\begin{equation}\label{eq:empirical-path-coupling}
    \sup_{0\leq r\leq1}
    \|\mathbb X_{n,T}(r)-\widetilde B_{n,T}(r)\|
       \mathbb I\{\mathcal E_n\}\xrightarrow{p}0.
\end{equation}
The conditional convergence for \(g\) is also in \(L^1\), by
boundedness. Multiplying it by any bounded continuous function of the
\(\mathcal G_n\)-measurable \(A_n\), and using tightness, proves
joint convergence of \((A_n,\widetilde B_{n,T})\) to the independent
pair in \eqref{eq:empirical-joint-fclt} of the main text. The uniform coupling transfers
this limit with the same \(2\epsilon{ \|f\|_\infty}\) bound as in Step 4.

\medskip\noindent
\textit{Step 6: burn-in and preconditioning.}
For deterministic \(T_0=o(T)\), evaluate the finite-horizon representation
used in Step~3 at \(k=T_0\). The initial term is at most \(C\rho/\sqrt T\);
the leading martingale term and the weighted remainder have conditional
second moments at most \(CT_0/T\) and \(CT_0^a/T\), respectively; and the
nonlinear term has conditional first moment at most
\(C(1+T_0^{1-a})/\sqrt T=o(1)\). Hence
\[
    \left\|T^{-1/2}\sum_{t=1}^{T_0}e_t\right\|\mathbb I\{\mathcal E_n\}
    \xrightarrow{p}0.
\]
Since \(T/(T-T_0)\to1\), deleting this initial segment and rescaling time
preserves the endpoint and functional conclusions.
For a fixed symmetric positive-definite preconditioner \(\Pi\), set
\(\psi\coloneqq\Pi^{-1/2}(\theta-\widehat\theta_n)\). Conditional on the
data set, the transformed sample average negative score has Jacobian
\(\Pi^{1/2}H_n\Pi^{1/2}\), and its update variance at the target is
\(\Pi^{1/2}V_{n,m,\Rsim}(\widehat\theta_n)\Pi^{1/2}\). Applying the
preceding argument in the transformed coordinates and mapping back
preserves the limiting algorithmic variance
\(H_*^{-1}V_{m,\Rsim}H_*^{-1}\), with an adjusted small-step threshold.
\end{proof}

\begin{proof}[Proof of Corollary~\ref{cor:multipass-inference} of the main text]
When \(n/T\to\kappa<\infty\), on \(\mathcal E_n\), \eqref{eq:empirical-auxiliary-coupling} implies
$$
\sqrt n(\bar\theta_T-\theta_*) = \sqrt n(\bar\theta_T-\widehat\theta_n)+  A_n=
A_n +\sqrt{n/T} \widetilde B_{n,T}  +{ o_{\mathbb P}}(1).
$$
Write
\[
    \widetilde{\mathbb X}_{n,T}\coloneqq
       A_n+\sqrt{n/T}\,\widetilde B_{n,T}
    \Rightarrow\mathcal N(0,\Sigma_\kappa)
\]
  This gives the finite-ratio limit used in Corollary~\ref{cor:multipass-inference} of the main text.
The variance assumptions and continuity of matrix inversion give
\(\|n\widehat C_{n,T}-\Sigma_\kappa\|
\mathbb I\{\mathcal E_n\}\xrightarrow{p}0\).
Use the auxiliary variance construction in the proof of
Corollary~\ref{cor:inference}, with
\(\widetilde{\mathbb X}_{n,T}\), \(n\widehat C_{n,T}\),
\(\Sigma_\kappa\), and \(\mathcal E_n\).
The auxiliary Wald statistic converges to \(\chi^2_\ell\);
the rejection probabilities differ from those of \(W_{n,T}\)
by at most \(\mathbb P(\mathcal E_n^c)+o(1)\).
\par

When \(n/T\to\infty\), we  instead work with the   vector
\(\sqrt{T/n}\,A_n+\widetilde B_{n,T}\), whose limit is
\(\mathcal N(0,\Sigma_{\mathrm{alg}})\).
Use the same construction with \(T\widehat C_{n,T}\) and
\(\Sigma_{\mathrm{alg}}\).
\end{proof}

\par\endgroup

\section{\texorpdfstring{Additional inference}{Additional inference}}\label{app:supp-inference}

This appendix gives delta-method inference for smooth scalar
parameters in Section~\ref{sec:supp-delta-method}, the common random-scaling
definitions in Section~\ref{sec:supp-rs-definitions}, adjusted multi-pass
random-scaling inference in Section~\ref{sec:supp-multipass-rs}, and online
random-scaling and plug-in inference in Section~\ref{sec:supp-online-intervals}.
All intervals use \(T_0=0\).

\subsection{\texorpdfstring{Delta-method inference for smooth scalar parameters}{Delta-method inference for smooth scalar parameters}}\label{sec:supp-delta-method}\label{sec:smooth-targets}

For scalar parameters \(h(\theta_*)\), such as a choice probability or
substitution elasticity, we use the following smoothness condition.

\begin{assumption}[Smooth scalar function]\label{ass:smooth-scalar}
The function \(h:\Theta\to\mathbb R\) is twice continuously
differentiable on a neighborhood \(\mathcal N_h\subseteq\Theta\) of
\(\theta_*\), with
\[
    \sup_{\theta\in\mathcal N_h}\|\nabla_\theta^2h(\theta)\|<\infty,
    \qquad \nabla_\theta h(\theta_*)\ne0.
\]
\end{assumption}

\begin{corollary}[Delta-method inference for smooth scalar parameters]\label{cor:multipass-delta}
Under the joint CLT, variance consistency, and relative-rate
conditions of Corollary~\ref{cor:multipass-inference} of the main text, suppose
Assumption~\ref{ass:smooth-scalar} holds. Choose \(\rho\) in
\eqref{eq:multipass-local-event} sufficiently small that
\(B_{n,\rho}\subseteq\mathcal N_h\) with probability tending to one.
For every \(\alpha\in(0,1)\), the delta-method interval
\(I_{n,T}^h\) in Section~\ref{sec:inference-regimes} of the main text satisfies
\[
    \limsup_{n,T\to\infty}
    \left|\mathbb P\{h(\theta_*)\in I_{n,T}^h\}-(1-\alpha)\right|
    \leq
    \limsup_{n,T\to\infty}\mathbb P(\mathcal E_n^c).
\]
\end{corollary}

For a fixed known \(h_0\), the corresponding test of
\(H_0^h:h(\theta_*)=h_0\) rejects when \(h_0\) lies outside \(I_{n,T}^h\).

\begin{proof}[Proof of Corollary~\ref{cor:multipass-delta}]
Use the endpoint approximations and local variance
consistency established in the proof of
Corollary~\ref{cor:multipass-inference} of the main text
(Section~\ref{app:multipass-proof}). Let \(r_{n,T}=\sqrt n\)
when \(n/T\to\kappa<\infty\), and \(r_{n,T}=\sqrt T\) when
\(n/T\to\infty\). Under Assumption~\ref{ass:smooth-scalar} and the
choice of \(\rho\), Taylor expansion gives
\[
    r_{n,T}\left\{h(\bar\theta_T)-h(\theta_*)
       -\nabla_\theta h(\theta_*)'(\bar\theta_T-\theta_*)\right\}
       \mathbb I\{\mathcal E_n\}\xrightarrow{p}0.
\]
The bounded Hessian makes the quadratic remainder negligible at either
rate. The nonzero gradient and positive-definite limiting variance in
that proof ensure a positive limiting scalar variance.
Local consistency of the gradient and plug-in variance therefore gives
a standard-normal limit for the auxiliary studentized statistic.
Its coverage event agrees with that of \(I_{n,T}^h\) on \(\mathcal E_n\),
apart from events of probability \(o(1)\). The coverage probabilities
thus differ by at most \(\mathbb P(\mathcal E_n^c)+o(1)\), proving the
stated bound.
\end{proof}

\subsection{Random-scaling definitions}\label{sec:supp-rs-definitions}

Use the linear restriction \(H_0:A\theta_*=a_0\) from
Section~\ref{sec:inference-definitions} of the main text and the smooth
scalar target \(h(\theta_*)\) from Section~\ref{sec:supp-delta-method}.
Random scaling \citep[e.g.,][]{leeLiaoSeoShin2022randomScaling,liLiangChangZhang2022local,xieZhang2022inference,leeLiaoSeoShin2025quantile,duZhuWuNa2025randomScaling} uses centered partial sums of the
iterates to studentize the estimation error without separately estimating
the algorithmic variance. For \(A\theta_*\), define the self-normalizer
\begin{equation}\label{eq:rs-selfnormalizer}
    \widehat V_{\mathrm{RS}}(A)
    \coloneqq \frac{1}{T^2}\sum_{t=1}^{T}
    \left\{\sum_{s=1}^{t} A(\theta_s-\bar\theta_T)\right\}
    \left\{\sum_{s=1}^{t} A(\theta_s-\bar\theta_T)\right\}'.
\end{equation}
The self-normalizer has a nondegenerate random limit. For the
algorithmic component of the estimation error, however, the unknown
variance cancels from the limiting studentized statistic. Let
\(c^{\mathrm{RS}}_{\ell,1-\alpha}\) denote the \(1-\alpha\) quantile of
the resulting Brownian functional
\[
    \mathbb W_\ell(1)'
    \left\{
        \int_0^1
        \bar{\mathbb W}_\ell(r)\bar{\mathbb W}_\ell(r)'dr
    \right\}^{-1}
    \mathbb W_\ell(1),
\]
where \(\mathbb W_\ell\) is an \(\ell\)-dimensional standard Wiener process
and \(\bar{\mathbb W}_\ell(r)\coloneqq
\mathbb W_\ell(r)-r\mathbb W_\ell(1)\).

For random-scaling inference on \(h(\theta_*)\), define the average of the
transformed iterates and its self-normalizer by
\[
    \bar h_T\coloneqq \frac1T\sum_{t=1}^T h(\theta_t),
    \qquad
    \widehat V_{\mathrm{RS}}^h
    \coloneqq
    \frac1{T^2}\sum_{t=1}^T
    \left\{\sum_{s=1}^t\{h(\theta_s)-\bar h_T\}\right\}^2.
\]
The random-scaling interval is centered at the average transformed iterate
\(\bar h_T\), whereas the plug-in interval is centered at
\(h(\bar\theta_T)\). For a scalar linear contrast, set
\(h(\theta)=A\theta\); then both centers equal \(A\bar\theta_T\).
The corresponding confidence intervals and tests are given in Sections~\ref{sec:supp-multipass-rs} and~\ref{sec:supp-online-intervals}.

\subsection{Multi-pass random-scaling inference}\label{sec:supp-multipass-rs}\label{sec:supp-multipass-intervals}

For random scaling, centering the iterates at \(\bar\theta_T\)
removes the sampling error \(\widehat\theta_n-\theta_*\) from their partial
sums. The self-normalizer therefore reflects algorithmic variation
conditional on the sample, and an adjustment is needed to account for
sampling uncertainty. Equation~\eqref{eq:vmr} gives the matrix upper bounds
\begin{equation}\label{eq:rs-covariance-bound}
    \Sigma_{\mathrm{stat}}\preceq m\Sigma_{\mathrm{alg}},
    \qquad
    C_{n,T}\preceq (m/n+T^{-1})\Sigma_{\mathrm{alg}}.
\end{equation}
These bounds yield the conservative adjustment in
Corollary~\ref{cor:supp-multipass-rs}, following the variance
upper-bound argument of \citet[Section~4]{chenKimLeeSeoSong2025slim}.

\begin{corollary}[Adjusted multi-pass random-scaling inference]\label{cor:supp-multipass-rs}
Under the conditions for the joint FCLT in
Theorem~\ref{thm:multipass}(iii) of the main text, with the local event
\(\mathcal E_n\) defined in \eqref{eq:multipass-local-event}, consider
the null hypothesis \(H_0:A\theta_*=a_0\).
For scalar results involving \(h\), suppose
Assumption~\ref{ass:smooth-scalar} holds and choose
\(\rho\) in \eqref{eq:multipass-local-event} sufficiently small that
\(B_{n,\rho}\subseteq\mathcal N_h\) with probability tending to one.
The adjusted random-scaling statistic is
\begin{equation}\label{eq:multipass-rs}
    Q_{n,T}^{\mathrm{RS}}\coloneqq
    \frac{T}{1+mT/n}(A\bar\theta_T-a_0)'
       \widehat V_{\mathrm{RS}}(A)^{-1}(A\bar\theta_T-a_0).
\end{equation}
Under \(H_0\), for every \(\alpha\in(0,1)\) and every
joint sequence \(n,T\to\infty\),
\[
    \limsup_{n,T\to\infty}
    \mathbb P\{Q_{n,T}^{\mathrm{RS}}>c^{\mathrm{RS}}_{\ell,1-\alpha}\}
    \leq\alpha+
    \limsup_{n,T\to\infty}\mathbb P(\mathcal E_n^c).
\]
If \(n/T\to\infty\), the absolute limiting rejection error is bounded by
\(\limsup_{n,T\to\infty}\mathbb P(\mathcal E_n^c)\).
For a function \(h\) satisfying Assumption~\ref{ass:smooth-scalar} and the hypothesis
\(H_0^h:h(\theta_*)=h_0\), the scalar result also holds with
\(A\bar\theta_T-a_0\) replaced by \(\bar h_T-h_0\) and the
self-normalizer by \(\widehat V_{\mathrm{RS}}^h\) from
Section~\ref{sec:supp-rs-definitions}.
\end{corollary}

\paragraph{Adjusted random-scaling inference.}
For a scalar parameter \(h(\theta_*)\) satisfying
Assumption~\ref{ass:smooth-scalar}, the adjusted interval based on the
transformed iterates is
\[
    \bar h_T\ \pm\
    \{c^{\mathrm{RS}}_{1,1-\alpha}(T^{-1}+m/n)
       \widehat V_{\mathrm{RS}}^h\}^{1/2}.
\]
Its limiting coverage is at least
\(1-\alpha-\limsup_{n,T\to\infty}\mathbb P(\mathcal E_n^c)\).
For a scalar linear contrast, use \(h(\theta)=A\theta\).
For scalar random scaling,
\(c^{\mathrm{RS}}_{1,0.90}=28.334\) and
\(c^{\mathrm{RS}}_{1,0.95}=45.522\); their square roots, 5.323 and 6.747,
are the corresponding confidence-interval critical values
\citep[see, e.g.,][]{leeLiaoSeoShin2022randomScaling}.

When \(T/n\to0\), the adjustment vanishes and the interval has
nominal limiting coverage up to the probability of exiting the local
neighborhood. At general relative rates, removing the conservativeness introduced by the variance
upper bound requires estimating sampling uncertainty or adding a resampling
step. Adjusted random scaling provides a conservative alternative when
separate variance estimation is costly.

\begin{proof}[Proof of Corollary~\ref{cor:supp-multipass-rs}]
Use \(A_n\) and the auxiliary process
\(\widetilde B_{n,T}(\cdot)\) constructed in the proof of
Theorem~\ref{thm:multipass} in Appendix~\ref{app:supp-theory}.
\begin{samepage}
Set \(\Sigma_A\coloneqq A\Sigma_{\mathrm{alg}}A'\succ0\) and
\(\Sigma_{s,A}\coloneqq A\Sigma_{\mathrm{stat}}A'\).
Take a subsequence along which \(n/T\to\kappa\in[0,\infty]\).
\par
Define
\par\end{samepage}
\[
\begin{split}
    \widetilde U_{n,T}&\coloneqq
       \frac{A\widetilde B_{n,T}(1)+\sqrt{T/n}\,AA_n}
            {\sqrt{1+mT/n}}, \quad
    \widetilde Z_{n,T}(r) \coloneqq
       A\{\widetilde B_{n,T}(r)-r\widetilde B_{n,T}(1)\},\\
    \widetilde V_{n,T}&\coloneqq T^{-1}\sum_{t=1}^T
       \widetilde Z_{n,T}(t/T)\widetilde Z_{n,T}(t/T)'.
\end{split}
\]
Their   joint limit is
\[
    (\widetilde U_{n,T},\widetilde V_{n,T})
    \Rightarrow (G_\kappa,\Sigma_A^{1/2}\mathcal J_\ell\Sigma_A^{1/2}),
    \qquad
    \mathcal J_\ell\coloneqq\int_0^1
       \bar{\mathbb W}_\ell(r)\bar{\mathbb W}_\ell(r)'\,dr.
\]
The Gaussian vector \(G_\kappa\) is independent of \(\mathcal J_\ell\). Its variance is
\[
    \operatorname{Var}(G_\kappa)=
    \begin{cases}
       (\kappa\Sigma_A+\Sigma_{s,A})/(\kappa+m),&\kappa<\infty,\\
       \Sigma_A,&\kappa=\infty.
    \end{cases}
\]
By \eqref{eq:rs-covariance-bound},
\(K_\kappa\coloneqq\Sigma_A^{-1/2}\operatorname{Var}(G_\kappa)
\Sigma_A^{-1/2}\preceq I_\ell\).
Conditional on \(\mathcal J_\ell\), the limiting quadratic form has the law
of a weighted sum of independent \(\chi^2_1\) variables, with weights
the eigenvalues of \(\mathcal J_\ell^{-1/2}K_\kappa\mathcal J_\ell^{-1/2}\).
Since
\[
\mathcal J_\ell^{-1/2}K_\kappa\mathcal J_\ell^{-1/2}\preceq \mathcal J_\ell^{-1},
\]
\par
its ordered eigenvalues, and hence its conditional upper tail,
are bounded by those of the random-scaling reference law.
The centered-path identity and \eqref{eq:empirical-path-coupling}
transfer this bound to the actual statistic, adding at most
\(\mathbb P(\mathcal E_n^c)+o(1)\) to its rejection probability.\par
The limiting quadratic form has no atom at the positive critical
value. Every subsequence has a further subsequence with such a
relative-rate limit, proving the bound for arbitrary relative rates.
When \(\kappa=\infty\), \(K_\kappa=I_\ell\), giving the equality case.

For \(h(\theta_t)\), expand at \(\widehat\theta_n\) instead of \(\theta_*\).
By Assumption~\ref{ass:smooth-scalar} and the choice of \(\rho\) in
Corollary~\ref{cor:supp-multipass-rs},
\(B_{n,\rho}\subseteq\mathcal N_h\) with probability tending to one.\par
On these data events,
\[
\begin{aligned}
&\sup_{0\leq r\leq1}
 \left|T^{-1/2}\sum_{t=1}^{\lfloor Tr\rfloor}
       \{h(\theta_t)-h(\widehat\theta_n)\}
       -\nabla_\theta h(\widehat\theta_n)'\mathbb X_{n,T}(r)\right|
       {\mathbb I\{\mathcal E_n\}} \\
&\qquad\leq CT^{-1/2}\sum_{t=1}^T
       \|e_t\|^2\mathbb I\{\tau_n>t\}.
\end{aligned}
\]
The conditional first moment of the right-hand side, on the stated
data and initialization events, is
\(O(T^{-1/2}+T^{1/2-a})=o(1)\).
Also \(\nabla_\theta h(\widehat\theta_n)\xrightarrow{p}
\nabla_\theta h(\theta_*)\), and the usual delta method applies to
\(\sqrt n\{h(\widehat\theta_n)-h(\theta_*)\}\).
The resulting joint  limit also applies.

\end{proof}

\subsection{Online inference}\label{sec:supp-online-intervals}\label{sec:inference-rs}

The functional central limit theorem in Theorem~\ref{thm:sa-linear}
gives online random-scaling and plug-in inference for \(\theta_*\).
Use \(A,a_0,H_0\) and the variance estimator
\(\widehat\Sigma_{\mathrm{alg}}\) from
Section~\ref{sec:inference-definitions} of the main text, and the
self-normalizer \(\widehat V_{\mathrm{RS}}(A)\) and Brownian critical value
\(c^{\mathrm{RS}}_{\ell,1-\alpha}\) from
Section~\ref{sec:supp-rs-definitions}. Set

\[
\begin{aligned}
    Q_T^{\mathrm{RS}}&\coloneqq T(A\bar\theta_T-a_0)'
       \widehat V_{\mathrm{RS}}(A)^{-1}(A\bar\theta_T-a_0),\\
    W_T^{\mathrm{on}}&\coloneqq T(A\bar\theta_T-a_0)'
       (A\widehat\Sigma_{\mathrm{alg}}A')^{-1}(A\bar\theta_T-a_0).
\end{aligned}
\]
\begin{corollary}[Online random-scaling and plug-in inference]\label{cor:inference}
Under the FCLT in Theorem~\ref{thm:sa-linear}, with the online no-exit event
\(\mathcal E\) defined in \eqref{eq:online-local-event}, the random-scaling
bound below holds under \(H_0\). If, additionally,
\[
    \{\|{ \widehat H}-H_*\|+\|{ \widehat V}_{m,\Rsim}-V_{m,\Rsim}\|\}
       \mathbb I\{\mathcal E\}\xrightarrow{p}0,
\]
then the plug-in bound also holds under \(H_0\). Specifically, for every
\(\alpha\in(0,1)\),
\[
\begin{aligned}
    \limsup_{T\to\infty}
    \left|\mathbb P\{Q_T^{\mathrm{RS}}>
       c^{\mathrm{RS}}_{\ell,1-\alpha}\}-\alpha\right|
    &\leq
       \limsup_{T\to\infty}\mathbb P(\mathcal E^c),\\
    \limsup_{T\to\infty}
    \left|\mathbb P\{W_T^{\mathrm{on}}>\chi^2_{\ell,1-\alpha}\}
       -\alpha\right|
    &\leq
       \limsup_{T\to\infty}\mathbb P(\mathcal E^c).
\end{aligned}
\]
Here \(\chi^2_{\ell,1-\alpha}\) is the \(1-\alpha\) quantile of the
chi-squared distribution with \(\ell\) degrees of freedom.
\end{corollary}

Theorem~\ref{thm:sa-linear} shows that the exit-probability bound can be
made arbitrarily small through the learning-rate choice and pilot
consistency. Online variability from new observations and simulation is
already included in \(T^{-1}\Sigma_{\mathrm{alg}}\). In comparison, the
multi-pass variance \eqref{eq:multipass-covariance} of the main text
additionally includes sampling error in \(\widehat\theta_n\); its online
formula is recovered at the root-\(T\) scale when \(n/T\to\infty\),
but not when \(T=n\).

Let \(h\) satisfy the conditions of Theorem~\ref{thm:supp-transformation}.
With \(\bar h_T\) and \(\widehat V_{\mathrm{RS}}^h\) defined in
Section~\ref{sec:supp-rs-definitions}, the corresponding
online intervals are
\[
\begin{aligned}
\text{random scaling:}\quad&
    \bar h_T\ \pm\
    \left\{\frac{c^{\mathrm{RS}}_{1,1-\alpha}}{T}
       \widehat V_{\mathrm{RS}}^h\right\}^{1/2},\\
\text{plug-in:}\quad&
    h(\bar\theta_T)\ \pm\ z_{1-\alpha/2}
    \bigl\{T^{-1}\nabla_\theta h(\bar\theta_T)'
       \widehat\Sigma_{\mathrm{alg}}
       \nabla_\theta h(\bar\theta_T)\bigr\}^{1/2}.
\end{aligned}
\]
Reject \(H_0^h\) when \(h_0\) lies outside the interval chosen.

For \(\ell=1\), the random-scaling
and plug-in \((1-\alpha)\) confidence intervals for \(A\theta_*\) are
\[
\begin{aligned}
\text{random scaling:}\quad&
    A\bar\theta_T\ \pm\
    \left\{\frac{c^{\mathrm{RS}}_{1,1-\alpha}}{T}
       \widehat V_{\mathrm{RS}}(A)\right\}^{1/2},\\
\text{plug-in:}\quad&
    A\bar\theta_T\ \pm\
    z_{1-\alpha/2}
    \{T^{-1}A\widehat\Sigma_{\mathrm{alg}}A'\}^{1/2},
\end{aligned}
\]
where \(z_{1-\alpha/2}\) is the standard-normal quantile.

The local coverage bounds follow from Corollary~\ref{cor:inference},
using Theorem~\ref{thm:supp-transformation} for nonlinear scalar functions.
In particular, their absolute limiting coverage error is
bounded by the limiting probability of leaving the local neighborhood.

For online MNP plug-in inference, use an increasing collection of i.i.d.\@
evaluation observations, either retained from the stream or collected
after estimation, in the construction of Section~\ref{sec:mnp-covariance}
of the main text. In those variance formulas, \(n\) denotes the number
of evaluation observations, which need not equal the \(mT\) observations
used in the updates. Fresh simulations are generated after estimation.
Thus online plug-in inference requires access to observations after
estimation and is not fully streaming. Under
Proposition~\ref{prop:mnp-verification}, a fixed \(S\geq2\) suffices for
the local variance consistency required by Corollary~\ref{cor:inference}.
Under correct MNP specification, use

\[
    \widehat\Sigma_{\mathrm{alg}}
    ={ \widehat\Sigma}_g^{-1}\widehat V_{m,\Rsim}^{\mathrm{MNP}}
       { \widehat\Sigma}_g^{-1}
\]
in the online plug-in statistic and intervals above. The \(1/S\)
simulation correction and independent replicated directions are the same
as in the main-text construction.

\begingroup
% Keep isolated paragraph lines with the new appendix's text and displays.
\clubpenalty=10000
\widowpenalty=10000
\displaywidowpenalty=10000
\section{Verification of the MNP conditions}\label{app:mnp-verification}

This appendix collects the MNP accepted-moment and information
discussion, the online verification statement, proofs of Proposition~\ref{prop:mnp-verification}
and Proposition~\ref{prop:mnp-empirical} of the main text, and the consistency argument for the variance estimators in
Section~\ref{sec:mnp-covariance} of the main text.

\subsection{MNP accepted moments and information}\label{sec:supp-mnp-information}

\begin{proposition}[MNP regularity conditions for online SAUSS]\label{prop:mnp-verification}
Consider the exact MNP recursion in the normalized coordinates of
\eqref{eq:covariance-cholesky} of the main text, under the online sampling
scheme in Assumption~\ref{ass:online-sampling}. Let
\(q=\max\{2+\delta,2p_0\}\), with \(\delta>0\) and
\(p_0>(1-a)^{-1}\). Suppose conditions (i)--(iv) of
Proposition~\ref{prop:mnp-empirical} of the main text hold.
Then there exists a neighborhood \(\mathcal N_1\subseteq\mathcal N_*\)
of \(\theta_*\) on which Assumption~\ref{ass:online-regularity} holds
for the exact MNP simulator. In particular, \(\Sigma_g=H_*\), and the
one-update variance and positive lower bound are those in
\eqref{eq:mnp-verification-variance} of the main text.
\end{proposition}

Every canonical choice
region $\mathcal D_j$ defined in Section~\ref{sec:mnp-sauss} of the main text has nonempty
interior. Together with compactness of $\overline{\mathcal N_*}$, this property
gives
$P_{ij}(\theta)\geq c\exp\{-C(1+\|\DX_i\|^2)\}$ uniformly in
$j$ and $\theta\in\mathcal N_*$. A Gaussian tail calculation then shows
that, for every fixed $r>0$, an exact accepted disturbance satisfies
\begin{equation}\label{eq:supp-mnp-verification-accepted-moment}
    \sup_{\theta\in\mathcal N_*}
    \mathbb E_\theta^*\{\|\nu_i^\star(\theta)\|^r\mid D_i\}
    \leq C_r(1+\|\DX_i\|^r).
\end{equation}
The constant $C_r<\infty$ may depend on the moment order $r$,
the fixed model dimension, and the parameter neighborhood, but not on the
observation $D_i$ or $\theta\in\mathcal N_*$.
Here $\mathbb E_\theta^*$ averages only over the accepted-draw distribution
at the trial parameter $\theta$, with $D_i=(y_i,\DX_i)$ fixed.
These moment bounds do not require choice probabilities to be
bounded away from zero uniformly over the support of the covariates.

Recall that
$g_i(\theta)=-\nabla_\theta\log P_{iy_i}(\theta)$ is the exact negative observed-data
score, and let
\(\Gamma_i(\theta,\nu)\) be the normalized negative complete-data score defined in
\eqref{eq:mnp-one-draw-gradient} of the main text.
As in Lemma~\ref{lem:hm-identity} of the main text, differentiation holds
the latent utility differences fixed, rather than the residuals.
The information identity then gives
\begin{equation}
    \mathbb E \nabla_\theta g_i(\theta_*)
    =
       \mathbb E
      \{\nabla_\theta\Gamma_i(\theta_*,\nu_i^\star(\theta_*))\}
    -    \mathbb E\!\left[
  {\operatorname{Var}_{\theta_*}^*}\{g_{ir}(\theta_*)\mid D_i\}
\right],
\end{equation}
where the first term is the complete-data Fisher information matrix
(if the utility differences were observed). Thus the observed information
matrix is complete-data information minus missing information.
Positive definiteness of the complete-data information
$ \mathbb E
      \{\nabla_\theta\Gamma_i(\theta_*,\nu_i^\star(\theta_*))\}$ alone does not ensure positive definiteness of the
observed-data information $H_*=  \mathbb E \nabla_\theta g_i(\theta_*)$. Here $H_*\succ0$ follows from the normalized
probability-gradient rank condition in Proposition~\ref{prop:mnp-empirical}(iv) of the main text.

The rank condition is imposed in the normalized coordinates: fixing
$L_{11}=1$ removes the scale direction but does not by itself ensure local
identification.

\subsection{Verification proofs}\label{sec:supp-mnp-proofs}

This subsection proves Proposition~\ref{prop:mnp-verification} and Proposition~\ref{prop:mnp-empirical} of the main text. We use the
normalized negative complete-data score $\Gamma_i(\theta,\nu)$ defined in
\eqref{eq:mnp-one-draw-gradient} of the main text, in the Cholesky coordinates of
\eqref{eq:covariance-cholesky}.
Let $\nu_i^\star(\theta)$ have the exact accepted conditional law
defined in Section~\ref{sec:mnp-sauss} of the main text.
Writing $\mathbb E_\theta^*$ and $\operatorname{Var}_\theta^*$ for expectation
and variance over that law with $D_i=(y_i,\DX_i)$ fixed, we have
\begin{align*}
    g_{ir}(\theta)
    &\coloneqq
       \Gamma_i\{\theta,\nu_{ir}^\star(\theta)\},\\
    g_i(\theta)
    &=\mathbb E_\theta^*
       \{g_{ir}(\theta)\mid D_i\}.
\end{align*}
Thus $g_{ir}(\theta)$ is the simulated score based on one accepted draw in
the maintained simulation setup of Section~\ref{sec:asymptotic-theory} of the main text. Expectations over
observations below are under the true population law at $\theta_*$, not the
trial parameter $\theta$. Under the online sampling scheme in
Proposition~\ref{prop:mnp-verification},
the covariate moment condition \eqref{eq:mnp-verification-X-moment} of the main text also
holds conditionally on $\mathcal F_{t-1}$, uniformly over iterations.

\begin{samepage}
\begin{lemma}[Choice probabilities and accepted-draw moments]
\label{lem:mnp-accepted-moments}
Under Proposition~\ref{prop:mnp-empirical}(ii) of the main text, the canonical MNP choice
regions imply that there are constants
$c,C>0$ such that, for every realization of $\DX_i$,
\begin{equation}\label{eq:mnp-verification-rare-choice-lower-bound}
    \inf_{\theta\in\mathcal N_*}\min_{1\leq j\leq J}P_{ij}(\theta)
    \geq c\exp\{-C(1+\|\DX_i\|^2)\}.
\end{equation}
Consequently, the accepted-draw bound in
\eqref{eq:supp-mnp-verification-accepted-moment} holds almost surely for every fixed $r>0$.

\end{lemma}
\end{samepage}

\begin{proof}
Compactness of the parameter neighborhood implies constants
$0<\underline\lambda<\overline\lambda<\infty$ such that, uniformly over
$\theta\in\mathcal N_*$,
\[
  \|\beta\|\leq C,
  \qquad
  \underline\lambda I_d\preceq\Omega\preceq\overline\lambda I_d.
\]
Each canonical region $\mathcal D_j$ has nonempty interior: the point
$-\mathbf 1_d$ lies strictly inside $\mathcal D_1$, and the $(j-1)$th
coordinate vector lies strictly inside $\mathcal D_j$ for $j\geq2$.
Because $J$ is fixed, choose
balls $B(w_j^\circ,\eta)\subset\mathcal D_j$ with a common radius $\eta>0$ and with
$\max_j\|w_j^\circ\|<\infty$.  Put $\mu_i\coloneqq \DX_i\beta$.  The normal density on
$B(w_j^\circ,\eta)$ satisfies
\[
  \inf_{w\in B(w_j^\circ,\eta)}
  \phi_d(w;\mu_i,\Omega)
  \geq c\exp\{-C(1+\|\mu_i\|^2)\}
  \geq c\exp\{-C(1+\|\DX_i\|^2)\},
\]
uniformly in $j$ and $\theta$.  Integrating this bound over the ball proves
\eqref{eq:mnp-verification-rare-choice-lower-bound}.  In particular, every choice probability
is positive, although no positive lower bound uniform over $\DX_i$ is
asserted.

We next bound Gaussian conditional moments in terms of the
probability of the conditioning event. If
$Z\sim N(0,\Omega)$ and $A$ has probability $p>0$, the eigenvalue bounds give
constants $c_0,C_0>0$ such that
\[
  \mathbb P(\|Z\|>u)\leq C_0e^{-c_0u^2},
  \qquad u\geq0,
\]
uniformly over $\theta\in\mathcal N_*$.  After increasing $C_0$ if necessary,
take $C_0\geq1$.  Hence
\begin{align*}
  \mathbb E(\|Z\|^r\mid A)
  &=r\int_0^\infty u^{r-1}
       \mathbb P(\|Z\|>u\mid A)\,du\\
  &\leq r\int_0^\infty u^{r-1}
       \min\{1,C_0p^{-1}e^{-c_0u^2}\}\,du\leq  C_r\{1+\log(C_0/p)\}^{r/2}.
\end{align*}
The last inequality follows by splitting the integral at
$u_p\coloneqq \{c_0^{-1}\log(C_0/p)\}^{1/2}$; the integral below $u_p$ is $u_p^r$,
and the Gaussian tail integral above $u_p$ is bounded by
$C_r(1+u_p^r)$.  Apply this result to
\[
  Z=\nu_i,
  \qquad
  A=\{\DX_i\beta+\nu_i\in\mathcal D_i\},
\]
and combine it with \eqref{eq:mnp-verification-rare-choice-lower-bound}.  This proves
\eqref{eq:supp-mnp-verification-accepted-moment}.
\end{proof}

To state the differentiation result, for a scalar-, vector-, or matrix-valued
function $h(\theta,D_i,\nu)$, define
\begin{equation}\label{eq:mnp-verification-conditional-moment-map}
    M_h(\theta,D_i)
    \coloneqq \mathbb E_\theta^*
      \{h(\theta,D_i,\nu_i^\star(\theta))\mid D_i\}.
\end{equation}

\begin{lemma}[Differentiation of accepted-draw expectations]
\label{lem:mnp-accepted-differentiation}
Under Proposition~\ref{prop:mnp-empirical}(ii) of the main text, suppose that $h$ is twice
continuously differentiable in $\theta$, and that $h$ and its derivatives
through order two are polynomially bounded in $(\DX_i,\nu)$, uniformly over
$\theta\in\mathcal N_*$,
with parameter derivatives taken at fixed latent utility
differences, as in Lemma~\ref{lem:hm-identity} of the main text.
Then $M_h(\theta,D_i)$ is twice continuously
differentiable on $\mathcal N_*$ for each fixed $D_i$. Its first derivative
satisfies, component by component,
\begin{equation}\label{eq:mnp-verification-derivative-identity}
\begin{split}
    \nabla_\theta M_h(\theta,D_i)
    ={}&
    \mathbb E_\theta^*
      \{\nabla_\theta h(\theta,D_i,\nu_i^\star(\theta))\mid D_i\}\\
    &+
    \operatorname{Cov}_\theta^*
    \left(
      h(\theta,D_i,\nu_i^\star(\theta)),
      -\Gamma_i(\theta,\nu_i^\star(\theta))
      \mid D_i
    \right).
\end{split}
\end{equation}
\end{lemma}

\begin{proof}
 Recall that $f_i(w;\theta)= \phi_d(w; \DX_i\beta, \Omega)$.
For a scalar component of \(h\), abbreviate
\(f=f_i(w;\theta)\), \(P=P_{iy_i}(\theta)\),
\(h=h(\theta,D_i,w-\DX_i\beta)\), and
\(\Gamma_{i,k}=\Gamma_{i,k}(\theta,w-\DX_i\beta)\).
The polynomial envelopes and Gaussian tails justify differentiation
twice under the integral; \(P>0\) for each fixed \(D_i\).
With \(w\) fixed,
\[
M_h=\frac1P\int_{\mathcal D_i}hf\,dw,\qquad
\partial_{\theta_k}f=-\Gamma_{i,k}f,\qquad
\partial_{\theta_k}P=-Pg_{i,k}.
\]
Consequently,
\begin{align*}
\partial_{\theta_k}M_h
&=\frac1P\int_{\mathcal D_i}
  (\partial_{\theta_k}h-h\Gamma_{i,k})f\,dw+M_hg_{i,k} \\
&=\mathbb E_\theta^*(\partial_{\theta_k}h\mid D_i)
  -\mathbb E_\theta^*(h\Gamma_{i,k}\mid D_i)
+\mathbb E_\theta^*(h\mid D_i)
         \mathbb E_\theta^*(\Gamma_{i,k}\mid D_i) \\
&=\mathbb E_\theta^*(\partial_{\theta_k}h\mid D_i)
  +\operatorname{Cov}_\theta^*(h,-\Gamma_{i,k}\mid D_i).
\end{align*}
\end{proof}

\begin{proof}[Proof of Proposition~\ref{prop:mnp-verification}]
\textit{Step 1: conditional unbiasedness.}\par
The negative complete-data score for $\beta$ is
$-\DX_i'\Omega^{-1}\nu$, and the negative score matrix for $\Omega$ is
\[
    -\frac12\Omega^{-1}(\nu\nu'-\Omega)\Omega^{-1}.
\]
Applying the chain rule through $\Omega=LL'$ gives the
negative complete-data score $\Gamma_i(\theta,\nu)$. The
conditional-expectation identity therefore gives
\[
    g_i(\theta)
    =
    \mathbb E_\theta^*
    \{\Gamma_i(\theta,\nu_i^\star(\theta))\mid D_i\}.
\]

Compactness of the normalized parameter neighborhood implies that $L$, $\Omega^{-1}$,
and the linear map $\mathcal C_{\theta_L}(\mathord\cdot;L)$ are uniformly
bounded.  Hence
\begin{equation}\label{eq:mnp-verification-Gamma-envelope}
    \|\Gamma_i(\theta,\nu)\|
    \leq
    C\{1+\|\DX_i\|\,\|\nu\|+\|\nu\|^2\}.
\end{equation}
For fixed $D_i$, Lemma~\ref{lem:mnp-accepted-moments} with $r=q$ and $r=2q$
gives, uniformly over $\theta\in\mathcal N_*$,
\begin{align*}
  \mathbb E_\theta^*
    \{\|\Gamma_i(\theta,\nu_i^\star(\theta))\|^q\mid D_i\}
  &\leq C\left\{
      1+\|\DX_i\|^q
        \mathbb E_\theta^*(\|\nu_i^\star(\theta)\|^q\mid D_i)
      +\mathbb E_\theta^*(\|\nu_i^\star(\theta)\|^{2q}\mid D_i)
    \right\}\\
  &\leq C(1+\|\DX_i\|^{2q}).
\end{align*}
Taking expectations over $D_i$ conditional on $\mathcal F_{t-1}$ and using
\eqref{eq:mnp-verification-X-moment} of the main text yields
\[
    \sup_{\theta\in\mathcal N_*}
    \mathbb E\!\left[
      \|\Gamma_i(\theta,\nu_i^\star(\theta))\|^q
      \mid\mathcal F_{t-1}
    \right]
    \leq C.
\]
Jensen's inequality then verifies
Assumption~\ref{ass:online-regularity}(iii).

\medskip
\noindent\textit{Step 2: simulation-variance continuity.}
Holding the complete-data value $w$ fixed, write
$\nu=w-\DX_i\beta$ and put
$\Lambda_i(\nu)\coloneqq 1+\|\DX_i\|+\|\nu\|$. Direct differentiation of
$\Gamma_i(\theta,\nu)$, together with compactness of the parameter neighborhood,
gives the pointwise envelope
\begin{equation}\label{eq:mnp-verification-Gamma-derivative-envelope}
  \|\Gamma_i(\theta,\nu)\|
  +\|\nabla_\theta\Gamma_i(\theta,\nu)\|
  +\|\nabla_\theta^2\Gamma_i(\theta,\nu)\|
  \leq C \Lambda_i(\nu)^2.
\end{equation}
Lemma~\ref{lem:mnp-accepted-differentiation}, applied first to $h=\Gamma_i$ and then to
$h=\Gamma_i\Gamma_i'$, therefore shows that
\begin{align*}
    g_i(\theta)
    &=M_{\Gamma_i}(\theta,D_i),\\
    \operatorname{Var}_\theta^*\{g_{ir}(\theta)\mid D_i\}
    &=M_{\Gamma_i\Gamma_i'}(\theta,D_i)
      -M_{\Gamma_i}(\theta,D_i)
       M_{\Gamma_i}(\theta,D_i)'
\end{align*}
are twice and once continuously differentiable, respectively.  For every
fixed $r>0$, Lemma~\ref{lem:mnp-accepted-moments} implies
\begin{equation}\label{eq:mnp-verification-T-moment}
  \sup_{\theta\in\mathcal N_*}
  \mathbb E_\theta^*\{\Lambda_i(\nu_i^\star(\theta))^r\mid D_i\}
  \leq C_r(1+\|\DX_i\|^r).
\end{equation}
Thus \eqref{eq:mnp-verification-derivative-identity}, its derivative,
and \eqref{eq:mnp-verification-Gamma-derivative-envelope} imply the pointwise-in-$D_i$
bounds
\begin{align}
  \sup_{\theta\in\mathcal N_*}\|g_i(\theta)\|
  &\leq C(1+\|\DX_i\|^2),
  \label{eq:mnp-verification-g-pointwise}\\
  \sup_{\theta\in\mathcal N_*}\|\nabla_\theta g_i(\theta)\|
  &\leq C(1+\|\DX_i\|^4),
  \label{eq:mnp-verification-dg-pointwise}\\
  \sup_{\theta\in\mathcal N_*}\|\nabla_\theta^2g_i(\theta)\|
  &\leq C(1+\|\DX_i\|^6).
  \label{eq:mnp-verification-d2g-pointwise}
\end{align}

Similarly, the conditional simulation variance satisfies
\[
  \operatorname{Var}_\theta^*\{g_{ir}(\theta)\mid D_i\}
  =M_{\Gamma_i\Gamma_i'}-g_ig_i'.
\]
The derivative of $M_{\Gamma_i\Gamma_i'}$ contains the derivative of
$\Gamma_i\Gamma_i'$, bounded by $C\Lambda_i^4$, and its covariance with
$-\Gamma_i$, bounded by a conditional moment of $C\Lambda_i^6$.
The derivative of $g_ig_i'$ is bounded using
\eqref{eq:mnp-verification-g-pointwise} and \eqref{eq:mnp-verification-dg-pointwise}.  Hence
\begin{equation}\label{eq:mnp-verification-dA-pointwise}
  \sup_{\theta\in\mathcal N_*}
  \left\|\nabla_\theta
     \operatorname{Var}_\theta^*\{g_{ir}(\theta)\mid D_i\}\right\|
  \leq C(1+\|\DX_i\|^6).
\end{equation}
Because $q>4$, we have $2q>8$.  Therefore \eqref{eq:mnp-verification-X-moment} of the main text,
\eqref{eq:mnp-verification-dg-pointwise}--\eqref{eq:mnp-verification-dA-pointwise}
give
\begin{equation}\label{eq:mnp-verification-derivative-bounds}
\begin{split}
    \sup_{\theta\in\mathcal N_*}
    \mathbb E\|\nabla_\theta g_i(\theta)\|^2&\leq C,\quad
    \sup_{\theta\in\mathcal N_*}
    \mathbb E\left\|\nabla_\theta
       \operatorname{Var}_\theta^*\{g_{ir}(\theta)\mid D_i\}\right\|\leq C,\\
    \sup_{\theta\in\mathcal N_*}
    \mathbb E\|\nabla_\theta^2 g_i(\theta)\|&\leq C.
\end{split}
\end{equation}
The first inequality verifies Assumption~\ref{ass:online-regularity}(ii).
The second inequality and the integral mean-value formula imply that, for
$\vartheta_1,\vartheta_2\in\mathcal N_*$,
\[
\left\|
  \mathbb E\!\left[
    \operatorname{Var}\{g_{ir}(\vartheta_1)\mid D_i\}
    -\operatorname{Var}\{g_{ir}(\vartheta_2)\mid D_i\}
  \right]
\right\|
\leq C\|\vartheta_1-\vartheta_2\|.
\]
This verifies Assumption~\ref{ass:online-regularity}(iv).

\medskip
\noindent\textit{Step 3: information matrix.} $\sum_{j=1}^JP_{ij}(\theta)=1$ gives
\[
    \mathcal R(\theta_*)=\mathbb E\{g_i(\theta_*)\}=0.
\]
Conditional on $\DX_i$,
$
  \sum_{j=1}^J P_{ij}(\theta)\nabla_\theta\log P_{ij}(\theta)
  =\sum_{j=1}^J\nabla_\theta P_{ij}(\theta)=0.$ Hence
\[
  -\sum_{j=1}^JP_{ij}(\theta)
       \nabla_\theta^2\log P_{ij}(\theta)
  =\sum_{j=1}^J
       \frac{\nabla_\theta P_{ij}(\theta)
             \nabla_\theta P_{ij}(\theta)'}
            {P_{ij}(\theta)}.
\]
Hence
\begin{align*}
    H_*= \mathbb E\{\nabla_\theta g_i(\theta_*)\}
    =\mathbb E\{g_i(\theta_*)g_i(\theta_*)'\}=
    \mathbb E\left[
      \sum_{j=1}^J
      \frac{
        \nabla_\theta P_{ij}(\theta_*)
        \nabla_\theta P_{ij}(\theta_*)'
      }{P_{ij}(\theta_*)}
    \right].
\end{align*}
For any $v\in\mathbb R^p$,
\[
  v' H_*v
  =\mathbb E\left[
      \sum_{j=1}^J
      \frac{\{v'\nabla_\theta P_{ij}(\theta_*)\}^2}
           {P_{ij}(\theta_*)}
    \right].
\]
This nonnegative quantity can equal zero only if
$v'\nabla_\theta P_{ij}(\theta_*)=0$ almost surely for every $j$.
Condition \eqref{eq:mnp-verification-rank} of the main text then gives $v=0$.  Because $p$ is fixed,
$H_*$ is therefore positive definite, and
$c_I\coloneqq \lambda_{\min}(H_*)>0$.
Applying \eqref{eq:mnp-verification-derivative-identity} with
$h=\Gamma_i$ gives the identity of \citet{louis1982finding},
\begin{equation}\label{eq:mnp-verification-louis}
    \nabla_\theta g_i(\theta)
    =
    \mathbb E_\theta^*
      \{\nabla_\theta\Gamma_i(\theta,\nu_i^\star(\theta))\mid D_i\}
    -\operatorname{Var}_\theta^*\{g_{ir}(\theta)\mid D_i\}.
\end{equation}
Thus observed information is complete-data information minus missing
information. To make the complete-data term explicit, let
$\Omega_*\coloneqq\Omega(\theta_{L*})$, $Q_*\coloneqq \Omega_*^{-1}$,
and define the Jacobian of the normalized variance matrix
\[
  J_{\Omega,*}
  \coloneqq
  \left.
  \frac{\partial\,\operatorname{vec}\Omega(\theta_L)}
       {\partial\theta_L'}
  \right|_{\theta_L=\theta_{L*}}.
\]
The normalized complete-data information is
\begin{equation}\label{eq:mnp-verification-complete-information}
  I_{\mathrm c}(\theta_*)
  \coloneqq
  \begin{pmatrix}
    \mathbb E(\DX_i' Q_*\DX_i) & 0\\[3pt]
    0 & \dfrac12 J_{\Omega,*}'
               (Q_*\otimes Q_*)J_{\Omega,*}
  \end{pmatrix}.
\end{equation}
The cross block is zero because centered Gaussian third moments vanish.
Taking expectations in \eqref{eq:mnp-verification-louis} yields
\begin{equation}\label{eq:mnp-verification-observed-minus-missing}
  H_*
  =I_{\mathrm c}(\theta_*)
   -\mathbb E\!\left[\operatorname{Var}\{g_{ir}(\theta_*)\mid D_i\}\right].
\end{equation}

The third bound in \eqref{eq:mnp-verification-derivative-bounds} justifies differentiation
under the expectation and gives
\[
  \nabla_\theta H(\theta)=\mathbb E\{\nabla_\theta^2g_i(\theta)\},
  \qquad
  \sup_{\theta\in\mathcal N_*}\|\nabla_\theta H(\theta)\|\leq C.
\]
Since \(H_*\succ0\), the inverse function theorem gives a
neighborhood \(\mathcal N_1\subseteq\mathcal N_*\) on which
\(\theta_*\) is the unique zero of \(\mathcal R\), completing
Assumption~\ref{ass:online-regularity}(i).
\par

\medskip
\begin{samepage}
\noindent\textit{Step 4: variance nondegeneracy.}
Under the online sampling scheme in Proposition~\ref{prop:mnp-verification}
and independent accepted conditional draws, the one-step variance at $\theta_*$ is constant over
iterations and decomposes as
\[
    V_{m,\Rsim}
    =
    \frac1m\operatorname{Var}\{g_i(\theta_*)\}
    +\frac1{m\Rsim}
      \mathbb E\!\left[\operatorname{Var}\{g_{ir}(\theta_*)\mid D_i\}\right].
\]
\end{samepage}
This yields
\eqref{eq:mnp-verification-variance} of the main text.  The second term is positive semidefinite, so
\[
    \lambda_{\min}(V_{m,\Rsim})
    \geq \frac1m\lambda_{\min}(H_*)
    \geq \frac{c_I}{m}.
\]
This verifies Assumption~\ref{ass:online-regularity}(v) and completes the proof.
\end{proof}

\begingroup
\begin{proof}[Proof of Proposition~\ref{prop:mnp-empirical} of the main text]
Take a sufficiently small closed convex neighborhood \(\mathcal K\) of
\(\theta_*\) inside \(\mathcal N_*\). The proof of
Proposition~\ref{prop:mnp-verification} gives \(H_*\succ0\),
\(\mathbb E g_i(\theta_*)=0\), and
\(\Sigma_g=H_*\). By continuity, \(H(\theta)\) is uniformly positive
definite on \(\mathcal K\). The envelopes
\eqref{eq:mnp-verification-g-pointwise}--\eqref{eq:mnp-verification-dA-pointwise}
and a finite-cover law-of-large-numbers argument give, uniformly on \(\mathcal K\),
\[
\begin{split}
    \nabla_\theta\mathcal R_n(\theta)
       &\xrightarrow{p}H(\theta),\\
    \frac1n\sum_{i=1}^n g_i(\theta)g_i(\theta)'
       &\xrightarrow{p}\mathbb E\{g_i(\theta)g_i(\theta)'\},\\
    \frac1n\sum_{i=1}^n
       \operatorname{Var}\{g_{ir}(\theta)\mid D_i\}
       &\xrightarrow{p}
       \mathbb E\!\left[\operatorname{Var}\{g_{ir}(\theta)\mid D_i\}\right].
\end{split}
\]

With probability tending to one, the empirical negative log-likelihood
is strictly convex on \(\mathcal K\).
For a fixed small \(b>0\) whose closed ball about \(\theta_*\)
lies in \(\operatorname{int}\mathcal K\), with probability tending to one,
\[
\inf_{\|v\|=b}v'\mathcal R_n(\theta_*+v)
\geq cb^2-b\|\mathcal R_n(\theta_*)\|>0,
\]
\par
since \(\mathcal R_n(\theta_*)=O_p(n^{-1/2})\).
The unique minimizer on this ball is interior and satisfies
\[
\|\widehat\theta_n-\theta_*\|
\leq c^{-1}\|\mathcal R_n(\theta_*)\|=O_p(n^{-1/2}).
\]
\par
Expanding the sample score equation with the uniformly convergent Jacobian
proves \eqref{eq:sample-root-expansion} of the main text and \(H_n\xrightarrow{p}H_*\).
For a sufficiently small fixed \(\rho\), \(B_{n,\rho}\subset\mathcal K\)
with probability tending to one. Taylor expansion and
\eqref{eq:mnp-verification-d2g-pointwise} give
\[
    \|\mathcal R_n(\theta)-H_n(\theta-\widehat\theta_n)\|
    \leq
    \left\{\frac Cn\sum_{i=1}^n(1+\|\DX_i\|^6)\right\}
    \|\theta-\widehat\theta_n\|^2.
\]
Differentiating \eqref{eq:empirical-vmr} of the main text and using
\eqref{eq:mnp-verification-g-pointwise}--\eqref{eq:mnp-verification-dA-pointwise}
gives
\[
\sup_{\theta\in\mathcal K}
\|\nabla_\theta V_{n,m,\Rsim}(\theta)\|
\leq\frac Cn\sum_{i=1}^n(1+\|\DX_i\|^6).
\]
\par
The law of large numbers bounds the right-hand side by a fixed
constant with probability tending to one.
Uniform convergence and consistency yield
\[
    V_{n,m,\Rsim}(\widehat\theta_n)
    \xrightarrow{p}
    \frac1m H_*+\frac1{m\Rsim}
      \mathbb E\!\left[\operatorname{Var}\{g_{ir}(\theta_*)\mid D_i\}\right]
    \succ0.
\]
Finally, for any fixed \(k\in[2,q]\),
\[
    \sup_{\theta\in\mathcal K}
    \mathbb E_n^*\|\widehat G_t(\theta)-\mathcal R_n(\theta)\|^k
    \leq\frac Cn\sum_{i=1}^n(1+\|\DX_i\|^{2q}).
\]
This establishes conditions (i)--(iii) of
Assumption~\ref{ass:empirical-regularity} of the main text.
Taking \(k=2p_0\leq q\) gives \eqref{eq:empirical-path-moment},
verifying condition (iv) with the existing \(p_0>(1-a)^{-1}\).
\end{proof}

\paragraph{Consistency of the MNP variance estimators.}
For multi-pass SAUSS, condition on the data set and \(\bar\theta_T\). The
post-estimation directions
\(\widehat g_{\Rsim}^{(s)}(\bar\theta_T;D_i)\) are independent across
observations and replicates under that conditioning. Fix \(S\geq2\). On
\(\{\bar\theta_T\in\mathcal K\}\), their conditional fourth moments are bounded by
\(C(1+\|\DX_i\|^8)\). Thus the conditional variances of the averages of
their first and second products are bounded by
\(Cn^{-2}\sum_i(1+\|\DX_i\|^8)=O_p(n^{-1})\).
For every \(i\),
\[
\begin{aligned}
    &\mathbb E\!\left[
       \bar g_i\bar g_i'-\frac1{S(S-1)}\sum_{s=1}^S
       \{\widehat g_{\Rsim}^{(s)}(\bar\theta_T;D_i)-\bar g_i\}
       \{\widehat g_{\Rsim}^{(s)}(\bar\theta_T;D_i)-\bar g_i\}'
       \,\middle|\,D_{1:n},\bar\theta_T\right]\\
    &\qquad=g_i(\bar\theta_T)g_i(\bar\theta_T)',\\
    &\mathbb E\!\left[
       \frac{\Rsim}{S-1}\sum_{s=1}^S
       \{\widehat g_{\Rsim}^{(s)}(\bar\theta_T;D_i)-\bar g_i\}
       \{\widehat g_{\Rsim}^{(s)}(\bar\theta_T;D_i)-\bar g_i\}'
       \,\middle|\,D_{1:n},\bar\theta_T\right]\\
    &\qquad=\operatorname{Var}\{g_{ir}(\bar\theta_T)\mid D_i\}.
\end{aligned}
\]
Chebyshev's inequality, the uniform laws just proved, and
\(\|\bar\theta_T-\theta_*\|\mathbb I\{\mathcal E_n\}
\xrightarrow{p}0\) therefore give
\[
\begin{gathered}
    \|\widehat\Sigma_g-\Sigma_g\|\mathbb I\{\mathcal E_n\}
       \xrightarrow{p}0,\\
    \left\|
    \begin{aligned}
       &\frac{\Rsim}{n(S-1)}\sum_{i=1}^n\sum_{s=1}^S
       \{\widehat g_{\Rsim}^{(s)}(\bar\theta_T;D_i)-\bar g_i\}
       \{\widehat g_{\Rsim}^{(s)}(\bar\theta_T;D_i)-\bar g_i\}'\\
       &\quad-\mathbb E\!\left[
          \operatorname{Var}\{g_{ir}(\theta_*)\mid D_i\}\right]
    \end{aligned}
    \right\|\mathbb I\{\mathcal E_n\}\xrightarrow{p}0.
\end{gathered}
\]
Since \(\Sigma_g=H_*\succ0\), continuity of matrix inversion and
\eqref{eq:mnp-simulation-plugin} of the main text give
\[
    \|\widehat\Sigma_{\mathrm{sim}}-\Sigma_{\mathrm{sim}}\|
       \mathbb I\{\mathcal E_n\}\xrightarrow{p}0.
\]
Equations~\eqref{eq:vmr} and \eqref{eq:mnp-vmr-plugin} of the main text then imply
\(\|\widehat V_{m,\Rsim}^{\mathrm{MNP}}-V_{m,\Rsim}\|
\mathbb I\{\mathcal E_n\}\xrightarrow{p}0\).
For online sampling, the same argument on the i.i.d.\ evaluation
observations, with \(\mathcal E\) in place of \(\mathcal E_n\),
gives the variance consistency required by
Corollary~\ref{cor:inference}.

\endgroup

\par
\endgroup

\section{Monte Carlo experiments: full designs and results}\label{app:supp-mc}

This appendix gives the complete designs and results of the Monte Carlo experiments summarized in
Section~\ref{sec:monte-carlo-results} of the main text. Section~\ref{sec:supp-mc-design} describes the designs and
the tuning, which later sections do not repeat. Section~\ref{sec:supp-mc-metrics} defines the performance criteria
and compares SAUSS with GHK SML as \(J\) grows. Section~\ref{sec:mc-robust} reports two
robustness designs, Section~\ref{sec:supp-mc-inference} the inference experiments, and Section~\ref{sec:mc-simulator}
the comparison of conditional samplers.

\subsection{Simulation design}\label{sec:supp-mc-design}\label{sec:mc-design}

All experiments use the hybrid sampler (except in the comparison of conditional samplers) and the variance parameterization \(\Omega=LL'\), identified
by the normalization \(L_{11}=1\). A sequential accept--reject variant produced
similar results and we do not report it separately.

The baseline design, labeled D1, is a multinomial probit model written in utility differences relative to
alternative~1. For individual \(i=1,\ldots,n\) and alternatives \(j=2,\ldots,J\),
\begingroup
\let\mcDGPLabel\label % Preserve the old alias without two amsmath display labels.
\begin{equation}\label{eq:supp-mc-dgp}\mcDGPLabel{eq:mc-dgp}
U_{ij}-U_{i1} = (x_{ij}-x_{i1})'\beta_0+\nu_{ij}, \qquad \nu_i=(\nu_{i2},\ldots,\nu_{iJ})'\sim
\mathcal{N}(0,\Omega_0), \end{equation}
\endgroup
and the observed choice \(y_i\) is the alternative with the largest latent
utility. Each alternative has \(K=3\) covariates drawn independently from the standard normal distribution, and the
true coefficient vector is \(\beta_0=(-0.5,1,1)'\). The differenced errors are equicorrelated: \(\Omega_0\) has unit
diagonal and all off-diagonal elements equal to \(0.3\). The number of alternatives varies over
\(J\in\{4,8,16\}\), so that the dimension of \(\Omega_0\) ranges from \(3\times 3\) to
\(15\times 15\).
The sample size is \(n=2{,}000\) in Sections~\ref{sec:supp-mc-metrics}, \ref{sec:mc-robust},
and~\ref{sec:mc-simulator}, and Section~\ref{sec:supp-mc-metrics} repeats the increasing-\(J\) experiment at
\(n=20{,}000\). In the inference experiments of Section~\ref{sec:supp-mc-inference}, \(n\) ranges from \(2{,}000\) to
\(10.24\) million. Every SAUSS design uses \(1{,}000\) replications, and every reported time is the median over
replications on one core of an AMD EPYC 9655 processor. Times exclude simulated-data generation
and subsequent Monte Carlo evaluation.

The implementation follows Section~\ref{sec:mnp-sauss} of the main text. At each iteration we draw a
mini-batch of \(m=10\) observations from the sample with replacement and obtain \(\Rsim=5\) exact conditional draws
for each of them. Accept--reject runs for at most \(K_0=8\Rsim J\) proposals, and the tilting sampler of
\citet{botev2017normal} supplies any missing draws (Appendix~\ref{sec:supp-hybrid}). The hybrid sampler left no
observation without accepted draws, so no update omits a contribution. The learning rates decay as \(t^{-0.501}\),
with constant \(\gamma_\beta\) for \(\beta\) and \(\gamma_L\) for the free Cholesky parameters.

The experiments use one of two tuning schemes. The scaling, robustness, and sampler experiments of
Sections~\ref{sec:supp-mc-metrics}, \ref{sec:mc-robust}, and~\ref{sec:mc-simulator} use fixed tuning. They run
\(T=80{,}000\) updates with \(\gamma_\beta=\gamma_L=0.2\) from the start \(\beta=0\) and \(\Omega=I\), with no
pilot, in every design. Section~\ref{sec:supp-mc-metrics} adds runs with \(T=800{,}000\) at \(n=20{,}000\), and the
sampler experiment changes the start by design. The multi-pass inference experiments of
Section~\ref{sec:supp-mc-inference} use calibrated tuning, following Section~\ref{sec:supp-numerical-settings}. A
pilot of \(40\) passes over the first \(500\) observations, from \(\beta=0\) and \(\Omega=I\) with
\((\gamma_\beta,\gamma_L)=(0.1,0.2)\), gives the starting value. One bank of \(S=100\) score replicates on
\(\min(n,2{,}000)\) observations at the pilot estimate gives the learning-rate scale by the curvature rule
\eqref{eq:numerical-settings-curvature} with \(c=0.5\) and \(\gamma_L=2\gamma_\beta\). Where the run length is
planned, the same bank gives \(T_{\mathrm{plan}}\) by \eqref{eq:numerical-settings-horizon}. The online experiment
instead uses the fixed constants \((\gamma_\beta,\gamma_L)=(0.1,0.2)\) and a pilot on \(500\) observations outside
the stream. Section~\ref{sec:mc-robust} describes the two robustness designs, D2 and D3, which change \(\beta_0\) and
\(\Omega_0\) at \(J=4\).

\subsection{Performance criteria and increasing-J results}\label{sec:supp-mc-metrics}\label{sec:mc-crossj}

We evaluate the estimates by three root mean squared error (RMSE) criteria. All three are invariant
to the scale normalization, so they also allow comparisons with implementations that normalize differently. The
coefficient criterion measures the recovery of \(\beta_0\) up to scale. Writing
\(\tilde\beta\coloneqq \beta/\|\beta\|\) and letting \({ \widehat\beta}_b\) denote the estimate in replication
\(b=1,\ldots,B\),
\[
    \mathrm{RMSE}_{\beta}
    \coloneqq
    \Bigl\{ B^{-1}\textstyle\sum_{b=1}^B \|\tilde{{ \widehat\beta}}_b-\tilde\beta_0\|^2 \Bigr\}^{1/2}.
\]
The correlation criterion measures the recovery of the error correlations. Let \(C(\Omega)\)
denote the correlation matrix implied by \(\Omega\), and let \(d=J-1\). Then
\[
    \mathrm{RMSE}_{C}
    \coloneqq
    \Bigl\{ \tfrac{2}{Bd(d-1)}\textstyle\sum_{b=1}^B\sum_{k<\ell}
    \bigl(C_{k\ell}({ \widehat\Omega}_b)-C_{k\ell}(\Omega_0)\bigr)^2 \Bigr\}^{1/2}.
\]
The probability criterion measures the out-of-sample fit of the implied choice probabilities. For
each \(J\) we fix a test design of \(N_{\mathrm{test}}={ 5{,}000}\) covariate draws from the distribution above,
independent of all estimation samples, and approximate the choice probabilities \(P_{qj}(\theta)\) at any parameter
value by frequency simulation with \(S_{\mathrm{eval}}={ 5{,}000}\) common random draws. The criterion is
\[
    \mathrm{RMSE}_{P}
    \coloneqq
    \Bigl\{ \tfrac{1}{BN_{\mathrm{test}}J}\textstyle\sum_{b=1}^B\sum_{q=1}^{N_{\mathrm{test}}}\sum_{j=1}^J
    \bigl({ \widehat P}_{qj}(\widehat\theta_b)-{ \widehat P}_{qj}({\theta_*})\bigr)^2 \Bigr\}^{1/2},
\]
where the same test design and the same evaluation draws are used for all replications and all estimators within
a given \(J\), so that simulation noise in the evaluation step is common across the comparison.
Because
average choice probabilities scale as \(1/J\), levels of \(\mathrm{RMSE}_P\) are not comparable across values of
\(J\), and we use \(\mathrm{RMSE}_P\) only for comparisons within a given \(J\).

At every \(J\) we compare SAUSS with GHK SML, simulated maximum likelihood based on the
Geweke--Hajivassiliou--Keane (GHK) simulator \citep{borschSupanHajivassiliou1993,hajivassiliou1996simulation,Train2009}.
As in the applications, the fit uses \(2{,}000\) GHK draws per observation and the \texttt{mlogit} defaults
for the starting values (a rescaled multinomial logit fit), the optimizer, and
the convergence tolerances. Each fit has a limit of five hours. Both estimators use the normalization \(L_{11}=1\)
and the same simulated data sets. We count fits that return an estimate within the limit, irrespective of
whether the optimizer reports normal termination. SAUSS completes its prescribed updates; completion itself is not a
convergence test for either estimator. A GHK SML fit returns no estimate if the line
search meets a non-finite simulated log-likelihood, if a step makes the implied variance matrix of the utility
differences non-positive definite, or if the fit reaches the time limit. Table~\ref{tab:mc-ghk} reports the results
for \(n=2{,}000\), and Table~\ref{tab:mc-ghk-n20k} for \(n=20{,}000\), where SAUSS runs with \(T=80{,}000\) and
\(T=800{,}000\).

\begin{table}[!b]
\centering
\caption{SAUSS and GHK SML across \(J\), design D1, \(n=2{,}000\)}
\label{tab:mc-ghk}
\spacingset{1}\small
\setlength{\tabcolsep}{4pt}
\begin{tabular*}{\textwidth}{@{\extracolsep{\fill}}rlrcccr@{}}
\toprule
\(J\) & Estimator & \shortstack{Estimates\\returned} & \(\mathrm{RMSE}_{\beta}\) & \(\mathrm{RMSE}_{C}\) & \(\mathrm{RMSE}_{P}\)
 & Time (s) \\
\midrule
4  & SAUSS & 1,000 & 0.0200 & 0.0895 & 0.0118 & 1.6 \\
   & GHK SML & 1,000 & 0.0200 & 0.0942 & 0.0120 & 15.0 \\
\midrule
8  & SAUSS & 1,000 & 0.0171 & 0.1407 & 0.0113 & 6.3 \\
   & GHK SML & 923   & 0.0171 & 0.1719 & 0.0136 & 168.1 \\
\midrule
16 & SAUSS & 1,000 & 0.0157 & 0.2421 & 0.0099 & 19.9 \\
   & GHK SML & 0     & --     & --     & --     & -- \\
\bottomrule
\end{tabular*}
\par
\begin{minipage}{\textwidth}
\vspace{1ex}
\footnotesize
\emph{Notes:} Estimates returned: number of replications returning an estimate within the five-hour
limit. RMSEs are computed over replications with returned estimates. GHK SML optimizer termination is reported
separately below. Time (s): median fitting time, including the logit start for GHK SML; no numerical
Hessian or SAUSS plug-in variance is computed.
\end{minipage}
\end{table}

\begin{table}[!htb]
\centering
\caption{SAUSS and GHK SML across \(J\), design D1, \(n=20{,}000\)}
\label{tab:mc-ghk-n20k}
\spacingset{1}\small
\setlength{\tabcolsep}{4pt}
\begin{tabular*}{\textwidth}{@{\extracolsep{\fill}}rlrrcccr@{}}
\toprule
\(J\) & Estimator & \(T\) & \shortstack{Estimates\\returned} & \(\mathrm{RMSE}_{\beta}\) & \(\mathrm{RMSE}_{C}\) & \(\mathrm{RMSE}_{P}\)
 & Time (s) \\
\midrule
4  & SAUSS & 80,000  & 1,000 & 0.0066 & 0.0342 & 0.0043 & 1.7 \\
   & SAUSS & 800,000 & 1,000 & 0.0064 & 0.0290 & 0.0039 & 19.6 \\
   & GHK SML & --      & 1,000 & 0.0064 & 0.0289 & 0.0039 & 143.6 \\
\midrule
8  & SAUSS & 80,000  & 1,000 & 0.0056 & 0.0872 & 0.0053 & 6.8 \\
   & SAUSS & 800,000 & 1,000 & 0.0055 & 0.0526 & 0.0043 & 65.6 \\
   & GHK SML & --      & 1,000 & 0.0055 & 0.0532 & 0.0044 & 943.2 \\
\midrule
16 & SAUSS & 80,000  & 1,000 & 0.0050 & 0.2189 & 0.0065 & 21.1 \\
   & SAUSS & 800,000 & 1,000 & 0.0047 & 0.1190 & 0.0049 & 200.8 \\
   & GHK SML & --      & 0 & -- & --     & --     & -- \\
\bottomrule
\end{tabular*}
\par
\begin{minipage}{\textwidth}
\vspace{1ex}
\footnotesize
\emph{Notes:} \(T\): SAUSS run length. At \(J=16\), the GHK SML batch was stopped
early; the summary uses \(283\) available fit records, including interrupted fits, with no returned estimate.
Other entries as in Table~\ref{tab:mc-ghk}.
\end{minipage}
\end{table}

SAUSS completes every replication at every \(J\) and at both sample sizes. With \(n\) fixed, its
coefficient criterion improves as \(J\) grows, from \(0.0200\) at \(J=4\) to \(0.0157\) at \(J=16\) when
\(n=2{,}000\). Each observation contributes a comparison among \(J\) alternatives, so the information per observation
about the coefficients increases with \(J\). The correlation criterion worsens as \(J\) grows, because the
number of free correlation parameters, \(d(d-1)/2\), grows quadratically in \(J\). For reference, the start
\(\Omega=I\) sets every correlation to zero and has correlation criterion \(0.300\) in design D1. At \(J=16\) the
criterion is \(0.242\) when \(n=2{,}000\), and at \(n=20{,}000\) the tenfold longer run halves it, from \(0.219\) to
\(0.119\). The probability criterion shows that the choice probabilities are recovered at every \(J\). The median time
rises with \(J\) but stays below \(20\) seconds at \(J=16\) and \(n=2{,}000\). At \(J=16\) and \(n=20{,}000\),
\(800{,}000\) updates take \(200.8\) seconds.

GHK SML returns an estimate in every data set at \(J=4\) and in \(923\) of the \(1{,}000\) at \(J=8\) when \(n=2{,}000\), and in
every data set at \(J=4\) and \(J=8\) when \(n=20{,}000\). With \(n=2{,}000\), the optimizer reports normal
termination in \(976\) of the \(1{,}000\) fits at \(J=4\) and in \(903\) of the \(923\) fits returning an estimate at \(J=8\). At
\(J=4\) the two estimators have the same \(\mathrm{RMSE}_\beta\) and similar \(\mathrm{RMSE}_C\) and
\(\mathrm{RMSE}_P\). At \(J=8\) and \(n=2{,}000\), on the \(923\) data sets where GHK SML returns an estimate, both have
\(\mathrm{RMSE}_\beta=0.0171\), but GHK SML is less accurate for the correlations (\(0.172\) against \(0.139\)) and the
choice probabilities (\(0.0136\) against \(0.0112\)). At \(n=20{,}000\), SAUSS with \(T=800{,}000\) matches GHK SML on
every criterion and is \(7\) times faster at \(J=4\) and \(14\) times faster at \(J=8\). With \(T=80{,}000\) it is
\(84\) times faster at \(J=4\) and about \(140\) times faster at \(J=8\), with a
somewhat larger correlation criterion.

At \(J=16\), GHK SML returns no estimate among the recorded runs at either sample size.
The numerical failures involve a non-finite simulated log-likelihood in the line search or a
non-positive definite variance matrix. These failures are partly properties of the optimizer rather than of the GHK simulator, and a different start
could avoid some of them. With the tuning held fixed across \(J\), the larger choice set thus separates the two
implemented fitting procedures by their ability to return estimates rather than by accuracy.

\subsection{\texorpdfstring{Heterogeneous variance matrices and a rare alternative}{Heterogeneous variance matrices and a rare alternative}}\label{sec:mc-robust}

Two further designs at \(J=4\) depart from the equicorrelated baseline. Design D2 raises the
coefficients to \(\beta_0=(-1.5,2,2)'\), which makes choices more deterministic given the covariates, and replaces
\(\Omega_0\) with the matrix
\begin{equation}\label{eq:mc-omega-d2}
    \Omega_0^{\mathrm{D2}}
    =
    \begin{pmatrix}
        1.00 & 0.85 & 0.75 \\
        0.85 & 1.40 & 1.05 \\
        0.75 & 1.05 & 2.00
    \end{pmatrix},
\end{equation}
which has unequal variances and large, uneven correlations. Design D3 adds to D2 the
alternative-specific intercepts \(\delta_0=(1.5,0.5,-2.0)'\) in the utility differences, so that alternative~4 is
rarely chosen. Its population choice share is \(9.2\) percent, against \(41\) percent for the most popular
alternative. We absorb the intercepts into the covariate matrix as alternative dummies, with no change to the
estimator, and report the coefficient criterion separately for the slope and intercept blocks.

\begin{table}[!htb]
\centering
\caption{SAUSS performance under the baseline and robustness designs at \(J=4\)}
\label{tab:mc-robust}
\spacingset{1}\small
\setlength{\tabcolsep}{4pt}
\begin{tabular*}{\textwidth}{@{\extracolsep{\fill}}lccccc@{}}
\toprule
Design & \(\mathrm{RMSE}_{\beta}\) & \(\mathrm{RMSE}_{\delta}\) & \(\mathrm{RMSE}_{C}\) & \(\mathrm{RMSE}_{P}\)
& Time (s) \\
\midrule
D1 & 0.0200 & --     & 0.0895 & 0.0118
    & 1.6 \\
D2 & 0.0140 & --     & 0.0804 & 0.0164
    & 1.2 \\
D3 & 0.0138 & 0.0696 & 0.0986
    & 0.0239 & 1.3 \\
\bottomrule
\end{tabular*}
\par
\begin{minipage}{\textwidth}
\vspace{1ex}
\footnotesize
\emph{Notes:} D1: equicorrelated baseline. \(\mathrm{RMSE}_{\delta}\): coefficient criterion
for the intercept block (D3 only). Other criteria as in Table~\ref{tab:mc-ghk}.
\end{minipage}
\end{table}

Heterogeneous variances and a rare alternative leave the coefficient and correlation criteria close to
their baseline values and raise only the probability criterion (Table~\ref{tab:mc-robust}). The coefficient criterion
falls from \(0.0200\) in D1 to \(0.0140\) in D2 and \(0.0138\) in D3, and the correlation criterion lies between
\(0.0804\) and \(0.0986\). The probability criterion rises from \(0.0118\) in D1 to \(0.0164\) in D2 and \(0.0239\) in
D3, reflecting the more extreme choice probabilities. Because the data-generating processes differ, these comparisons
are robustness checks rather than a ranking of the designs.

\subsection{Inference}\label{sec:supp-mc-inference}\label{sec:mc-inference}

With the full plug-in variance, multi-pass confidence intervals cover close to the nominal level for
every target from \(n=2{,}000\) to \(2\) million. The naive interval undercovers once the run makes few passes, as the
variance decomposition \eqref{eq:multipass-three-components} of the main text predicts.
Lengthening the run reduces, but does not eliminate, this undercoverage.
Online intervals cover the coefficient ratio well. For the other targets, online plug-in coverage improves with \(n\)
but some undercoverage remains, and random scaling needs a still larger \(n\).

The experiments use design D1 with \(J=4\). They evaluate the multi-pass intervals of Section~\ref{sec:inference} of
the main text and the online intervals of Appendix~\ref{sec:inference-rs}. An online run makes one pass over a fresh
stream of \(n=T\) observations with mini-batches of \(m=1\). A multi-pass run makes repeated passes over one observed
sample of \(n\) observations. The main targets are the coefficient ratio \(\beta_2/\beta_3\) and the choice
probability \(P_2(x_0)\) of alternative~2 at a covariate value \(x_0\) drawn once from the covariate distribution.
Both are scale-free, and they are the quantities that most multinomial probit applications report. We also report the
coefficient \(\beta_2\), the variance \(\Omega_{22}\) of the second utility difference, and the correlation
\(C_{23}\) between the second and third utility differences.

In the online experiment, \(n=T\in\{0.32,\,1.28,\,5.12,\,10.24\}\times10^6\). We compare the plug-in interval based on
\(T^{-1}\widehat\Sigma_{\mathrm{alg}}\), with \(\widehat\Sigma_{\mathrm{alg}}\) estimated on the first \(200{,}000\)
stream observations, and the random-scaling interval of Corollary~\ref{cor:inference}. In the multi-pass experiment,
\(n\) is \(2{,}000\), \(20{,}000\), \(0.2\) million, or \(2\) million, and \(T=1.28\) million, so the run makes between
\(6.4\) and \(6{,}400\) passes over the data. At \(n=2\) million we also evaluate the same path at \(T=5.12\) million
(\(25.6\) passes). We compare the plug-in interval based on \(\widehat C_{n,T}\) of Section~\ref{sec:mnp-covariance}
of the main text with the naive interval, which uses the sampling component \(\widehat\Sigma_{\mathrm{stat}}/n\)
alone. In the multi-pass experiment the curvature rule gives a median \(\gamma_\beta\) of \(0.30\) at every sample
size, with \(5\)th and \(95\)th percentiles \(0.24\) to \(0.25\) and \(0.36\).

For each target we report the Monte Carlo bias relative to the Monte Carlo standard deviation of the estimates
(Bias/SD), the mean estimated standard error relative to that standard deviation (SE/SD), and the coverage of nominal
\(95\%\) intervals. For the multi-pass experiment we also report the shares of the estimated variance due to the
sampling, mini-batch, and simulation components. The Monte Carlo standard error of a coverage of \(0.95\) is
\(0.007\).

\begin{figure}[!htb]
\centering

\caption{Online SAUSS: coverage of nominal 95\% intervals against \(n=T\)}
\label{fig:mc-inf-online}
\includegraphics[width=\textwidth]{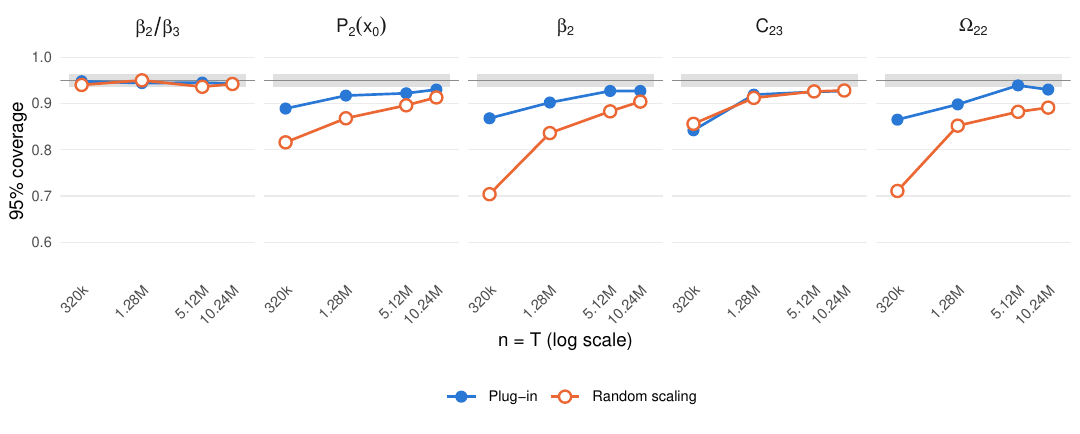}
\par
\begin{minipage}{\textwidth}
\vspace{1ex}
\footnotesize
\emph{Notes:} \(n=T\) on a log scale. Blue:
plug-in; orange: random scaling. The grey band is \(0.95\) plus or minus two Monte Carlo standard errors.
\end{minipage}
\end{figure}

\begin{table}[!htb]
\centering

\caption{Online SAUSS: coverage of nominal 95\% intervals, design D1, \(n=T\)}
\label{tab:mc-inf-online}
\spacingset{1}\small
\begin{tabular*}{\textwidth}{@{\extracolsep{\fill}}llccccc@{}}
\toprule
Target & \(n=T\) & Bias/SD & SE/SD & Plug-in & RS & Length RS/plug-in \\
\midrule
\(\beta_2/\beta_3\) & 0.32M  & 0.01 & 0.99 & 0.948 & 0.940 & 1.30 \\
                    & 1.28M  & 0.01 & 0.97 & 0.944 & 0.950 & 1.30 \\
                    & 5.12M  & \(-\)0.01 & 0.99 & 0.945 & 0.936 & 1.31 \\
                    & 10.24M & 0.00 & 0.97 & 0.943 & 0.942 & 1.28 \\
\addlinespace
\(P_2(x_0)\)        & 0.32M  & \(-\)0.06 & 0.78 & 0.889 & 0.816 & 1.02 \\
                    & 1.28M  & \(-\)0.12 & 0.90 & 0.917 & 0.868 & 1.15 \\
                    & 5.12M  & \(-\)0.19 & 0.91 & 0.922 & 0.896 & 1.21 \\
                    & 10.24M & \(-\)0.27 & 0.98 & 0.930 & 0.913 & 1.20 \\
\addlinespace
\(\beta_2\)         & 0.32M  & 0.10 & 0.76 & 0.868 & 0.704 & 0.84 \\
                    & 1.28M  & 0.16 & 0.86 & 0.902 & 0.836 & 1.10 \\
                    & 5.12M  & 0.29 & 0.96 & 0.927 & 0.883 & 1.17 \\
                    & 10.24M & 0.32 & 0.97 & 0.927 & 0.904 & 1.19 \\
\addlinespace
\(C_{23}\)          & 0.32M  & 0.01 & 0.65 & 0.842 & 0.856 & 1.29 \\
                    & 1.28M  & 0.07 & 0.89 & 0.919 & 0.912 & 1.29 \\
                    & 5.12M  & 0.14 & 0.81 & 0.925 & 0.926 & 1.33 \\
                    & 10.24M & 0.20 & 0.89 & 0.927 & 0.928 & 1.29 \\
\addlinespace
\(\Omega_{22}\)     & 0.32M  & 0.09 & 0.75 & 0.865 & 0.711 & 0.88 \\
                    & 1.28M  & 0.15 & 0.86 & 0.898 & 0.852 & 1.12 \\
                    & 5.12M  & 0.25 & 0.98 & 0.939 & 0.882 & 1.20 \\
                    & 10.24M & 0.25 & 0.96 & 0.930 & 0.891 & 1.21 \\
\bottomrule
\end{tabular*}
\par
\begin{minipage}{\textwidth}
\vspace{1ex}
\footnotesize
\emph{Notes:} M: million. Plug-in: interval based on \(T^{-1}\widehat\Sigma_{\mathrm{alg}}\);
RS: random scaling. Bias/SD: Monte Carlo bias over the Monte Carlo standard deviation of the estimates. SE/SD: mean
plug-in standard error over that standard deviation. Length: mean interval length. All entries refer to the average
\(\bar\theta_T\) of all iterates.
\end{minipage}
\end{table}

In the online experiment the algorithmic variance is the entire variance
(Table~\ref{tab:mc-inf-online} and Figure~\ref{fig:mc-inf-online}). For the coefficient ratio, the plug-in interval
covers between \(0.943\) and
\(0.948\) and the random-scaling interval between \(0.936\) and \(0.950\) at every \(n\), and SE/SD is between
\(0.97\) and \(0.99\). For the other targets, plug-in coverage is between \(0.842\) and \(0.919\) at
\(n\leq1.28\) million and between \(0.922\) and \(0.939\) from \(n=5.12\) million on. The plug-in standard error is too
small at the smallest \(n\), with SE/SD between \(0.65\) and \(0.78\). From \(n=5.12\) million on, it is within
\(10\%\) of the Monte Carlo standard deviation for every target except \(C_{23}\).
At the largest \(n\), residual bias and some underestimation of standard errors accompany the remaining
undercoverage. The bias is up to a third of a standard deviation in \(\beta_2\), \(\Omega_{22}\), and \(P_2(x_0)\),
whereas it is negligible for the coefficient ratio. One possible explanation is slow convergence along a direction
involving approximately proportional changes in the coefficients and free Cholesky parameters, with the normalizing
entry held fixed.

Random scaling needs a much larger \(n=T\) for these three targets. Its coverage for \(\beta_2\) rises from \(0.704\)
at \(n=0.32\) million to \(0.904\) at \(n=10.24\) million, and at the largest \(n\) it covers \(0.891\) for
\(\Omega_{22}\) and \(0.913\) for \(P_2(x_0)\). From \(n=1.28\) million on, random-scaling intervals for these targets
are \(10\) to \(21\) percent longer than plug-in intervals. For \(C_{23}\) the two intervals have similar coverage.

\begin{table}[!t]
\centering

\caption{Multi-pass SAUSS: inference at a fixed computational budget, design D1}
\label{tab:mc-inf-multipass}
\spacingset{1}\small
\setlength{\tabcolsep}{4pt}
\begin{tabular*}{\textwidth}{@{\extracolsep{\fill}}lllccccccc@{}}
\toprule
 & & & & & \multicolumn{2}{c}{Coverage} & \multicolumn{3}{c}{Share of variance} \\
\cmidrule(lr){6-7}\cmidrule(lr){8-10}
Target & \(n\) & \(T\) & Bias/SD & SE/SD & SAUSS & Naive & Sampling & Minib. & Simul. \\
\midrule
\(\beta_2/\beta_3\) & 2,000  & 1.28M & 0.03 & 1.00 & 0.952 & 0.952 & 1.000 & 0.000 & 0.000 \\
                    & 20,000 & 1.28M & \(-\)0.03 & 0.98 & 0.939 & 0.939 & 0.997 & 0.002 & 0.001 \\
                    & 0.2M   & 1.28M & \(-\)0.06 & 1.05 & 0.959 & 0.957 & 0.973 & 0.015 & 0.012 \\
                    & 2M     & 1.28M & \(-\)0.02 & 1.01 & 0.958 & 0.922 & 0.781 & 0.122 & 0.097 \\
                    & 2M     & 5.12M & 0.00 & 1.00 & 0.960 & 0.949 & 0.935 & 0.037 & 0.029 \\
\addlinespace
\(P_2(x_0)\)        & 2,000  & 1.28M & 0.01 & 0.99 & 0.952 & 0.952 & 0.998 & 0.000 & 0.001 \\
                    & 20,000 & 1.28M & \(-\)0.03 & 1.01 & 0.958 & 0.955 & 0.986 & 0.002 & 0.012 \\
                    & 0.2M   & 1.28M & \(-\)0.11 & 1.01 & 0.953 & 0.931 & 0.879 & 0.014 & 0.107 \\
                    & 2M     & 1.28M & \(-\)0.09 & 0.98 & 0.938 & 0.787 & 0.422 & 0.066 & 0.513 \\
                    & 2M     & 5.12M & \(-\)0.08 & 1.01 & 0.948 & 0.909 & 0.744 & 0.029 & 0.226 \\
\addlinespace
\(\beta_2\)         & 2,000  & 1.28M & 0.07 & 1.00 & 0.962 & 0.961 & 0.995 & 0.000 & 0.004 \\
                    & 20,000 & 1.28M & 0.08 & 1.01 & 0.960 & 0.953 & 0.959 & 0.001 & 0.039 \\
                    & 0.2M   & 1.28M & 0.12 & 0.98 & 0.942 & 0.895 & 0.705 & 0.011 & 0.284 \\
                    & 2M     & 1.28M & 0.13 & 0.97 & 0.949 & 0.595 & 0.193 & 0.030 & 0.777 \\
                    & 2M     & 5.12M & 0.17 & 1.03 & 0.959 & 0.837 & 0.489 & 0.019 & 0.492 \\
\addlinespace
\(C_{23}\)          & 2,000  & 1.28M & \(-\)0.08 & 1.01 & 0.947 & 0.947 & 0.999 & 0.000 & 0.001 \\
                    & 20,000 & 1.28M & \(-\)0.05 & 0.99 & 0.951 & 0.951 & 0.991 & 0.002 & 0.008 \\
                    & 0.2M   & 1.28M & 0.02 & 1.00 & 0.945 & 0.938 & 0.914 & 0.014 & 0.072 \\
                    & 2M     & 1.28M & 0.00 & 0.96 & 0.938 & 0.828 & 0.514 & 0.080 & 0.405 \\
                    & 2M     & 5.12M & 0.02 & 1.01 & 0.959 & 0.924 & 0.809 & 0.032 & 0.159 \\
\addlinespace
\(\Omega_{22}\)     & 2,000  & 1.28M & 0.05 & 0.99 & 0.951 & 0.950 & 0.997 & 0.000 & 0.003 \\
                    & 20,000 & 1.28M & 0.05 & 1.01 & 0.955 & 0.950 & 0.970 & 0.002 & 0.028 \\
                    & 0.2M   & 1.28M & 0.13 & 1.00 & 0.958 & 0.924 & 0.768 & 0.012 & 0.220 \\
                    & 2M     & 1.28M & 0.09 & 0.96 & 0.945 & 0.645 & 0.249 & 0.039 & 0.712 \\
                    & 2M     & 5.12M & 0.13 & 1.01 & 0.955 & 0.868 & 0.570 & 0.022 & 0.408 \\
\bottomrule
\end{tabular*}
\par
\begin{minipage}{\textwidth}
\vspace{1ex}
\footnotesize
\emph{Notes:} Coverage and shares of variance as in Table~\ref{tab:mc-main-inference} of the
main text; Bias/SD and SE/SD as in Table~\ref{tab:mc-inf-online}. The two \(n=2\)M rows are checkpoints of one
path. Seconds per replication for the pilot, calibration, estimation, and variance:
\(0.0\), \(0.3\), \(20.7\), \(0.3\) at \(n=2{,}000\); \(0.0\), \(0.3\), \(27.5\), \(29.6\) at \(n=0.2\)M; \(0.0\),
\(0.3\), \(30.1\), \(301.6\) at \(n=2\)M and \(T=1.28\)M.
\end{minipage}
\end{table}

\begin{figure}[!htb]
\centering

\caption{Multi-pass SAUSS: coverage of nominal 95\% intervals against \(n\)}
\label{fig:mc-inf-multipass}
\includegraphics[width=\textwidth]{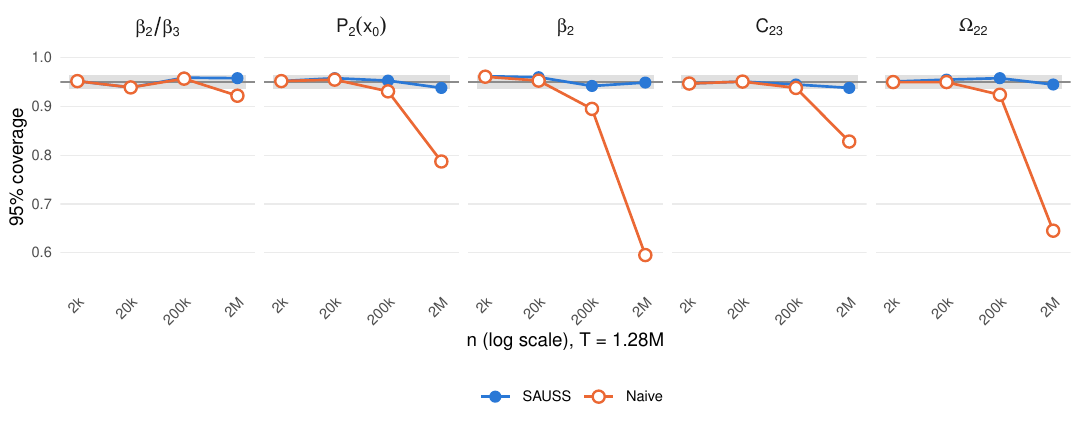}
\par
\begin{minipage}{\textwidth}
\vspace{1ex}
\footnotesize
\emph{Notes:} \(T=1.28\)M; \(n\) on a log scale. Blue (SAUSS): the full plug-in variance
\(\widehat C_{n,T}\); orange (Naive): its sampling component alone. Grey band as in
Figure~\ref{fig:mc-inf-online}.
\end{minipage}
\end{figure}

Table~\ref{tab:mc-inf-multipass} and Figure~\ref{fig:mc-inf-multipass} report all five targets and
sample sizes of the multi-pass experiment. Section~\ref{sec:monte-carlo-results} of the main text discusses the
coverage of the full and naive intervals for \(\beta_2/\beta_3\), \(\beta_2\), and \(\Omega_{22}\). Throughout the
table, the plug-in standard error is within \(5\%\) of the Monte Carlo standard deviation, and the Monte Carlo bias is
at most \(0.17\) of a standard deviation. The two targets that the main text omits behave alike. At \(n=2\) million
and \(T=1.28\) million, the simulation component accounts for \(0.513\) of the variance of \(P_2(x_0)\) and \(0.405\)
of that of \(C_{23}\), and the naive interval covers \(0.787\) and \(0.828\). Lengthening the run to \(T=5.12\)
million lowers these shares to \(0.226\) and \(0.159\) and raises the naive coverage to \(0.909\) and \(0.924\), while
the full interval covers \(0.948\) and \(0.959\).

A second multi-pass experiment chooses the run length from the data, in design D1 with \(n=2{,}000\). On each data
set, SAUSS uses the calibrated tuning of Section~\ref{sec:supp-mc-design}, including \(T_{\mathrm{plan}}\), the
smallest \(T\) at which the pilot-based estimate of the first-order ratio of algorithmic to sampling variance is at most \(5\%\) in every
direction. The median \(T_{\mathrm{plan}}\) is \(162{,}505\) updates, or \(813\) passes. To see how much of this run
length is needed, we also evaluate each path at \(T_{\mathrm{plan}}/16\), \(T_{\mathrm{plan}}/8\),
\(T_{\mathrm{plan}}/4\), and \(T_{\mathrm{plan}}/2\), each time with the full plug-in variance. GHK SML, with standard
errors from the inverse numerical Hessian, returns an estimate in every data set.

Table~\ref{tab:mc-studyB} adds the checkpoints \(T_{\mathrm{plan}}/8\) and \(T_{\mathrm{plan}}/2\) to Panel~B of
Table~\ref{tab:mc-main-inference} of the main text, which discusses the comparison with GHK SML. Most of the accuracy is
reached by \(T_{\mathrm{plan}}/4\), where the root mean squared errors are within about \(1\%\) of their values at
\(T_{\mathrm{plan}}\). Coverage is near the nominal level at every checkpoint, because the plug-in variance accounts
for the larger algorithmic variation of a shorter run. This coverage result does not establish that
the shorter runs meet the planned variance budget. Using variance estimates recomputed at the
final SAUSS estimate at \(T_{\mathrm{plan}}\), the estimated first-order ratio of algorithmic to sampling variance
is at most \(5\%\) in every parameter direction in \(99.6\%\) of the data sets.

\begin{table}[!t]
\centering
\caption{SAUSS with a data-dependent run length and GHK SML, design D1,
\(J=4\), \(n=2{,}000\)}
\label{tab:mc-studyB}
\spacingset{1}\small
\setlength{\tabcolsep}{3pt}
\begin{tabular*}{\textwidth}{@{\extracolsep{\fill}}lrrccccc@{}}
\toprule
 & & & \multicolumn{5}{c}{RMSE \(\times10^3\)} \\
\cmidrule(lr){4-8}
Estimator & Passes & Time (s) & \(\beta_2/\beta_3\) & \(P_2(x_0)\) & \(\beta_2\) & \(C_{23}\) &
\(\Omega_{22}\) \\
\midrule
SAUSS, \(T_{\mathrm{plan}}/16\) & 51  & 0.8 & 31.6 & 19.4 & 74.1 & 104.1 & 189.4 \\
SAUSS, \(T_{\mathrm{plan}}/8\)  & 102 & 0.9 & 31.4 & 18.7 & 68.2 & 98.6 & 180.5 \\
SAUSS, \(T_{\mathrm{plan}}/4\)  & 203 & 1.3 & 31.2 & 18.4 & 65.2 & 96.1 & 177.1 \\
SAUSS, \(T_{\mathrm{plan}}/2\)  & 406 & 1.9 & 31.0 & 18.4 & 64.1 & 95.5 & 175.4 \\
SAUSS, \(T_{\mathrm{plan}}\)    & 813 & 3.3 & 31.0 & 18.4 & 64.4 & 95.2 & 177.5 \\
GHK SML                 & --  & 36.8 & 31.0 & 18.5 & 65.0 & 95.3 & 179.9 \\
\midrule
 & & & \multicolumn{5}{c}{Coverage of nominal 95\% intervals} \\
\cmidrule(lr){4-8}
SAUSS, \(T_{\mathrm{plan}}/16\) & & & 0.947 & 0.951 & 0.961 & 0.957 & 0.936 \\
SAUSS, \(T_{\mathrm{plan}}/8\)  & & & 0.947 & 0.950 & 0.957 & 0.960 & 0.943 \\
SAUSS, \(T_{\mathrm{plan}}/4\)  & & & 0.944 & 0.951 & 0.961 & 0.960 & 0.947 \\
SAUSS, \(T_{\mathrm{plan}}/2\)  & & & 0.946 & 0.950 & 0.966 & 0.951 & 0.948 \\
SAUSS, \(T_{\mathrm{plan}}\)    & & & 0.946 & 0.947 & 0.970 & 0.952 & 0.944 \\
GHK SML                 & & & 0.945 & 0.942 & 0.961 & 0.946 & 0.939 \\
\bottomrule
\end{tabular*}
\par
\begin{minipage}{\textwidth}
\vspace{1ex}
\footnotesize
\emph{Notes:} Each SAUSS path is evaluated at fractions of \(T_{\mathrm{plan}}\). Passes: median
over replications. Time (s): median computation time. SAUSS times include the pilot, calibration,
and plug-in variance estimation at the reported checkpoint; GHK SML time includes the logit start and numerical
Hessian. GHK SML on the same data sets.
\end{minipage}
\end{table}

\subsection{\texorpdfstring{Accept--reject, Botev, and hybrid samplers}{Accept--reject, Botev, and hybrid samplers}}\label{sec:mc-simulator}

This experiment compares capped accept--reject, the exact tilting sampler of \citet{botev2017normal},
and the hybrid sampler when many observed choices are unlikely at the current iterate. Accept--reject needs on
average \(1/P(y_i\mid x_i;\theta)\) proposals per accepted draw, so it becomes slow when the observed choice is
unlikely, and the trial cap \(\Kmax\) may bind. Here \(\Kmax=20{,}000\). If the cap is reached after at least one
acceptance, the available draws are averaged. If no draw is accepted, the observation contributes zero to the update,
which can bias the capped simulated score (Remark~\ref{rem:capped} of the main text). The Botev
sampler avoids this. It draws exactly from the truncated normal distribution of the latent errors, and unlike that of
accept--reject its acceptance probability tends to one in the rare-event limit, but it solves a small optimization
problem for every observation. The hybrid sampler invokes tilting only for observations at which accept--reject has
not produced \(\Rsim\) draws within \(K_0\) proposals.

To create such observations, we estimate design D1 with \(n=2{,}000\) under the fixed tuning of
Section~\ref{sec:supp-mc-design}, but start the recursion at \(\beta=\kappa\beta_0\) and \(\Omega_0\), for
\(\kappa\in\{1,2,4,8\}\). A larger \(\kappa\) mimics a poor start at which many observed choices are nearly
impossible. The three samplers share the data sets and the starting values.

\begin{table}[!t]
\centering

\caption{Estimation time and accuracy by simulator, start \(\beta=\kappa\beta_0\)}
\label{tab:mc-simulator}
\spacingset{1}\small
\begin{tabular*}{\textwidth}{@{\extracolsep{\fill}}cccccccccc@{}}
\toprule
 & \multicolumn{3}{c}{Accept--reject} & \multicolumn{3}{c}{Botev} & \multicolumn{3}{c}{Hybrid} \\
\cmidrule(lr){2-4}\cmidrule(lr){5-7}\cmidrule(lr){8-10}
\(\kappa\) & Time (s) & Failure (\%) & RMSE & Time (s) & Ratio & RMSE &
Time (s) & Ratio & RMSE \\
\midrule
1 & 1.6 & 0.0003 & 0.093 & 14.0 & 8.48 & 0.093 & 1.5 & 0.94 & 0.093 \\
2 & 1.8 & 0.0015 & 0.262 & 13.9 & 7.91 & 0.256 & 1.6 & 0.89 & 0.253 \\
4 & 7.2 & 0.1747 & 1.828 & 13.8 & 1.93 & 1.113 & 1.7 & 0.24 & 1.118 \\
8 & 37.9 & 2.1624 & 9.247 & 12.8 & 0.34 & 3.630 & 1.9 & 0.05 & 3.632 \\
\bottomrule
\end{tabular*}
\par
\begin{minipage}{\textwidth}
\vspace{1ex}
\footnotesize
\emph{Notes:} Time (s): computation time per estimation run. Ratio: time relative to
accept--reject. Failure:
share of observation evaluations with no accepted draw; the Botev and hybrid samplers had none. RMSE:
root mean squared error of \(\widehat\beta\) in levels.
\end{minipage}
\end{table}

The hybrid sampler is the fastest at every \(\kappa\), has no failures, and is as accurate as the
Botev sampler (Table~\ref{tab:mc-simulator}). Accept--reject is cheap near \({\theta_*}\) but degrades quickly with a
poor start. Its failure share rises from almost zero at \(\kappa=1\) to \(2.2\) percent at \(\kappa=8\), and its time
rises more than twentyfold. At \(\kappa\geq4\) the omitted contributions bias the capped score, and its root mean
squared error is \(1.6\) times that of the exact samplers at \(\kappa=4\) and \(2.5\) times at \(\kappa=8\). After
\(80{,}000\) updates from these starts, the exact samplers are still far from \(\beta_0\) but agree with each other.
The time of the Botev sampler is nearly constant in \(\kappa\), but near \({\theta_*}\) it is about eight times slower
than accept--reject. The hybrid sampler is \(6\) percent faster than accept--reject at \(\kappa=1\) and twenty times
faster at \(\kappa=8\). It keeps the low cost of accept--reject where accept--reject works and the reliability of
tilting where it does not.

\section{Additional application material}\label{app:supp-applications}

This appendix complements Section~\ref{sec:application} of the main text. It adds the application
settings that the main text omits and the delta-method standard errors of the GHK SML variance estimates. It also reports
longer SAUSS runs, the computation times of the original \texttt{mlogit} implementation, and the estimation path for
the brand-choice application. Section~\ref{sec:supp-numerical-settings} gives fuller guidance on the numerical choices
of Section~\ref{sec:numerical-settings} of the main text.

Both applications follow the procedure of Section~\ref{sec:application} of the main text. A pass selects each sample
observation once on average. The curvature rule and the variance-based run length use one bank of score replicates,
computed on the full sample at the pilot estimate. We report a single estimation run with one seed. Reported SAUSS
times exclude data loading and the evaluation of the log-likelihood.

The GHK simulation loops of \texttt{mlogit} 2.0-0 \citep{croissant2020mlogit} are written in R. Our \texttt{Rcpp} port
of Section~\ref{sec:application} computes the GHK probabilities, the log-likelihood, and its analytic gradient in C++,
and it follows \texttt{mlogit} 2.0-0 for the optimizer, the starting values, and the variance estimator. With the same
simulation draws, it reproduces the \texttt{mlogit} results in both applications. On the same processor it is
about \(106\) times faster for the labor-supply application with \(2{,}000\) draws and
about \(570\) times faster for the
brand-choice application with \(200\) draws. Logit log-likelihoods use the closed-form choice probabilities. Logit
coefficients and standard errors are divided by \(\pi/\sqrt3\), the standard deviation of an error difference under
the usual logit normalization.

\subsection{\texorpdfstring{Maternal labor-supply intentions}{Maternal labor-supply intentions}}\label{sec:supp-app-bgr}

In the labor-supply application, the curvature rule gives \(\gamma_\beta=0.053\) for the
coefficients and \(\gamma_L=0.106\) for the variance block. Once the scale normalizations are aligned, our GHK SML
estimates are very close to those of \citet{boneva2026}.

We obtain the GHK SML standard errors of \(\Omega_{21}\) and \(\Omega_{22}\) in Table~\ref{tab:bgr} of the main text by
the delta method, from the estimated Cholesky parameters and their estimated covariance matrix. Since
\(\Omega_{21}=L_{21}\) and \(\Omega_{22}=L_{21}^{2}+L_{22}^{2}\), the variance estimate for \(\widehat\Omega_{22}\) is
\[
\widehat{\operatorname{Var}}(\widehat\Omega_{22})
=4\widehat L_{21}^{\,2}s_{21}^{\,2}+4\widehat L_{22}^{\,2}s_{22}^{\,2}
+8\widehat L_{21}\widehat L_{22}\,\widehat{\operatorname{Cov}}(\widehat L_{21},\widehat L_{22}),
\]
where \(s_{21}^{2}\) and \(s_{22}^{2}\) are the estimated variances of \(\widehat L_{21}\) and
\(\widehat L_{22}\), and \(\widehat{\operatorname{Cov}}(\widehat L_{21},\widehat L_{22})\) is their estimated
covariance. With \(\widehat L_{21}=-0.0022\), \(\widehat L_{22}=0.6704\), \(s_{21}=0.0946\), \(s_{22}=0.1029\), and
covariance \(-0.0046\), the standard errors are \(0.0946\) for \(\Omega_{21}\) and \(0.1381\) for \(\Omega_{22}\). The
SAUSS standard errors of \(\Omega_{21}\) and \(\Omega_{22}\) are the analogous delta-method standard errors from the
plug-in variance.

Longer SAUSS runs from the same pilot raise the log-likelihood slightly and move the
family-outcomes coefficient toward the GHK SML estimate of \(1.161\). After \(2{,}000\), \(4{,}455\), and \(17{,}820\)
passes, the log-likelihood is \(-2{,}190.313\), \(-2{,}190.092\), and \(-2{,}189.953\), against \(-2{,}190.972\) for
the reported run. The coefficient is \(1.207\), \(1.196\), and \(1.170\), against \(1.238\).

Computing the SAUSS standard errors from the plug-in variance \(\widehat C_{n,T}\) in \eqref{eq:multipass-covariance}
of the main text takes \(0.5\) seconds. The \(0.78\%\) algorithmic share reported in Section~\ref{sec:app-bgr} is the
share of the sum of the estimated coefficient variances on the standardized covariate scale. In \texttt{mlogit}, GHK SML
with \(2{,}000\) draws takes \(3{,}639.0\) seconds.

\subsection{\texorpdfstring{Canned tuna brand choice}{Canned tuna brand choice}}\label{sec:supp-app-tuna}

In the brand-choice application, the curvature rule gives \(\gamma_\beta=0.112\) and
\(\gamma_L=0.223\). GHK SML uses \(2{,}000\) draws per observation. The logit comparator has the same intercepts and
price as the probit; \citet{kimBlattbergRossi1995} let its coefficients vary across households.

\begin{figure}[!t]
\centering
\caption{Estimation path, canned-tuna brand choice}
\label{fig:paths-tuna}
\includegraphics[width=\textwidth]{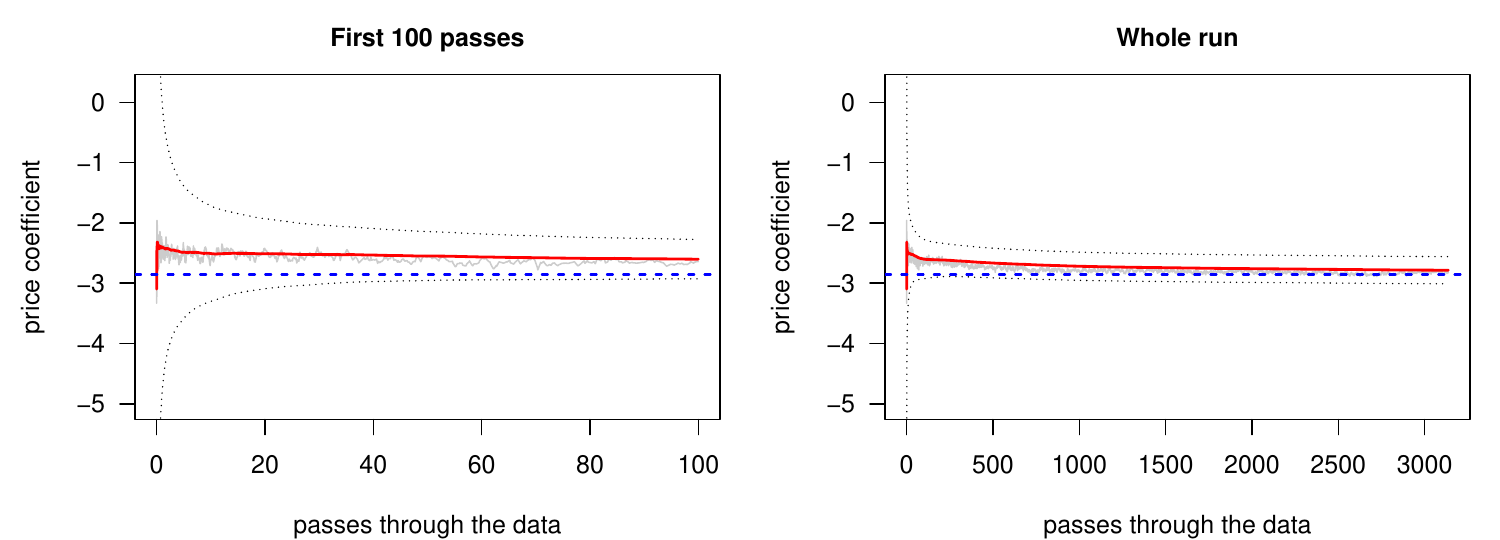}
\par
\begin{minipage}{\textwidth}
\vspace{1ex}
\footnotesize
\emph{Notes:} Price coefficient from one SAUSS run. Left: the first \(100\)
passes through the data; right: the full run of approximately \(3{,}137\) passes. Lines as in
Figure~\ref{fig:paths-bgr} of the main text.
\end{minipage}
\end{figure}

Figure~\ref{fig:paths-tuna} shows the estimation path of the price coefficient. The SAUSS estimate
approaches the GHK SML estimate quickly and stays close to it. The band covers both sampling and algorithmic variation, so
agreement within the band does not by itself show that the algorithmic error is negligible.

A longer run of \(3{,}736\) passes (\(T=5.12\) million updates) from the same pilot gives a price estimate of
\(-2.790\;(0.116)\), against \(-2.784\;(0.115)\) for the reported run and \(-2.855\) for GHK SML. Its log-likelihood is
\(-16{,}273.638\), against \(-16{,}273.729\), and its main run takes \(227.7\)
seconds. With \(200\) instead of
\(2{,}000\) draws, GHK SML gives a price estimate of \(-2.845\;(0.109)\) in \(48.1\) seconds. In
\texttt{mlogit}, GHK SML with
\(200\) draws takes \(27{,}228.4\) seconds.

\subsection{Practical numerical settings}\label{sec:supp-numerical-settings}

The local theory allows fixed \(m\) and \(\Rsim\), any \(a\in(1/2,1)\), and sufficiently small fixed
learning-rate scales under its regularity and initialization conditions. The defaults of
Section~\ref{sec:numerical-settings} of the main text simplify implementation, but the theory does not establish
their optimality.

As in Section~\ref{sec:numerical-settings}, the pilot may use the full sample or a representative random subset, and
its sample size \(n_0\) must grow if the argument relies on the consistency of the pilot estimate. No independent data
split and no restriction on the relative rates of \(n_0\) and \(n\) are needed.

We keep the exponent at \(a=0.501\). In the normalized parameter coordinates, we write the matrix of learning-rate
constants as \(\Gamma=\operatorname{diag}(\gamma_\beta I_\beta,\gamma_L I_L)=\gamma_0\Pi_0\), where \(I_\beta\) is
the identity matrix for the regression coefficients and \(I_L\) the identity matrix for the free Cholesky parameters.
The applications fix \(\Pi_0=\operatorname{diag}(I_\beta,2I_L)\), so that only the overall scale \(\gamma_0\) is
calibrated. For a fixed positive-definite \(\Pi_0\) and a positive-definite pilot curvature estimate
\(\widehat H_0\), we set the overall scale by the rule of thumb
\begin{equation}\label{eq:numerical-settings-curvature}
    \gamma_0=
    \frac{c}{\lambda_{\max}
      (\Pi_0^{1/2}\widehat H_0\Pi_0^{1/2})},
\end{equation}
where \(\lambda_{\max}(A)\) is the largest eigenvalue of a matrix \(A\) and the subscript \(0\)
denotes evaluation at the pilot. To estimate the curvature, we evaluate the score replicates of
Section~\ref{sec:mnp-covariance} of the main text at \(\theta_{\mathrm{init}}\) on \(N\) representative evaluation
observations, with \(\Rsim=5\) and \(S=100\). The evaluation sample may be the pilot sample (\(N=n_0\)) or the full
sample (\(N=n\)). The exact conditional draws are fresh and independent across observations and replicates given the
data and the pilot. Define
\[
\begin{aligned}
    \widehat g_{i0}^{(s)}
    &=\widehat g_{\Rsim}^{(s)}(\theta_{\mathrm{init}};D_i),\\
    \bar g_{i0}&=\frac1S\sum_{s=1}^S\widehat g_{i0}^{(s)},
    \qquad
    \bar g_0=\frac1N\sum_{i=1}^N\bar g_{i0}.
\end{aligned}
\]
The pilot curvature estimate is
\[
\begin{aligned}
    \widehat H_0
    &=\frac1N\sum_{i=1}^N\bar g_{i0}\bar g_{i0}'-\bar g_0\bar g_0'
    -\frac1{NS(S-1)}\sum_{i=1}^N\sum_{s=1}^S
    (\widehat g_{i0}^{(s)}-\bar g_{i0})
    (\widehat g_{i0}^{(s)}-\bar g_{i0})'.
\end{aligned}
\]
This is the simulation-corrected score-variance estimator evaluated at the pilot. Under correct
specification of the MNP model, it estimates \(H_*\) as the pilot estimate becomes consistent and \(N\) increases.
The dimensionless multiplier \(c>0\) sets the learning rate relative to the estimated local curvature, and we assess it
in short preliminary runs. In the applications, \(\widehat H_0\) on the full sample (\(N=n\)) gives
\(\lambda_{\max}(\Pi_0^{1/2}\widehat H_0\Pi_0^{1/2})=9.40\) for the labor-supply application and \(4.48\) for the
brand-choice application, so \(c=0.5\) yields \(\gamma_0=0.053\) and \(\gamma_0=0.112\). We also recommend
standardizing continuous covariates, as in both applications.

The main run holds the chosen scales fixed and restarts the learning-rate index at \(t=1\) with fresh simulation
draws. The pilot only supplies the initialization and is excluded from the reported average. With a consistent pilot
estimate, our theory permits averaging all subsequent iterates, so we set \(T_0=0\).

The run length uses the same score replicates at the pilot. We denote the resulting estimates of the variance
components by \(\widehat\Sigma_{\mathrm{stat},0}\) and \(\widehat\Sigma_{\mathrm{sim},0}\), to distinguish them from
the final post-estimation estimates. For a prespecified finite collection of contrasts of interest, define
\[
    \widehat q_{d,0}=
    \frac{d'\widehat\Sigma_{\mathrm{sim},0}d}
         {d'\widehat\Sigma_{\mathrm{stat},0}d},
    \qquad
    \widehat q_{\max,0}=\max_d\widehat q_{d,0},
\]
where the maximum is over that collection and each denominator is positive.
For a prespecified smooth scalar quantity \(h(\theta_*)\), the
first-order comparison uses \(d=\nabla h(\theta_*)\), approximated by
\(\nabla h(\theta_{\mathrm{init}})\) during calibration.

We recommend checking that the pilot variance components and their ratios are stable under
fresh simulation draws. With a tolerance \(\eta=0.05\), \eqref{eq:numerical-settings-ratio} gives the planned run
length
\begin{equation}\label{eq:numerical-settings-horizon}
    T_{\mathrm{plan}}=
    \left\lceil\frac{n}{10\eta}
       \left(1+\frac{\widehat q_{\max,0}}{5}\right)\right\rceil,
\end{equation}
in which the constants \(10\) and \(5\) are the defaults \(m=10\) and \(\Rsim=5\), and the ceiling
rounds upward. The \(5\%\) budget is relative to the sampling variance, not the total variance. At the boundary the
algorithmic share of the total variance is \(0.05/1.05\), and the standard error exceeds the sampling-only standard
error by the factor \(\sqrt{1.05}\), about \(2.5\%\). These comparisons rest on the first-order variance approximation
and on estimated component ratios. A fixed positive tolerance leaves the algorithmic variance non-negligible relative
to the sampling variance. For fixed \(m\) and \(\Rsim\), the full-sample first-order distribution under \(\sqrt n\)
normalization is recovered as \(mT/n\to\infty\), in the same local \(\epsilon\)-sense as Theorem~\ref{thm:multipass}.

In the applications, we take the maximum over all directions \(d\), which is the largest eigenvalue of
\(\widehat\Sigma_{\mathrm{stat},0}^{-1/2}\widehat\Sigma_{\mathrm{sim},0}\widehat\Sigma_{\mathrm{stat},0}^{-1/2}\),
evaluated on the full sample at the pilot. Rounded to one decimal place, it equals \(158.0\) for the labor-supply
application and \(779.3\) for the brand-choice application. The run lengths \(T=187{,}374\) and \(T=4{,}299{,}792\)
use the unrounded values.

The run length does not control the transient error or local exit (Section~\ref{sec:supp-localization}). As
complementary diagnostics, we examine the mean negative scores and the changes in the averaged estimates relative to
their estimated uncertainty.

The score averages for variance estimation use the same \(\Rsim\) as the estimation itself. A larger \(S\) supplies
more independent replicates of the same score estimator and reduces the Monte Carlo variation in the estimated
variance components, but it changes neither the point estimate nor its first-order variance. Any fixed \(S\geq2\)
suffices for local consistency under the stated MNP conditions as the number of evaluation observations increases.

\printbibliography[title={Supplement References}]
\end{refsection}
\end{document}